\documentclass[a4paper,10pt]{article}

\usepackage{amssymb}
\usepackage{amsmath}
\usepackage{subfigure}

\usepackage{color}

\usepackage{titlesec}

\titleformat{\paragraph}
{\normalfont\normalsize\bfseries}{\theparagraph}{1em}{}
\titlespacing*{\paragraph}
{0pt}{3.25ex plus 1ex minus .2ex}{1.5ex plus .2ex}

\allowdisplaybreaks

\usepackage{graphicx}
\usepackage{amsthm}
\usepackage{latexsym}
\usepackage[dvips]{epsfig}
\usepackage{wasysym}
\usepackage{mathrsfs}
\usepackage{eufrak}
\usepackage{bm}
\usepackage{authblk}
\usepackage{slashed}
\usepackage{yhmath} 
\usepackage{stmaryrd}

\theoremstyle{plain}
\newtheorem{proposition}{Proposition}
\newtheorem{lemma}{Lemma}
\newtheorem{theorem}{Theorem}

\newtheorem*{main}{Main Result}

\newtheorem{remark}{Remark}

\def\bma{{\bm a}}
\def\bmb{{\bm b}}
\def\bmc{{\bm c}}
\def\bmd{{\bm d}}
\def\bme{{\bm e}}
\def\bmf{{\bm f}}
\def\bmg{{\bm g}}
\def\bmh{{\bm h}}
\def\bmi{{\bm i}}
\def\bmj{{\bm j}}
\def\bmk{{\bm k}}
\def\bml{{\bm l}}

\def\bmm{{\bm m}}

\def\bmx{{\bm x}}

\def\bmz{{\bm z}}

\def\bmA{{\bm A}}
\def\bmB{{\bm B}}
\def\bmC{{\bm C}}
\def\bmD{{\bm D}}
\def\bmE{{\bm E}}
\def\bmF{{\bm F}}

\def\bmP{{\bm P}}
\def\bmQ{{\bm Q}}

\def\bmX{{\bm X}}

\def\Schouten{{\bm S}{\bmc}{\bm h}}

\def\bmalpha{{\bm \alpha}}
\def\bmbeta{{\bm \beta}}
\def\bmgamma{{\bm \gamma}}
\def\bmdelta{{\bm \delta}}

\def\bmomega{{\bm \omega}}

\def\bmmu{{\bm \mu}}

\def\bmpi{{\bm \pi}}

\def\bmphi{{\bm \phi}}

\def\bmsigma{{\bm \sigma}}

\def\bmUpsilon{{\bm \Upsilon}}

\def\bmpartial{{\bm \partial}}
\def\bmnabla{{\bm \nabla}}
\def\bmhbar{{\bm \hbar}}

\newcounter{mnotecount}

\newcommand{\mnotex}[1]
{\protect{\stepcounter{mnotecount}}$^{\mbox{\footnotesize $\bullet$\themnotecount}}$ 
\marginpar{
\raggedright\tiny\em
$\!\!\!\!\!\!\,\bullet$\themnotecount: #1} }

\usepackage{titlesec}
\titleformat{\paragraph}
{\normalfont\normalsize\bfseries}{\theparagraph}{1em}{}
\titlespacing*{\paragraph}
{0pt}{3.25ex plus 1ex minus .2ex}{1.5ex plus .2ex}

\newcommand{\dbar}{\mathchar'26\mkern-13mu d}

\newcommand{\phibar}{\mathchar'26\mkern-13mu \phi}

\begin{document}

\title{\textbf{Perturbations of the asymptotic region of the Schwarzschild-de Sitter spacetime}}

\author[,1]{E. Gasper\'in\footnote{E-mail address:{\tt
      e.gasperingarcia@qmul.ac.uk}}} \author[,1]{J. A. Valiente
  Kroon \footnote{E-mail address:{\tt j.a.valiente-kroon@qmul.ac.uk}}}
\affil[1]{School of Mathematical Sciences, Queen Mary, University of
  London, Mile End Road, London E1 4NS, United Kingdom.}

\maketitle

\begin{abstract}
The conformal structure of the Schwarzschild-de Sitter spacetime is
analysed using the extended conformal Einstein field equations. To
this end, initial data for an asymptotic initial value problem for the
Schwarzschild-de Sitter spacetime is obtained. This initial data
allows to understand the singular behaviour of the conformal structure
at the asymptotic points where the horizons of the Schwarzschild-de
Sitter spacetime meet the conformal boundary. Using the insights
gained from the analysis of the Schwarzschild-de Sitter spacetime in a
conformal Gaussian gauge, we consider nonlinear perturbations close to
the Schwarzschild-de Sitter spacetime in the asymptotic region.  We
show that small enough perturbations of asymptotic initial data for
the Schwarzschild de-Sitter spacetime give rise to a solution to the
Einstein field equations which exists to the future and has an
asymptotic structure similar to that of the Schwarzschild-de Sitter
spacetime.

\end{abstract}

\textbf{Keywords:} Conformal methods, spinors, black holes,
Schwarzschild-de Sitter spacetime, global existence.

\medskip
\textbf{PACS:} 04.20.Ex, 04.20.Ha, 04.20.Gz

\setcounter{tocdepth}{2}
\tableofcontents

\section{Introduction}
\label{Introduction}

The stability of black hole spacetimes is, arguably, one of the
outstanding problems in mathematical General Relativity. The challenge
in analysing the stability of black hole spacetimes lies in both the
mathematical problems as well as in the physical concepts to be
grasped. By contrast, the nonlinear stability of  Minkowski spacetime
 ---see e.g. \cite{ChrKla93,Fri86b}--- and de Sitter
spacetimes ---see \cite{Fri86c,Fri86b}--- are well understood.

The results in \cite{Fri86c,Fri86b} show that the so-called
\emph{conformal Einstein field equations} are a powerful tool for the
analysis of the stability and global properties of vacuum
asymptotically simple spacetimes ---see
\cite{Fri81a,Fri86c,Fri86b,Fri14b}.  They provide a system of field
equations for geometric objects defined on a 4-dimensional Lorentzian
manifold $(\mathcal{M},\bmg)$, the so-called \emph{unphysical
  spacetime}, which is conformally related to a spacetime
$(\tilde{\mathcal{M}},\tilde{\bmg})$, the so-called \emph{physical
  spacetime}, satisfying the Einstein field equations. The conformal
framework allows to recast global problems in
$(\tilde{\mathcal{M}},\tilde{\bmg})$ as local problems in
$(\mathcal{M},\bmg)$.  The metrics $\bmg$ and $\tilde{\bmg}$ are
related to each other via a rescaling of the form $\bmg = \Xi^2
\tilde{\bmg}$ where $\Xi$ is a so-called \emph{conformal factor}.
Crucially, the conformal Einstein field equations are regular at the
points where $\Xi=0$ ---the so-called \emph{conformal
  boundary}. Moreover, a solution thereof implies, wherever $\Xi\neq
0$, a solution to the Einstein field equations.

 At its core, the conformal Einstein field equations constitute a
 system of differential conditions on the curvature tensors respect to
 the Levi-Civita connection of $\bmg$ and the conformal factor $\Xi$.
 The original formulation of the equations as given in, say
 \cite{Fri81a,Fri83}, requires the introduction of so-called
 \emph{gauge source functions}. An
 alternative approach to gauge fixing is to adapt the analysis to a
 congruence of curves. In the context of conformal methods, a natural
 candidate for a congruence is given by \emph{conformal geodesics}
 ---see \cite{FriSch87,Fri03c}. To combine gauges based on the
 properties of congruences of conformal geodesics with the conformal
 Einstein field equations, one needs a more general version of the
 latter ---the so-called \emph{extended conformal Einstein field
   equations} \cite{Fri95}.  The extended conformal field equations
 have been used to obtain an alternative proof of the semiglobal
 nonlinear stability of the Minkowski spacetime and of the global
 nonlinear stability of the de-Sitter spacetime ---see
 \cite{LueVal09}. In view of these results, a natural question is
 whether conformal methods can be used in the global analysis of
 spacetimes containing black holes. This article gives a first step in
 this direction by analysing certain aspects of the conformal
 structure of the Schwarzschild-de Sitter spacetime.

\subsection{The  Schwarzschild-de Sitter spacetime}
The Schwarzschild-de Sitter spacetime is a spherically symmetric
solution to the vacuum Einstein field equations with Cosmological
constant. It depends on two parameters: the Cosmological constant
$\lambda$ and the mass parameter $m$.  The assumption of spherical
symmetry almost completely singles out the Schwarzschild-de Sitter
spacetime among the vacuum solutions to the Einstein field equations
with de Sitter-like Cosmological constant. The other admissible
solution is the so-called Nariai spacetime. This observation can be
regarded as a generalisation of Birkhoff's theorem ---see
\cite{Sta98}.  For small values of the areal radius $r$, the solution
behaves like the Schwarzschild spacetime and for large values its
behaviour resembles that of the de Sitter spacetime. In the
Schwarzschild-de Sitter spacetime the relation between the mass and
Cosmological constant determines the location of the
\emph{Cosmological} and \emph{black hole horizons}.

The presence of a Cosmological constant makes the Schwarzschild-de
Sitter solution a convenient candidate for a global analysis by means
of the extended conformal field equations: the solution is an example
of a spacetime which admits a smooth conformal extension towards the
future (respectively, the past) ---see Figures
\ref{fig:SubSdSDiagram}, \ref{fig:eSdSDiagram} and
\ref{fig:HypSdSDiagram} in the main text. This type of spacetimes are called
future (respectively, past) asymptotically de Sitter ---see
  Section \ref{AsymptDeSitterSpaces} for definitions and
   \cite{AndGal02,Gal04} for a more extensive discussion. As the
Cosmological constant takes a de Sitter-like value, the conformal
boundary of the spacetime is spacelike and, moreover, there exists a
conformal representation in which the induced 3-metric on the
conformal boundary $\mathscr{I}$ is homogeneous. Thus, it is possible
to integrate the extended conformal field equations along single
conformal geodesics.

In this article \emph{we analyse the Schwarzschild-de Sitter spacetime
  as a solution to the extended conformal Einstein field equations}
and use the insights thus obtained to discuss nonlinear perturbations
of the spacetime. A natural starting point for this discussion is the
analysis of conformal geodesic equations on the spacetime. The results
of this analysis can, in turn, be used to rewrite the spacetime in the
conformal gauge associated to these curves. However, despite the fact
that the conformal geodesic equations for spherically symmetric
spacetimes can be written in quadratures \cite{Fri03c}, in general,
the integrals involved cannot be solved analytically. In view of this
difficulty, in this article we analyse the conformal properties of the
exact Schwarzschild-de Sitter spacetime by means of an asymptotic
initial value problem for the conformal field equations. Accordingly,
we compute the initial data implied by the Schwarzschild-de Sitter
spacetime on the conformal boundary and then use it to analyse the
behaviour of the conformal evolution equations. An important property
of these evolution equations is that their essential dynamics is
governed by a \emph{core system}. Consequently, an important
aspect of our discussion consists of the analysis of the formation of
singularities in the core system.  This analysis
is irrespective of the relation between $\lambda \neq 0$ and $m$.
This allows us to formulate a result which is valid for the
subextremal, extremal and hyperextremal Schwarzschild-de Sitter
spacetime characterised by the conditions $0<9m^2|\lambda|<1$,
$9m^2|\lambda|=1$ and $9m^2|\lambda|>1$ respectively.

\begin{figure}[t]
\centering
\includegraphics[width=0.3\textwidth]{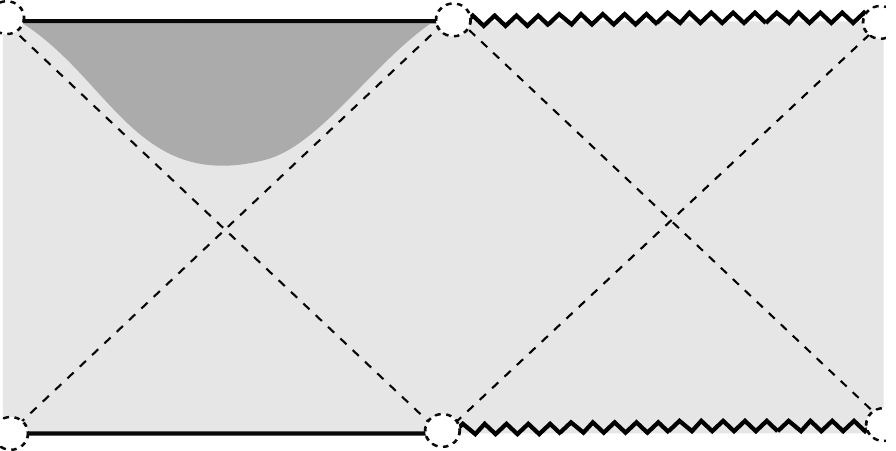}
\includegraphics[width=0.3\textwidth]{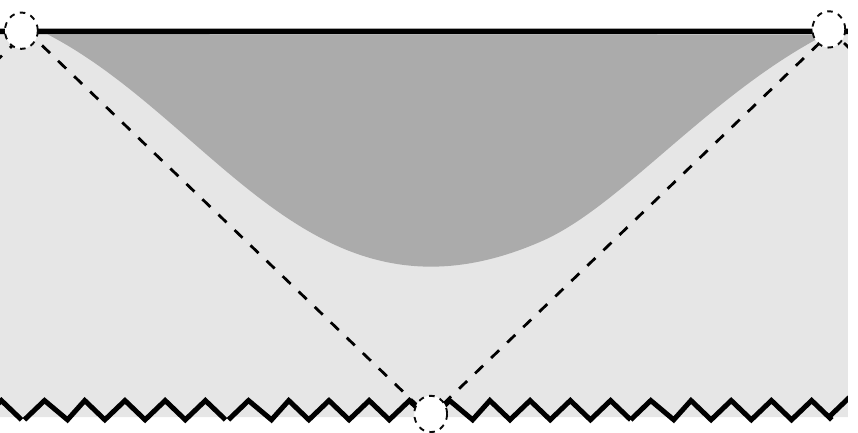}
\includegraphics[width=0.33\textwidth]{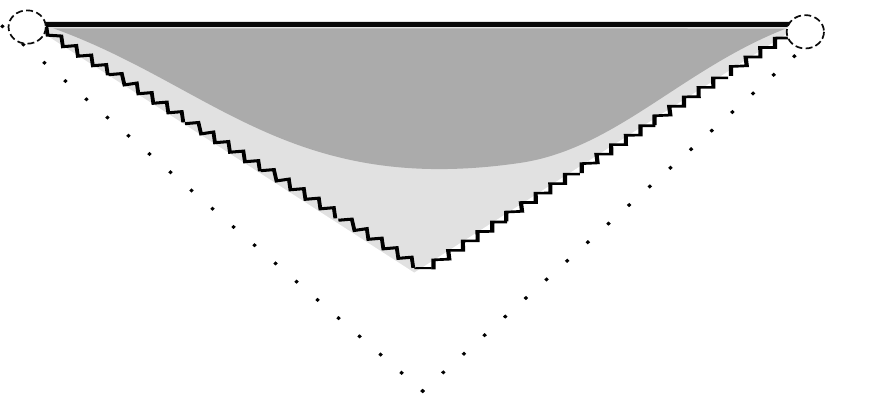}
\put(-390,70){$a)$} \put(-250,70){$b)$} \put(-120,70){$c)$}
\caption{Schematic depiction of the Main Result. Development of
  asymptotic initial data close to that of the Schwarzschild-de Sitter
  spacetime in the global representation ---the initial metric is
  $\bm\hbar$, the standard metric on $\mathbb{S}^3$, and the
  asymptotic points $\mathcal{Q}$ and $\mathcal{Q'}$ are excluded
  (denoted by empty circles in the diagram).  Figures a), b) and c)
  illustrate the evolution of initial data close to the
  Schwarzschild-de Sitter spacetime in the subextremal, extremal and
  hyperextremal cases respectively.  See also Figures
  \ref{fig:SubSdSDiagram},\ref{fig:eSdSDiagram} and
  \ref{fig:HypSdSDiagram}. }
\label{fig:SdSAsymptoticPerturbation}
\end{figure}

\begin{figure}
\includegraphics[width=1\textwidth]{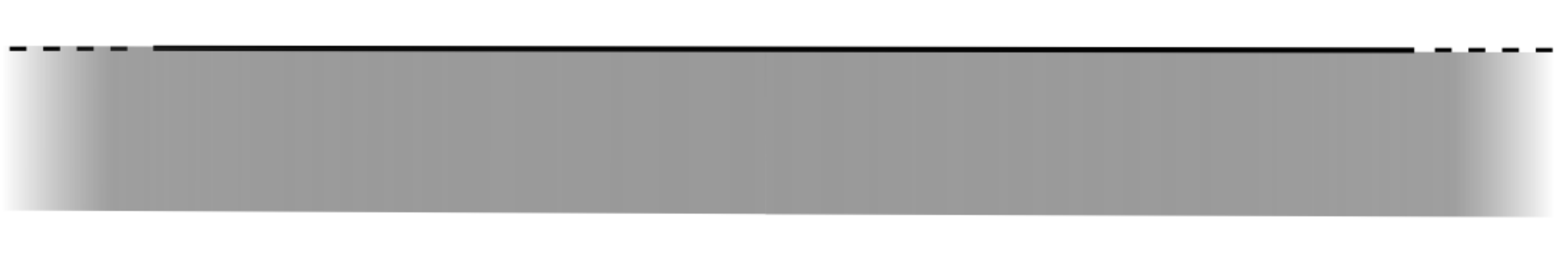}
\put(-420,65){$-\infty$} \put(-15,65){$\infty$}
\caption{Schematic depiction of the Main Result.  Development of
  asymptotic initial data close to that of the Schwarzschild-de Sitter
  spacetime in the representation in which Theorem
  \ref{ExistenceCauchySatbility-eSdS1} was obtained. The initial
  metric is $\bmh$, the standard metric on $\mathbb{R}\times
  \mathbb{S}^2$, and the asymptotic points $\mathcal{Q}$ and
  $\mathcal{Q'}$ are at infinity respect to $\bmh$ ---since $\bmhbar$
  and $\bmh$ are conformally flat one has $\bmh=\omega^2\bmhbar$.  The
  initial data for the subextremal, extremal and hyperextremal cases
  is formally identical.  Small enough perturbations the development
  has the same asymptotic structure as the reference spacetime}
\label{fig:CylinderRepresentationSdS}
\end{figure}

\subsection{The main result}
The analysis of the conformal properties of the 
Schwarzschild-de Sitter spacetime allows us to formulate 
 a result concerning the existence of solutions to the asymptotic
 initial value problem for the Einstein field equations with de
 Sitter-like Cosmological constant which can be regarded as
 perturbations of the \emph{asymptotic region} of the Schwarzschild-de Sitter spacetime ---see Figures \ref{fig:SdSAsymptoticPerturbation} and 
\ref{fig:CylinderRepresentationSdS}.  
 Our existence result  can be stated as:

\begin{main}[\textbf{\em asymptotically de Sitter 
spacetimes close to the asymptotic region of the SdS spacetime}] 
 Given asymptotic initial data which is suitably close to data for the
 Schwarzschild-de Sitter spacetime there exists a solution to the Einstein
 field equations which exists towards the future (past) and has an
 asymptotic structure similar to that of the Schwarzschild-de Sitter
 spacetime ---that is, the solution is future (past) asymptotically de Sitter.
\end{main}

\begin{remark}
{\em A detailed formulation of the Main Result of this
article 
 can be found
in Section \ref{Section:AsymptoticIVP} 
---see  Theorem \ref{ExistenceCauchySatbility-eSdS1}.}
\end{remark}

\medskip
Our analysis of the conformal evolution equations governing
the dynamics of the background solution provides explicit minimal
existence intervals for the solutions. These intervals are certainly
not optimal.  An interesting question related to the
class of solutions to the Einstein field equations obtained in this article is
to obtain their maximal development. To address this problem one requires different methods of
the theory of partial differential equations and it will be discussed
elsewhere.  A schematic depiction of the Main Result 
 is given in Figure \ref{fig:SdSAsymptoticPerturbation}.

\medskip

As part of the analysis of the background solution we require
asymptotic initial data for the Schwarzschild-de Sitter spacetime. The
construction of this initial data allows to study in detail the
singular behaviour of the conformal structure of the family of
background spacetimes at the \emph{asymptotic points}
$\mathcal{Q}$ and $\mathcal{Q}'$, where the
horizons of the spacetime meet the conformal boundary. As a
consequence of the singular behaviour of the asymptotic initial data,
the discussion of the asymptotic initial value problem has to exclude
these points.
 In view
of this, it turns out that a more convenient conformal representation
to analyse the conformal evolution equations for both the exact
Schwarzschild-de Sitter spacetime and its perturbations is one
 in which the the conformal boundary is metrically 
  $\mathbb{R}\times \mathbb{S}^2$ rather than $\mathbb{S}^3
  \backslash\{\mathcal{Q},\mathcal{Q'}\}$ so that the problematic
  asymptotic points are \emph{sent} to infinity ---see Figure
\ref{fig:CylinderRepresentationSdS}.  In
this representation, the methods of the theory of partial differential
equations used to analyse the existence of solutions to the conformal
evolution equations implicitly impose some decay conditions at infinity on the
perturbed initial data.

\subsection{Related results}

The properties of the Schwarzschild-de Sitter spacetime have been
systematically probed by means of an analysis of the solutions of the
scalar wave equation using \emph{vector field methods} ---see
\cite{Sch15}. This type of analysis requires special care when
discussing the behaviour of the solution close to the horizons. In the
asymptotic initial value problem considered in this article, the domain
of influence of the initial data is contained in the region
corresponding to the asymptotic region of the Schwarzschild-de Sitter
spacetime.

\smallskip
The properties of the Nariai spacetime
 ---the other solution appearing in the
generalisation of Birkhoff's theorem to spacetimes with a de
Sitter-like Cosmological constant--- have been analysed by means of
both analytic and numerical methods in \cite{Bey09a,Bey09b}. In
particular, in the former reference it is shown that the Nariai
solution does not admit a smooth conformal extension ---see also \cite{Fri14b}.
 Thus, it cannot be obtained from an asymptotic initial value problem.

\smallskip
Finally, it is pointed out that the singularity of the
 Schwarzschild-de Sitter spacetime is not a conformal
gauge singularity since  $\tilde{C}_{abcd}\tilde{C}^{abcd} \rightarrow \infty$ 
as $r \rightarrow 0$. Accordingly, theory of the extendibility of
conformal  gauge singularities as developed in \cite{LueTod09} cannot
be applied to our analysis. For any
of the possible conformal gauges available, one either has a
singularity of the Weyl tensor arising at a finite value of the
parameter of a conformal geodesic or one has an inextendible conformal
geodesic along which the Weyl tensor is always smooth. 

\subsection*{Notations and conventions}
 The signature convention for (Lorentzian) spacetime metrics is $
 (+,-,-,-)$.  In these conventions the Cosmological constant $\lambda$
 of the de Sitter spacetime takes negative values. Cosmological
 constants with negative values will be said to be \emph{de
   Sitter-like}.

In what follows, the Latin indices ${}_a ,\,{}_b ,{}_c,\,\ldots $ are
used as abstract tensor indices while the boldface Latin indices
${}_\bma,\, {}_\bmb,{}_\bmc, \ldots$ are used as spacetime frame
indices taking the values $0, \ldots , 3$. In this way, given a basis
$\{\bme_{\bma}\}$ a generic tensor is denoted by $T_{ab}$ while its
components in the given basis are denoted by $T_{\bma \bmb}\equiv
T_{ab}\bme_{\bma}{}^{a}\bme_{\bmb}{}^{b}$. We reserve the indices
${}_\bmi,\, {}_\bmj,\, {}_\bmk,\ldots$ to denote frame spatial indices
respect to an adapted frame taking the values $1,\,2,\,3$. We make
systematic use of spinors and follow the conventions and notation of
Penrose \& Rindler \cite{PenRin84} ---in particular,
${}_A,\,{}_B,\,{}_C,\ldots$ are abstract spinorial indices while
${}_\bmA,\, {}_\bmB,\,{}_\bmC,\ldots$ will denote frame spinorial
indices with respect to some specified spin dyad
$\{\epsilon_\bmA{}^{A}\}$.  Our conventions for the curvature tensors
are fixed by the relation:
\[
(\nabla_a \nabla_b -\nabla_b \nabla_a) v^c = R^c{}_{dab} v^d.
\]

 In addition,
 $D^{\pm}(\mathcal{A})$, $H(\mathcal{A})$, $J^{\pm}(\mathcal{A})$ and
 $I^{\pm}(\mathcal{A})$ will denote, respectively, the future (past) domain of dependence, the
 Cauchy horizon, the causal and the chronological futures (pasts) of $\mathcal{A}$ 
  ---see e.g. \cite{HawEll73,Wal84}.

\section{The asymptotic initial value problem in General Relativity}

In this section we briefly revisit the notion of asymptotically
  de Sitter spacetimes ---see \cite{AndGal02,Gal04,HawEll73}. After that, we
  review the properties of the extended conformal Einstein field
equations that will be used in our analysis of the Schwarzschild-de
Sitter spacetime. This general conformal representation of the
Einstein field equations was originally introduced in \cite{Fri95}
---see also \cite{Fri04,Tod02,CFEBook} for further discussion. For completeness, the conformal constraint equations are
presented ---see \cite{Fri83, Fri84,Fri86a,Fri04}. In addition, we
provide a discussion on the notion of conformal geodesics and
conformal Gaussian systems of coordinates ---see
\cite{Sch86,FriSch87,Fri03c,Fri95}. In this section we also discuss
how to use the conformal field equations expressed in terms of a
conformal Gaussian system to set up an \emph{asymptotic initial value
  problem} for a spacetime with a spacelike conformal boundary.
We conclude this section with a  discussion of
 the structural properties of the
  conformal evolution equations in the framework of the theory of
  symmetric hyperbolic systems contained in \cite{Kat75}.

 \subsection{Asymptotically de Sitter spacetimes}
 \label{AsymptDeSitterSpaces}


\medskip
A spacetime $(\tilde{\mathcal{M}},\tilde{\bmg})$ satisfying the vacuum
Einstein field equations
\begin{equation}
\tilde{R}_{ab} = \lambda \tilde{g}_{ab},
\label{EinsteinFieldEquations}
\end{equation}
is \emph{future
  asymptotically de Sitter} if there exist a spacetime with boundary
$(\mathcal{M},\bmg)$, a smooth conformal factor $\Xi$ and
a diffeomorphism
 $\varphi:\tilde{\mathcal{M}}\rightarrow
 \mathcal{U}\subseteq
\mathcal{M}$, such that:
\begin{eqnarray*}
&& \Xi >0 \qquad \text{in} \qquad \mathcal{U}, \\
 && \Xi =0 \quad
  \text{and} \quad \textbf{\mbox{d}}\Xi \neq 0 \quad \text{on} \quad
  \mathscr{I}^{+} \equiv \partial{\mathcal{U}}, \\ && \mathscr{I}^{+} \quad
  \text{is spacelike ---i.e. }\quad \bmg(\textbf{\mbox{d}}\Xi,
\textbf{\mbox{d}}\Xi) >0 \qquad\text{on}
  \qquad \mathscr{I}^{+}, \\ && \mathscr{I}^{+} \quad \text{lies to the
    future of } \tilde{\mathcal{M}}\quad \mbox{---i.e.}\quad \mathscr{I}^{+}\subset I^{+}(\tilde{\mathcal{M}}).
\end{eqnarray*}
Observe that this definition does not restrict the topology of
$\mathscr{I}^{+}$. In particular, it does not have to be compact
---see \cite{Gal04}.  The notion of past asymptotically de Sitter is
defined in analogous way. Additionally, $(\tilde{\mathcal{M}},\tilde{\bmg})$ is
asymptotically de Sitter if it is \emph{future and past asymptotically
  de Sitter}.  Notice that a spacetime which is asymptotically de
Sitter is not necessarily \emph{asymptotically simple} ---see
\cite{HawEll73} for a precise definition of asymptotically simple
spacetime.  In the following, in a slight abuse of notation, the
mapping $\varphi: \tilde{\mathcal{M}}\rightarrow \mathcal{U}\subseteq
\mathcal{M}$ will be omitted in the notation and we write 
\begin{equation}
\label{BasicConformalRescaling}
\bmg=\Xi^2\tilde{\bmg}.
\end{equation}
Furthermore, the term \emph{asymptotic region} will be used to refer
to the set $J^{-}(\mathscr{I}^{+})$ of a future asymptotically de
Sitter spacetime or $J^{+}(\mathscr{I}^{-})$ of a past asymptotically
de Sitter spacetime.

\subsection{The extended conformal Einstein field equations}
\label{sec:XCEFE}

In this section we provide a succinct discussion of the extended
conformal Einstein field equations. 

\subsubsection{Basic notions}
\label{sec:FrameFormalismAndWeylConnections}
 Given any connection $\bmnabla$ over a spacetime manifold
 $\tilde{\mathcal{M}}$, the \emph{torsion} and \emph{Riemann
   curvature} tensors are defined, respectively, by the expressions
 \[
  (\nabla_{a}\nabla_{b} -\nabla_b
 \nabla_a)\phi=\Sigma_{a}{}^{c}{}_{b}\nabla_{c}\phi, \qquad \qquad
 (\nabla_{a}\nabla_{b} -\nabla_b \nabla_a) u^c=R^{c}{}_{dab}u^{d} +
 \Sigma_{a}{}^{d}{}_{b}\nabla_{d}u^{c},
\]
 where $\phi$ and $u^d$ are smooth scalar and vector fields
 respectively, while $\Sigma_{a}{}^{c}{}_{b}$ and $R^{d}{}_{cab}$
 denote the torsion and Riemann tensors of $\bmnabla$.

\subsubsection{Frames and connection coefficients}
Let $\{\bme_{\bma}\}$ denote a set of frame fields on
$\tilde{\mathcal{M}}$ and let $\{ \bmomega^\bma\}$ be the associated
coframe. One has that $\langle \omega^\bma, \bme_\bmb\rangle =
\delta_\bmb{}^\bma$.  We define the frame metric as $g_{\bma
  \bmb}\equiv\bmg(\bme_{\bma},\bme_{\bmb})$ ---in abstract index
notation $g_{\bma \bmb} \equiv
\bme_{\bma}{}^{a}\bme_{\bmb}{}^{b}g_{ab}$. From now on we will restrict
our attention to orthonormal frames, so that $g_{\bma \bmb}=\eta_{\bma
  \bmb}$, where consistent with our signature conventions $\eta_{\bma
  \bmb}=\text{diag}(1,-1,-1,-1)$. The metric $\bmg$ is then expressed in
terms of the coframe $\{\bmomega^\bma \}$ as
\[
\bmg = \eta_{\bma\bmb} \bmomega^\bma \otimes \bmomega^\bmb.
\]

  The connection coefficients
$\Gamma_{\bma}{}^{\bmc}{}_{\bmb}$ of the connection $\bmnabla$ with
respect to the frame $\{ \bme_\bma\}$ are defined via the relation
 \[ 
\nabla_{\bma}\bme_{\bmb}=\Gamma_{\bma}{}^{\bmc}{}_{\bmb}\bme_{\bmc},
\]
  where $\nabla_{\bma}\equiv \bme_{\bma}{}^a\nabla_{a}$ denotes the
  covariant directional derivative in the direction of $\bme_{\bma}$.
  The torsion of $\bmnabla$ can be expressed in terms of the frame
  $\{\bme_{\bma}\}$ and the connection coefficients
  $\Gamma_{\bma}{}^{\bmc}{}_{\bmb}$ via
\[ 
\Sigma_{\bma}{}^{\bmc}{}_{\bmb}\bme_{\bmc} = [\bme_{\bma},\bme_{\bmb}]-
 (\Gamma_{\bma}{}^{\bmc}{}_{\bmb}-\Gamma_{\bmb}{}^{\bmc}{}_{\bma})\bme_{\bmc}.
\]

\subsubsection{Conformal rescalings}

Following the notation introduced in Section \ref{AsymptDeSitterSpaces},
 two spacetimes $(\mathcal{M},\bmg)$ are said to be
 $(\tilde{\mathcal{M}},\tilde{\bmg})$  \emph{conformally
   related} if the metrics $\bmg$ and $\tilde{\bmg}$ satisfy equation \eqref{BasicConformalRescaling} 
for some scalar field $\Xi$.
 In the remainder of this article the
symbols $\bmnabla$ and $\tilde{\bmnabla}$ will be reserved for the
Levi-Civita connection of the metrics $\bmg$ and $\tilde{\bmg}$. The
connection coefficients of $\bmnabla$ and $\tilde{\bmnabla}$ are
related to each other through the expression
\[ 
\Gamma_{\bma}{}^{\bmc}{}_{\bmb}=\tilde{\Gamma}_{\bma}{}^{\bmc}{}_{\bmb}
 + S_{\bma\bmb}{}^{\bmc\bmd}\Upsilon_{\bmd}, 
\]
 where
\[
 S_{\bma\bmb}{}^{\bmc\bmd} \equiv
 \delta_{\bma}{}^{\bmc}\delta_{\bmb}{}^{\bmd} +
 \delta_{\bmb}{}^{\bmc}\delta_{\bma}{}^{\bmd} -
 g_{\bma\bmb}g^{\bmc\bmd} \qquad \text{and}\qquad
 \Upsilon_{\bma}\equiv \Xi^{-1}\nabla_{\bma}\Xi.
\] 
In particular, observe that the 1-form $\bmUpsilon\equiv \Upsilon_\bma
\bmomega^\bma$ is exact.

\subsubsection{Weyl connections}
A \emph{Weyl connection} $\hat{\bmnabla}$ is a torsion-free connection
satisfying the relation
\begin{equation}
\label{WeylConnectionf}
 \hat{\nabla}_{a}g_{bc} =-2f_{a}g_{bc},
\end{equation}
where $f_a$ is an arbitrary 1-form ---thus, $\hat{\bmnabla}$ is not
necessarily a metric connection. Property \eqref{WeylConnectionf} is preserved
under the conformal rescaling \eqref{BasicConformalRescaling} as it
can be verified that 
$\hat{\nabla}_{a}\tilde{g}_{bc} =-2\tilde{f}_{a}\tilde{g}_{bc}$ where
$\tilde{f}_a \equiv f_a + \Upsilon_a$.  The connection coefficients of 
$\hat{\bmnabla}$ are related to those of $\bmnabla$ through the relation
\begin{equation} 
\label{WeylConectionCoefsToLeviCivitaConnectionCoef}
 \hat{\Gamma}_{\bma}{}^{\bmc}{}_{\bmb} =
 \Gamma_{\bma}{}^{\bmc}{}_{\bmb} + S_{\bma \bmb}{}^{\bmc\bmd}f_{\bmd}.
\end{equation}
A Weyl connection is a Levi-Civita connection of
 some element of the conformal class $[\bmg]$ if and only if the 1-form $f_a$
is exact. The Schouten tensor $L_{ab}$ of the connection $\bmnabla$ is
defined as 
\[
L_{ab} \equiv \frac{1}{2}R_{ab}-\frac{1}{12}Rg_{ab}. 
\]
The Schouten tensors of the connections $\hat{\nabla}$ and
 $\nabla$ are related to each other by
\begin{equation} 
\label{WeylSchoutenDefinition}
 L_{a b}-\hat{L}_{a b}=\nabla_{a}f_{b}-\frac{1}{2}S_{a b}{}^{c d}f_{c}f_{ d}
\end{equation}
Notice that, in general, $\hat{L}_{ab} \neq \hat{L}_{(ab)}$. 

\subsubsection{The extended conformal Einstein field equations}

From now on, we will consider Weyl connections
$\hat{\bmnabla}$ related to a conformal metric $\bmg$ as in equation
\eqref{WeylConnectionf}. Let $\hat{P}^{\bma}{}_ {\bmb \bmc \bmd}$ denote the
\emph{geometric curvature} of $\hat{\bmnabla}$ ---that is, the expression
 of the Riemann tensor of 
$\hat{\bmnabla}$ written in terms of derivatives of the connection
coefficients $\hat{\Gamma}_{\bma}{}^{\bmc}{}_{\bmb}$:
\[ 
\hat{P}^{\bma}{}_{\bmb \bmc \bmd} \equiv \bme_{\bma}
(\hat{\Gamma}_{\bmb}{}^{\bmc}{}_{\bmd})-\bme_{\bmb}
(\hat{\Gamma}_{\bma}{}^{\bmc}{}_{\bmd}) + \hat{\Gamma}_{\bmf}{}^{\bmc}{}_{\bmd}
(\hat{\Gamma}_{\bmb}{}^{\bmf}{}_{\bma}-\hat{\Gamma}_{\bma}{}^{\bmf}{}_{\bmb})+
 \hat{\Gamma}_{\bmb}{}^{\bmf}{}_{\bmd}\hat{\Gamma}_{\bma}{}^{\bmc}{}_{\bmf}
-\hat{\Gamma}_{\bma}{}^{\bmf}{}_{\bmd}\hat{\Gamma}_{\bmb}{}^{\bmc}{}_{\bmf}.
\]
 The expression of the irreducible decomposition of  Riemann tensor
 $\hat{R}^{\bma}{}_{\bmb \bmc  \bmd}$  given by
\begin{equation}
 \label{algebraicCurvatureprevious}
\hat{\rho}^{\bma}{}_{\bmb \bmc \bmd}\equiv \Xi d^{\bma}{}_{\bmb \bmc \bmd} +
2 S_{\bmb[\bmc}{}^{\bma \bmf}{\hat{L}}_{\bmb]\bmf}
\end{equation}
 will be called the \emph{algebraic curvature}. In the last
 expression $d^{\bma}{}_{\bmb \bmc \bmd}$ represents the so-called
 \emph{rescaled Weyl tensor}, defined as 
\[
d^{\bma}{}_{\bmb \bmc \bmd}
 \equiv \Xi^{-1} C^{\bma}{}_{\bmb \bmc \bmd},
\]
 where $C^{\bma}{}_{\bmb \bmc \bmd}$ is the Weyl tensor of the metric
 $\bmg$. Despite the fact that the definition of the rescaled
 Weyl tensor may look singular at the conformal boundary, it can be
 shown that under suitable
 assumptions the tensor $d^a{}_{bcd}$ it is regular  even when
 $\Xi=0$. Finally, let us introduce a 1-form $\bmd$ defined by the
 relation 
\[
d_\bma \equiv \Xi f_\bma + \nabla_\bma \Xi.
\]

\medskip
 With the above definitions one can
 write the \emph{vacuum extended conformal Einstein field equations} as
\begin{equation} 
\label{frameXCEFEvanishZeroQuantities} 
\hat{\Sigma}_{\bma}{}^{\bmc}{}_{\bmb}=0 , \qquad 
\hat{\Xi}^{\bma}{}_{\bmb \bmc \bmd}=0, \qquad 
\hat{\Delta}_{\bmc \bmd \bmb}=0, \qquad 
\hat{\Lambda}_{\bmb\bmc\bmd} =0
\end{equation}
where
\begin{subequations}
\begin{eqnarray}
&& \hat{\Sigma}_{\bma}{}^{\bmc}{}_{\bmb} \equiv [\bme_{\bma},\bme_{\bmb}]-
(\hat{\Gamma}_{\bma}{}^{\bmc}{}_{\bmb}-\hat{\Gamma}_{\bmb}{}^{\bmc}{}_{\bma})
\bme_{\bmc},\label{frameXCEFE1} \\
&& \hat{\Xi}^{\bma}{}_{\bmb \bmc \bmd} \equiv  \hat{P}^{\bma}{}_{\bmb \bmc \bmd} 
-\hat{R}^{\bma}{}_{\bmb \bmc \bmd},\label{frameXCEFE2} \\
&& \hat{\Delta}_{\bmc \bmd \bmb}\equiv \hat{\nabla}_{\bmc}\hat{L}_{\bmb\bmd}
-\hat{\nabla}_{\bmd}\hat{L}_{\bmc \bmb}-d_{\bma}d^{\bma}{}_{\bmb \bmc \bmd} =0,
\label{frameXCEFE3}\\
&& \hat{\Lambda}_{\bmb\bmc\bmd} \equiv \hat{\nabla}_{\bma}d^{\bma}{}_{\bmb \bmc \bmd}
-f_{\bma}d^{\bma}{}_{\bmb \bmc \bmd} .\label{frameXCEFE4}
\end{eqnarray}
\end{subequations}
The fields $\hat{\Sigma}_{\bma}{}^{\bmc}{}_{\bmb}$,
$\hat{\Xi}^{\bma}{}_{\bmb \bmc \bmd}$, $\hat{\Delta}$ and
$\hat{\Lambda}_{\bmb\bmc\bmd}$ will be called
\emph{zero-quantities}. 

\medskip
The geometric meaning of the extended conformal field equations is the
following: $\hat{\Sigma}_{\bma}{}^{\bmc}{}_{\bmb}=0$ describes the
fact that the connection $\hat{\nabla}$ is torsion-free. The equation
$\hat{\Xi}^{\bma}{}_{\bmb \bmc \bmd}=0$ expresses the fact that the
algebraic and geometric curvature coincide. The equations
$\hat{\Delta}_{\bmc \bmd \bmb} =0$ and $\hat{\Lambda}_{\bmb\bmc\bmd}
=0$ encode the contracted second Bianchi identity. Observe that there
is no differential condition for neither the 1-form $\bmd$ nor the
conformal factor. In Section
\ref{ConformalGeodesicsAndGaussianSystems} it will be discussed how to
fix these fields by gauge conditions.

In order to relate the conformal equations
\eqref{frameXCEFEvanishZeroQuantities} to the vacuum Einstein field
equations \eqref{EinsteinFieldEquations} one introduces the constraints
\begin{equation}
\label{XCEFEconstraints}
\delta_{\bma}=0, \qquad \gamma_{\bma \bmb}=0, \qquad \zeta_{\bma
  \bmb}=0
\end{equation}
encoded in the \emph{supplementary zero-quantities}
\begin{subequations}
\begin{eqnarray*}
&& \delta_{\bma} \equiv d_{\bma}-\Xi
  f_{\bma}-\hat{\nabla}\Xi \label{XCEFEconst1}, \\ && \gamma_{\bma
    \bmb}\equiv \frac{1}{6}\lambda \Xi^{-2}\eta_{\bma
    \bmb}-\hat{\nabla}_{\bma}(\Xi^{-1}d_{\bmb})-\Xi^{-2}S_{\bma
    \bmb}{}^{\bmc \bmd}d_{\bmc}d_{\bmd}, \label{XCEFEconst2} \\ &&
  \zeta_{\bma \bmb} \equiv \hat{L}_{[\bma
      \bmb]}-\hat{\nabla}_{\bma}f_{\bmb} \label{XCEFEconst3}.
 \end{eqnarray*}
\end{subequations} 
 The first equation in \eqref{XCEFEconstraints} encodes the definition
of the 1-form $d_{\bma}$; the second equation in
\eqref{XCEFEconstraints} arises from the transformation law between
the Schouten tensor $\hat{L}_{\bma\bmb}$ of $\hat{\bmnabla}$ and the
physical Schouten tensor
$\tilde{L}_{\bma\bmb}=\frac{1}{6}\tilde{\eta}_{\bma\bmb}$ determined
by the Einstein field equations \eqref{EinsteinFieldEquations};  the last
equation in \eqref{XCEFEconstraints} relates the
antisymmetry of the Schouten tensor $\hat{L}_{ab}$ to the derivative
of the 1-form $f_a$. 

\medskip
The precise relation between the extended conformal Einstein field
equations and the Einstein field equations is given by the following
lemma:

\begin{lemma}
\label{lemma:XCEFE-EFE}
Let
$(\bme_{\bma}{}^{a},\hat{\Gamma}_{\bma}{}^{\bmb}{}_{\bmc},\hat{L}_{\bma
  \bmb},d^{\bma}{}_{\bmb \bmc \bmd})$ denote a solution to the vacuum
extended conformal Einstein field equations \eqref{frameXCEFEvanishZeroQuantities}
 for some choice of gauge fields $(\Xi,d_{\bma})$ satisfying the
 constraint equations
\eqref{XCEFEconstraints}.  Assume, further,
that 
\[
\Xi\neq 0 \qquad  \text{and} \qquad 
det(\eta^{\bma \bmb}\bme_{\bma}\bme_{\bmb}) \neq 0
\]
on an open subset $ \mathcal{U}\subset \tilde{\mathcal{M}}$. Then
\[
\tilde{\bmg}= \Theta^{-2}\eta_{\bma \bmb}\bmomega^{\bma}\otimes
\bmomega^{\bmb},
\]
 where $\{\bmomega^{\bma}\}$ is the coframe dual to
$\{\bme_{\bma}\}$, is a solution to the vacuum Einstein field equations
\eqref{EinsteinFieldEquations} on
$\mathcal{U}$.
\end{lemma}
 The proof of this lemma can be found in  \cite{Fri03c,CFEBook}.

\subsubsection{Conformal geodesics and conformal Gaussian systems}
\label{ConformalGeodesicsAndGaussianSystems}

A conformal geodesic on a spacetime
$(\tilde{\mathcal{M}},\tilde{\bmg})$ consists of a pair
$(x(\tau),\bmbeta(\tau))$ where $x(\tau)$ is a curve with tangent
$\dot{\bmx}(\tau)$ and $\bm\beta(\tau)$ is a 1-form defined along
$x(\tau)$ satisfying the \emph{conformal geodesic equations}
\begin{subequations}
\begin{eqnarray}
&&\dot{x}^{c}\tilde{\nabla}_{c}\dot{x}^{a} = 
  -\dot{x}^{d}\dot{x}^{b}S_{db}{}^{af}\beta_{f},\label{tildeConformalGeodesicEquation1}
  \\ &&\dot{x}^c\tilde{\nabla}_{c}\beta_{a} = -\tfrac{1}{2}\dot{x}^{c}
  S_{c a}{}^{b d}\beta_{b}\beta_{d} + \tilde{L}_{ca}\dot{x}^{c},\label{tildeConformalGeodesicEquation2}
\end{eqnarray}
\end{subequations}
where $\tilde{L}_{ca}$ denotes the Schouten tensor of
$\tilde{\nabla}$ and
\[
S_{ab}{}^{cd} \equiv \delta_a{}^c \delta_b{}^d + \delta_a{}^d
\delta_b{}^c - \tilde{g}_{ab} \tilde{g}^{cd}.
\]
In addition, it is convenient to consider a \emph{Weyl
propagated frame} ---that is,  a frame field $\{\bme_{\bma}{}^a\}$ satisfying
\[ 
\dot{x}^c\tilde{\nabla}_{c}\bme_{\bma}{}^a = -S_{cd}{}^{af}
\bme_{\bma}{}^d \dot{x}^c\beta_{f}. 
\]

 The definition of conformal geodesics is motivated by the
transformation laws of equations
\eqref{tildeConformalGeodesicEquation1}-\eqref{tildeConformalGeodesicEquation2}
under conformal rescalings and transitions to Weyl
connections. More precisely, given an arbitrary 1-form
$\wideparen{\bmf}$ one can construct a Weyl connection
$\wideparen{\nabla}$ as in equation
\eqref{WeylConnectionf}. Then, defining
$\wideparen{\bmbeta}\equiv \bmbeta-\wideparen{\bmf}$ the pair
$(x(\tau),\wideparen{\bmbeta}(\tau))$ will satisfy the equations
\begin{eqnarray*}
&&\dot{x}^{c}\wideparen{\nabla}_{c}\dot{x}^{a} = 
  -\dot{x}^{d}\dot{x}^{b}S_{db}{}^{af}\wideparen{\beta}_{f},
  \\ &&\dot{x}^c\wideparen{\nabla}_{c}\wideparen{\beta}_{a} = -\tfrac{1}{2}\dot{x}^{c}
  S_{c a}{}^{b d}\wideparen{\beta}_{b}\wideparen{\beta}_{d} +
  \wideparen{L}_{ca}\dot{x}^{c},
\end{eqnarray*}
 where $\wideparen{L}_{ca}$ is the Schouten tensor of the connection
$\wideparen{\bmnabla}$ as defined in equation
\eqref{WeylSchoutenDefinition}. If one chooses a Weyl
connection $\hat{\bmnabla}$ whose defining 1-form $\bmf$ coincides with
the 1-form $\bmbeta$ of the $\hat{\bmnabla}$-conformal geodesic equations
\eqref{tildeConformalGeodesicEquation1}-\eqref{tildeConformalGeodesicEquation2},
then the conformal geodesic equations reduce to
\begin{equation} 
\label{protoGauge}
  \dot{x}^{c}\hat{\nabla}_{c}\dot{x}^{a}=0, \qquad
  \hat{L}_{ab}\dot{x}^{b}=0.
\end{equation}
 Similarly, the Weyl propagation of the frame becomes
\begin{equation}\label{weylPropagationFrame}
\dot{x}^{c}\hat{\nabla}_{c}\bme_{\bma}{}^{a}=0.
\end{equation}

\medskip
 The conformal geodesics
equations admit more general reparametrisations than the usual affine
parametrisation of metric geodesics.  This is summarised in the
following lemma:
\begin{lemma}
\label{lemmaCGreparametrisations}
The admissible reparametrisations mapping (non-null) conformal
geodesics into (non-null) conformal geodesics are given by fractional
transformations of the form
\[ 
\tau \mapsto \frac{a \tau+b}{c \tau + d}
\]
where $a,b,c,d \in \mathbb{R}$.
\end{lemma}
 The proof of this lemma can be found in \cite{Fri03c} ---see also
 \cite{Tod02,CFEBook}. Conformal geodesics allow to single out a
 canonical representative of
 the conformal class $[\tilde{\bmg}]$. This observation is contained
 in the following key result:

\begin{lemma}
\label{LemmaCF}
Let $(\tilde{\mathcal{M}},\tilde{\bmg})$ be a spacetime where
$\tilde{\bmg}$ is a solution to the vacuum Einstein field equations
\eqref{EinsteinFieldEquations}. Moreover, let
$(x(\tau),\bmbeta(\tau))$ satisfy the conformal geodesic equations
\eqref{tildeConformalGeodesicEquation1}-\eqref{tildeConformalGeodesicEquation2}
and let $\{\bme_{\bma}\}$ denote a Weyl propagated $\bmg$-orthonormal
frame along $x(\tau)$ with
\[
\bmg\equiv\Theta^2\tilde{\bmg},
\]
 such that 
\[
\bmg(\dot{\bmx},\dot{\bmx})=1.
\]
Then the conformal factor $\Theta$ is given, along $x(\tau)$, by
\begin{equation}
\label{canonicalCF}
\Theta(\tau)= \Theta_{\star} + \dot{\Theta}_{\star}(\tau-\tau_{\star})
+ \frac{1}{2}\ddot{\Theta}_{\star}(\tau-\tau_{\star})^2
\end{equation}
where the coefficients $\Theta_{\star}\equiv \Theta(\tau_{\star}),
\dot{\Theta}_{\star}\equiv \dot{\Theta}(\tau_{\star})$ and $
\ddot{\Theta}_{\star}\equiv \ddot{\Theta}(\tau_{\star})$ are constant
along the conformal geodesic and satisfy the constraints
\begin{equation}
\dot{\Theta}_{\star}= \langle \bmbeta_{\star},\dot{\bmx}_{\star}\rangle
 \Theta_{\star},\qquad \Theta_{\star}\ddot{\Theta}_{\star}=
 \frac{1}{2}\tilde{\bmg}^{\sharp}(\bmbeta_{\star},\bmbeta_{\star})+
 \frac{1}{6}\lambda.
\label{LemmaCFConstraint}
\end{equation}
 Moreover, along each conformal geodesic
\[ 
\Theta \beta_{\bm0} = \dot{\Theta}, \qquad  \qquad \Theta
\beta_{\bmi}=\Theta_{\star} \beta_{\bmi\star},
\]
where $\beta_{\bma}\equiv \langle \bmbeta,\bme_{\bma}\rangle$. 
\end{lemma}
 The proof of this Lemma and a further discussion of the
properties of conformal geodesics can be found in \cite{Fri95, CFEBook}.

\medskip
For spacetimes with a spacelike conformal boundary the relation between 
metric geodesics and conformal geodesics is particularly simple.
This observation is the content of the following: 

\begin{lemma}\label{ReparametrisationMetricToConformalLeavingScri}
Any conformal geodesic leaving $\mathscr{I}^{+}$ ($\mathscr{I}^{-}$)
orthogonally into the past (future) is up to reparametrisation a
timelike future (past) complete geodesic for the physical metric
$\tilde{\bmg}$. The reparametrisation required is determined by
\begin{equation}
\label{ReparametrisationConformalGeodesics}
\frac{\mbox{\em d}\tilde{\tau}}{\mbox{\em d}\tau}= \frac{1}{\Theta(\tau)}
\end{equation}
where $\tilde{\tau}$ is the $\tilde{\bmg}$-proper time and $\tau$ is
the $\bmg$-proper time and  $\bmg=\Theta^2\tilde{\bmg}$.
\end{lemma}
The proof of this Lemma can be found in \cite{FriSch87}.  

\subsubsection{Conformal Gaussian systems}
\label{ConformalGaussSystems}

In what follows it will be assumed that there is a region of the
spacetime $(\tilde{\mathcal{M}},\tilde{\bmg})$ which can be covered by
non-intersecting conformal geodesics emanating orthogonally from some
initial hypersurface $\tilde{\mathcal{S}}$. Using Lemma \ref{LemmaCF},
the conformal factor \eqref{canonicalCF} is \emph{a priori} known and
completely determined from the specification of
$\Theta_{\star}$, $\dot{\Theta}_{\star}$ and $\ddot{\Theta}_{\star}$
on $\tilde{\mathcal{S}}$.  A \emph{conformal Gaussian system} is then
constructed by adapting the time leg of the $\bmg$-orthonormal tetrad
$\{\bme_{\bma}\}$ to the tangent to the conformal geodesic
$(x(\tau),\bmbeta(\tau))$ ---i.e. one sets $\bme_{0}= \dot{\bmx}$. The
rest of the tetrad is then assumed to be Weyl propagated along the
conformal geodesic. If one writes this condition together with the
conformal geodesic equations expressed in terms of the Weyl connection
singled out by $\bmbeta$, as in equations \eqref{protoGauge} and
\eqref{weylPropagationFrame}, one obtains the \emph{gauge conditions}
\begin{equation}
\label{conformalGaugeConditionsFrameVersion}
\hat{\Gamma}_{\bm0}{}^{\bma}{}_{\bmb}=0, \qquad  \hat{L}_{\bm0
  \bma}=0, \qquad  f_{\bm0}=0.
\end{equation}
One can further specialise the gauge by using the parameter $\tau$ along
the conformal geodesics as a time coordinate so that
\begin{equation} 
\label{GaugeAdaptTimelikeLeg}
 \bme_{\bm0}=\bm\partial_{\tau}.
\end{equation}

Now, consider a system of coordinates $(\tau,x^{\bm\alpha})$ where $(x^{\bm\alpha})$
are some local coordinates on $\tilde{\mathcal{S}}$. The coordinates $(x^{\bm\alpha})$
are extended off the initial hypersurface $\tilde{\mathcal{S}}$ by requiring
them to remain constant along the conformal geodesic which intersects
a point $p \in \tilde{\mathcal{S}}$ with coordinates $(x^{\bm\alpha})$. This type of
coordinates will be called a \emph{conformal Gaussian coordinate system}.
This construction naturally leads to consider a 1+3 decomposition of
the field equations. 

\subsubsection{Spinorial extended conformal Einstein field equations}

A spinorial version of the extended conformal Einstein field equations
\eqref{frameXCEFE1}-\eqref{frameXCEFE4} is readily obtained by
suitable contraction with the \emph{ Infeld-van der Waerden symbols}
$\sigma^{\bma}{}_{\bmA \bmA'}$. Given the components
$T_{\bma\bmb}{}^\bmc$ of a tensor $T_{ab}{}^c$, the components of its spinorial
counterpart are given by
\[
T_{\bmA \bmA'\bmB
  \bmB'}{}^{\bmC \bmC'}\equiv T_{\bma \bmb}{}^{\bmc}\sigma_{\bmA \bmA'
}{}^{\bma} \sigma_{\bmB \bmB' }{}^{\bmb} \sigma^{\bmC \bmC'
}{}_{\bmc},
\]
 where,
\begin{eqnarray*}
\sigma_{\bmA \bmA'}{}^{\bm0} &\equiv
  \displaystyle\frac{1}{\sqrt{2}} \begin{pmatrix} 1 & 0 \\ 0 & 1
    \\ \end{pmatrix}, 
\qquad  
\sigma_{\bmA \bmA'}{}^{\bm1}
  &\equiv \frac{1}{\sqrt{2}} \begin{pmatrix} 0 & 1 \\ 1 & 0
    \\ \end{pmatrix},  \\  
\sigma_{\bmA \bmA'}{}^{\bm2}
  &\equiv \displaystyle\frac{1}{\sqrt{2}} \begin{pmatrix} 0 & -\mbox{i} \\ \mbox{i}
    & 0 \\ \end{pmatrix},  
\qquad  
\sigma_{\bmA \bmA'}{}^{\bm3}
 &\equiv \frac{1}{\sqrt{2}} \begin{pmatrix} 1 & 0 \\ 0 & -1
    \\ \end{pmatrix},   \\
  \sigma^{\bmA \bmA'}{}_{\bm0} &\equiv
  \displaystyle\frac{1}{\sqrt{2}} \begin{pmatrix} 1 & 0 \\ 0 & 1
    \\ \end{pmatrix}, 
\qquad 
\sigma^{\bmA \bmA'}{}_{\bm1}
  &\equiv \frac{1}{\sqrt{2}} \begin{pmatrix} 0 & 1 \\ 1 & 0
    \\ \end{pmatrix},  \\ \sigma^{\bmA \bmA'}{}_{\bm2}
  &\equiv \displaystyle\frac{1}{\sqrt{2}} \begin{pmatrix} 0 & \mbox{i} \\ -\mbox{i}
    & 0 \\ \end{pmatrix}, 
\qquad 
\sigma^{\bmA \bmA'}{}_{\bm3}
  &\equiv \frac{1}{\sqrt{2}} \begin{pmatrix} 1 & 0 \\ 0 & -1
    \\ \end{pmatrix}.
\end{eqnarray*}
In particular, the spinorial counterpart of the frame metric $g_{\bma
  \bmb}=\eta_{\bma \bmb}$ is given by $g_{\bmA \bmA' \bmB \bmB'}
\equiv \epsilon_{\bmA \bmB}\epsilon_{\bmA' \bmB'}$ while the frame
$\{\bme_\bma\}$ and coframe $\{ \bmomega^\bma \}$ imply a frame $\{
\bme_{\bmA\bmA'} \}$ and a coframe $\{ \bmomega^{\bmA\bmA}\}$ such
that
\[
\bmg (\bme_{\bmA\bmA'},\bme_{\bmB\bmB'}) =\epsilon_{\bmA\bmB}
\epsilon_{\bmA'\bmB'}. 
\]

\medskip
 If one denotes with the same kernel letter the unknowns of the
frame version of the extended conformal Einstein field equations one
is lead to consider the following spinorial zero-quantities:
\begin{subequations}
\begin{eqnarray}
&& \hat{\Sigma}_{\bmA \bmA'\bmB \bmB'} \equiv [\bme_{\bmA \bmA'},\bme_{\bmB \bmB'}]
-(\hat{\Gamma}_{\bmA \bmA'}{}^{\bmC \bmC'}{}_{\bmB \bmB'}
-\hat{\Gamma}_{\bmB \bmB'}{}^{\bmC \bmC'}{}_{\bmA \bmA'})\bme_{\bmC \bmC'},\label{spacetimeSpinorXCEFE1} \\
&& \hat{\Xi}^{\bmC \bmC'}{}_{\bmD \bmD' \bmA \bmA'\bmB \bmB'} \equiv 
 \hat{r}^{\bmC \bmC'}{}_{\bmD \bmD' \bmA \bmA' \bmB \bmB'}
 -\hat{R}^{\bmC \bmC'}{}_{\bmD \bmD' \bmA \bmA' \bmB \bmB'},\label{spacetimeSpinorXCEFE2} \\
&& \hat{\Delta}_{\bmC \bmC' \bmD \bmD' \bmB \bmB'}\equiv 
\hat{\nabla}_{\bmC \bmC'}\hat{L}_{\bmD \bmD' \bmB \bmB'} 
-\hat{\nabla}_{\bmD \bmD'}\hat{L}_{\bmC \bmC' \bmB \bmB'}
-d_{\bmA \bmA'}d^{\bmA \bmA'}{}_{\bmB \bmB' \bmD \bmD'},\label{spacetimeSpinorXCEFE3}\\
&& \hat{\Lambda}_{\bmB \bmB'\bmC \bmC'\bmD \bmD'} \equiv 
\hat{\nabla}_{\bmA \bmA'}d^{\bmA \bmA'}{}_{\bmB \bmB' \bmC \bmC' \bmD \bmD'}
-f_{\bmA \bmA'}d^{\bmA \bmA'}{}_{\bmB \bmB' \bmC \bmC' \bmD \bmD'} \label{spacetimeSpinorXCEFE4}.
\end{eqnarray}
\end{subequations}

The spinorial version of the extended conformal Einstein
 field equations are then succinctly written as
\begin{equation}
\label{uncontractedXCEFE}
 \hat{\Sigma}_{\bmA \bmA'\bmB \bmB'} =0, \qquad
 \hat{\Xi}^{\bmC \bmC'}{}_{\bmD \bmD' \bmA \bmA'\bmB \bmB'}
 =0, \qquad \hat{\Delta}_{\bmC \bmC' \bmD \bmD' \bmB
   \bmB'}=0, \qquad \hat{\Lambda}_{\bmB \bmB'\bmC \bmC'\bmD
   \bmD'} =0.
\end{equation}
In the spinor description one can exploit the symmetries of
the fields and equations to obtain expressions in terms of lower
valence spinors. In particular, one has the decompositions
\begin{eqnarray*}
&& d_{\bmA \bmA' \bmB \bmB' \bmC \bmC' \bmD \bmD'} = -\phi_{\bmA \bmB
  \bmC \bmD} \epsilon_{\bmA' \bmB'}\epsilon_{\bmC' \bmD'} -
\bar{\phi}_{\bmA' \bmB' \bmC' \bmD'}\epsilon_{\bmA \bmB}\epsilon_{\bmC
  \bmD}, \\
&& \hat{\Gamma}_{\bmA \bmA'}{}^{\bmB \bmB'}{}_{\bmC \bmC'} =
\hat{\Gamma}_{\bmA \bmA'}{}^{\bmB}{}_{\bmC}\epsilon_{\bmC'}{}^{\bmB'}
+ \bar{\hat{\Gamma}}_{\bmA
  \bmA'}{}^{\bmB'}{}_{\bmC'}\epsilon_{\bmC}{}^{\bmB},
\end{eqnarray*}
 where $\phi_{\bmA \bmB \bmC \bmD}= \phi_{(\bmA \bmB \bmC
  \bmD)}$ are the components of the rescaled Weyl spinor and
$\hat{\Gamma}_{\bmA \bmA'}{}^{\bmB}{}_{\bmC}$ are the reduced
connection coefficients of $\hat{\bmnabla}$. Using the
spinorial version of equation
\eqref{WeylConectionCoefsToLeviCivitaConnectionCoef} and contracting
appropriately one obtains
\begin{equation} 
\label{SpinorWeylConnectionToLeviCivitaConnection}
\hat{\Gamma}_{\bmC \bmC' \bmA \bmB}=\Gamma_{\bmC \bmC' \bmA
  \bmB}-\epsilon_{\bmA \bmC}f_{\bmB \bmC'}.
\end{equation}
Likewise, one has the following reduced curvature spinors
\begin{eqnarray*}
 && \hat{R}^{\bmC}{}_{\bmD \bmA \bmA' \bmB \bmB'} \equiv
  \frac{1}{2}\hat{R}^{\bmC\bmQ'}{}_{\bmD \bmQ' \bmA \bmA' \bmB \bmB'}\\
  && \phantom{\hat{R}^{\bmC}{}_{\bmD \bmA \bmA' \bmB \bmB'}}=
 \bme_{\bmA \bmA'}\left( \hat{\Gamma}_{\bmB
    \bmB'}{}^{\bmC}{}_{\bmD} \right) -\bme_{\bmB \bmB'}
  \left(\hat{\Gamma}_{\bmA \bmA'}{}^{\bmC}{}_{\bmD} \right) \nonumber
  \\ && \hspace{3cm}-\hat{\Gamma}_{\bmF
    \bmB'}{}^{\bmC}{}_{\bmD}\hat{\Gamma}_{\bmA
    \bmA'}{}^{\bmF}{}_{\bmB}-\hat{\Gamma}_{\bmB
    \bmF'}{}^{\bmC}{}_{\bmD}\bar{\hat{\Gamma}}_{\bmA
    \bmA'}{}^{\bmF'}{}_{\bmB'} +\hat{\Gamma}_{\bmF \bmA'}{}^{\bmC}{}
  _{\bmD}\hat{\Gamma}_{\bmB\bmB'}{}^{\bmF}{}_{\bmA} \nonumber
  \\ && \hspace{3cm} + \hat{\Gamma}_{\bmA
    \bmF'}{}^{\bmC}{}_{\bmD}\bar{ \hat{\Gamma}}_{\bmB
    \bmB'}{}^{\bmF'}{}_{\bmA'} + \hat{\Gamma}_{\bmA
    \bmA'}{}^{\bmC}{}_{\bmE} \hat{\Gamma}_{\bmB
    \bmB'}{}^{\bmE}{}_{\bmD}-\hat{\Gamma}_{\bmB
    \bmB'}{}^{\bmC}{}_{\bmE} \hat{\Gamma}_{\bmA
    \bmA'}{}^{\bmE}{}_{\bmD}, \\ 
&&\hat{P}_{\bmA \bmB \bmC \bmC' \bmD
    \bmD'}\equiv \frac{1}{2} \hat{P}_{\bmA}{}^{\bmQ'}{}_{\bmB \bmQ'
    \bmC \bmC' \bmD \bmD'} \\
&& \phantom{\hat{P}_{\bmA \bmB \bmC \bmC' \bmD
    \bmD'}}
= -\Theta\phi_{\bmA \bmB \bmC
    \bmD}\epsilon_{\bmC' \bmD'} - L_{\bmC' (\bmA \bmB)
    \bmD'}\epsilon_{\bmC \bmD} -
  \frac{1}{2}\epsilon_{\bmA \bmB} \left(\hat{L}_{\bmC \bmC' \bmD
    \bmD'}-\hat{L}_{\bmD \bmD' \bmC \bmC'}\right),\\ 
&&
  \hat{\Lambda}_{\bmA \bmA' \bmB \bmC} \equiv
  \frac{1}{2}\hat{\Lambda}_{\bmA \bmA' \bmB \bmQ' \bmC}{}^{\bmQ'} \\
&&\phantom{\hat{\Lambda}_{\bmA \bmA' \bmB \bmC}}=
  \hat{\nabla}^{\bmQ}{}_{\bmA}\phi_{\bmA \bmB \bmC \bmQ}
  -f^{\bmQ}{}_{\bmA'}\phi_{\bmA \bmB \bmC \bmQ}.
\end{eqnarray*}
 With these definitions, the spinorial extended conformal Einstein
 field equations can be alternatively written as
\begin{equation}
\label{contractedXCEFE}
\hat{\Sigma}_{\bmA \bmA'}{}_{\bmB \bmB'}=0, \qquad
\hat{\Xi}^{\bmC}{}_{\bmD \bmA \bmA' \bmB
  \bmB'}=0, \qquad\hat{\Delta}_{\bmC \bmC' \bmD \bmD' \bmB
  \bmB'}=0, \qquad  \hat{\Lambda}_{\bmB \bmB' \bmC \bmD}=0.
\end{equation}
The last set of equations is completely equivalent to 
the equations in \eqref{uncontractedXCEFE}.

\subsubsection{Space spinor formalism }
\label{SpaceSpinorFormalism}

In what follows, let the Hermitian spinor $\tau^{AA'}$ denote the
spinor counterpart of the vector $\sqrt{2}\bme_{\bm0}{}^a$. In
addition, let $\{\epsilon_{\bmA}{}^{A}\}$ with $
\epsilon_{\bm0}{}^{A}= o^{A}, \epsilon_{\bm1}{}^{A}= \iota^{A}$ denote
a spinor dyad such that
\[ 
\tau^{A A'}=\epsilon_{\bm0}{}^{A}\epsilon_{\bm0'}{}^{A'} 
+ \epsilon_{\bm1}{}^{A}\epsilon_{\bm1'}{}^{A'}.
\]
 We have chosen the normalisation $\tau^{AA'}\tau_{AA'}= 2$,
 in accordance with the conventions of \cite{Fri91}. In what follows
let $\tau^{\bmA \bmA'}$ denote the components of $\tau^{AA'}$  respect
to $\{\epsilon^{\bmA}{}_{A}\}$. The Hermitian spinor $\tau^{AA'}$ 
 can  be used to
perform a \emph{space spinor split} of the frame 
 $\{\bme_{\bmA \bmA'}\}$ and coframe $\{\bm\omega^{\bmA
  \bmA'}\}$. Namely, one can write 
\begin{equation} 
\label{FrameSplit}
 \bme_{\bmA \bmA'}=\frac{1}{2}\tau^{\bmA
   \bmA'}\bme-\tau^{\bmB}{}_{\bmA'}\bme_{\bmA \bmB}, \qquad
 \bm\omega^{\bmA \bmA'}=\frac{1}{2}\tau^{\bmA \bmA'}\bm\omega +
 \tau_{\bmC}{}^{\bmA'}\bm\omega^{\bmC \bmA},
\end{equation}
where
\[
\bme \equiv \tau^{\bmP \bmP'}\bme_{\bmP \bmP'}, \qquad
\bme_{\bmA \bmB}\equiv
\tau_{(\bmA}{}^{\bmP'}\bme_{\bmB)\bmP'}, \qquad \bm\omega
\equiv \tau_{\bmP \bmP'}\bm\omega^{\bmP \bmP'}, \qquad 
\bm\omega^{\bmA \bmB}=-\tau^{(\bmA}{}_{\bmP'}\bm\omega^{\bmB) \bmP'}.
\]
It follows from the above expressions that the metric $\bmg$ admits
the split
\[ 
\bmg =\frac{1}{2}\bm\omega \otimes \bm\omega +
 h_{\bmA \bmB \bmC \bmD}\bm\omega^{\bmA \bmB}\otimes \bm\omega^{\bmC
   \bmD}
\]
 where 
\[
 h_{\bmA \bmB \bmC \bmD} \equiv \bmg(\bme_{\bmA
  \bmB},\bme_{\bmC \bmD})=
-\epsilon_{\bmA(\bmC}\epsilon_{\bmD)\bmB}. 
\]
Similarly, any general connection $\breve{\bmnabla}$ 
can be split as
\[ 
\breve{\nabla}_{\bmA \bmA'}= \frac{1}{2}\tau_{\bmA
  \bmA'}\mathcal{P}-\tau_{\bmA'}{}^{\bmQ}\mathcal{\breve{D}}_{\bmA \bmQ},
\]
where
\[
\mathcal{P}\equiv\tau^{\bmA \bmA'}\breve{\nabla}_{\bmA \bmA'} \qquad
\text{and} \qquad \mathcal{\breve{D}}_{\bmA \bmB}\equiv
\tau_{(\bmB}{}^{\bmA'}\breve{\nabla}_{\bmA)\bmA'},
\]
denote, respectively, the derivative along the direction given by
$\tau^{\bmA \bmA'}$ and $\mathcal{\breve{D}}_{\bmA \bmB}$ is the \textit{Sen
  connection} of $\breve{\nabla}$ relative to $\tau^{\bmA \bmA'}$.

The Hermitian spinor
$\tau^{AA'}$ induces a notion of Hermitian conjugation: given an
arbitrary spinor with components $\mu_{\bmA \bmA'}$ its Hermitian
conjugate has components
\begin{equation}
\label{HermitianConjugation}
 \mu^{\dagger}_{\bmC \bmD} \equiv\tau_{\bmC}{}^{\bmA'} \tau_{\bmD}{}^{\bmA}
 \overline{\mu_{\bmA \bmA'}} = \tau_{\bmC}{}^{\bmA'} 
\tau_{\bmD}{}^{\bmA} \overline{\mu}_{\bmA' \bmA},
\end{equation}
 where the bar denotes complex conjugation. In a similar
manner, one can extend the definition to contravariant indices
 and higher valence spinors  by requiring that
 $(\bmpi \bmmu)^\dagger=\bmpi^\dagger\bmmu^\dagger$. 

\subsection{Conformal evolution and constraint equations}

In this section the evolution equations implied by the
extended conformal field equations and the conformal
Gaussian gauge are discussed. In addition, a brief overview of the
conformal constraint equations is given.

\subsubsection{Conformal Gauss gauge in spinorial form  and evolution equations}
\label{subsectionSpaceSpinorAndEvolutioneqs}

 The space spinor formalism leads to a systematic split of the
 extended conformal Einstein field equations \eqref{contractedXCEFE}
 into evolution and constraint equations.  To
this end, one performs a space spinor split for the fields 
$\bme_{\bmA \bmA'}$, $f_{\bmA \bmA'}$, $\hat{L}_{\bmA \bmA'}$,
$\hat{\Gamma}_{\bmA \bmA'}{}^{\bmB}{}_{\bmC}$ .

The frame coefficients $e_{\bmA\bmA'}{}^{\bma}$ satisfy formally
identical splits to those in \eqref{FrameSplit}, where $\bme_{\bmA
\bmA'}= e_{\bmA \bmA'}{}^{\bma}\bmc_{\bma}$ with $\bmc_{\bma} \in
\{\bmpartial_{\tau}, \bmc_{\bmi} \}$ representing a fixed frame field
---the latter is not necessarily $\bmg$-orthonormal.  Observe that, in
terms of tensor frame components, the gauge condition
\eqref{GaugeAdaptTimelikeLeg} implies that $\bme_{\bm0}{}^{\bma} =
\delta_{\bm0}{}^{\bma}$. The gauge conditions
\eqref{conformalGaugeConditionsFrameVersion} and
\eqref{GaugeAdaptTimelikeLeg} are rewritten as
\begin{equation}
\label{conformalGaugeConditionsSpacetimeSpinors} \tau^{\bmA
\bmA'}\bme_{\bmA \bmA'}=\sqrt{2}\bm\partial_{\tau}, \qquad \tau^{\bmA
\bmA'}\hat{\Gamma}_{\bmA \bmA'}{}^{\bmB}{}_{\bmC}=0, \qquad \tau^{\bmA
\bmA'}\hat{L}_{\bmA \bmA' \bmB \bmB'}=0.
\end{equation}
In addition, we define
\begin{subequations}
\begin{eqnarray} 
\label{spacespinorXCEFEfirstdefinitions}
&  \hat{\Gamma}_{\bmA \bmB \bmC \bmD} \equiv \tau^{\bmB \bmA'}
\hat{\Gamma}_{\bmA \bmA'\bmC \bmD}, \qquad
 \Gamma_{\bmA \bmB \bmC \bmD} \equiv \tau_{\bmB}{}^{\bmA'}\Gamma_{\bmA \bmA' \bmC \bmD},
 \qquad f_{\bmA \bmB} \equiv \tau_{\bmB}{}^{\bmA'}f_{\bmA \bmA'}, &
\\
 &{L}_{\bmA \bmB \bmC \bmD} \equiv
 \tau_{\bmB}{}^{\bmA'}\tau_{\bmD}{}^{\bmC'}\hat{L}_{\bmA \bmA' \bmC \bmC'}, 
\qquad \Theta_{\bmA \bmB \bmC \bmD} \equiv {L}_{\bmA \bmB (\bmC \bmD)}
 \qquad \Theta_{\bmA \bmB}\equiv {L}_{\bmA \bmB \bmQ}{}^{\bmQ}.&
\end{eqnarray}
\end{subequations}

 Recalling equation
\eqref{SpinorWeylConnectionToLeviCivitaConnection} one obtains
\[ 
\hat{\Gamma}_{\bmA \bmB \bmC \bmD}=\Gamma_{\bmA \bmB \bmC \bmD }- 
\epsilon_{\bmC \bmA}f_{\bmD \bmA'}\tau_{\bmB}{}^{\bmA'},
\]
 where $\Gamma_{\bmA \bmB \bmC \bmD} \equiv
\tau_{\bmB}{}^{\bmA'}\Gamma_{\bmA \bmA'\bmC \bmD}$. This relation allows us
to write the equations in terms of the reduced connection coefficients
of the Levi-Civita connection of $\bmg$ instead of the reduced
connection coefficients of $\hat{\nabla}$. Only the
spinorial counterpart of the Schouten tensor of the connection
 $\hat{\nabla}$ will
not be written in terms of its Levi-Civita counterpart. Exploiting the
notion of Hermitian conjugation given in equation 
\eqref{HermitianConjugation} one defines
\[ 
\chi_{\bmA \bmB \bmC \bmD} \equiv -\frac{1}{\sqrt{2}}
\left( \Gamma_{\bmA \bmB \bmC \bmD} + \Gamma_{\bmA \bmB \bmC
  \bmD}^{\dagger}\right),
 \qquad  \xi_{\bmA \bmB \bmC \bmD} \equiv
 \frac{1}{\sqrt{2}}\left(\Gamma_{\bmA \bmB \bmC \bmD} -\Gamma_{\bmA \bmB
   \bmC \bmD}^{\dagger}\right),
\]
 where $\chi_{\bmA \bmB \bmC \bmD}$ and $\xi_{\bmA \bmB \bmC
  \bmD}$ correspond, respectively, to the real and imaginary part of the
connection coefficients $\Gamma_{\bmA \bmB \bmC \bmD}$. We define the electric and magnetic parts of the
rescaled Weyl spinor as
\[ 
\eta_{\bmA \bmB \bmC \bmD} \equiv  \frac{1}{2} 
\left( \phi_{\bmA \bmB \bmC \bmD} + \phi_{\bmA \bmB \bmC
  \bmD}^{\dagger}\right),
 \qquad \mu_{\bmA \bmB \bmC \bmD}\equiv -\frac{1}{2}\mbox{i}
 \left( \phi_{\bmA \bmB \bmC \bmD} - \phi_{\bmA \bmB \bmC
   \bmD}^{\dagger}\right).  
\]

\medskip
The gauge conditions
\eqref{conformalGaugeConditionsSpacetimeSpinors} can be rewritten in
terms of the spinors defined in
\eqref{spacespinorXCEFEfirstdefinitions} as
\begin{equation}
f_{\bmA \bmB} = f_{(\bmA \bmB)}, \qquad
\Gamma_{\bmQ}{}^{\bmQ}{}_{\bmA \bmB}=0, \qquad
\hat{L}_{\bmQ}{}^{\bmQ}{}_{\bmA \bmB}=0.
\label{ConformalGaugeConditionSpinors}
\end{equation}
 The last condition implies the decomposition
\[
 \hat{L}_{\bmA \bmB \bmC \bmD} =\Theta_{\bmA \bmB \bmC \bmD}
 + \frac{1}{2}\epsilon_{\bmC \bmD}\Theta_{\bmA \bmB} 
\]
for the components of the spinorial counterpart of the Schouten tensor
of the Weyl connection where $\Theta_{\bmA\bmB\bmC\bmD}\equiv
\hat{L}_{(\bmA\bmB)(\bmC\bmD)}$ and $\Theta_{\bmA\bmB} \equiv
\hat{L}_{\bmA\bmB\bmQ}{}^\bmQ$. 

\medskip
The fields defined in the previous paragraphs 
 allow us to derive from the expressions
\begin{subequations}
\begin{eqnarray}
 &\tau^{\bmA \bmA'}\hat{\Sigma}_{\bmA \bmA'}{}^{\bmP \bmP'}{}_{\bmB
    \bmB'} \bme_{\bmP \bmP'}{}^{\bma}=0, \hspace{0.5cm} \tau^{\bmC
    \bmC'} \hat{\Xi}_{\bmA \bmB \bmC \bmC' \bmD \bmD'} =0,& \label{EvEqsZeroQuantities1}\\
 & \tau^{\bmA \bmA'}\hat{\Delta}_{\bmA \bmA' \bmB \bmB' \bmC \bmC'}=0,
 \hspace{1.2cm} \tau_{(\bmA}{}^{\bmA'}\hat{\Lambda}_{|\bmA'|\bmB \bmC
   \bmD)}=0, \label{EvEqsZeroQuantities2}& 
\end{eqnarray}
\end{subequations}
a set of evolution equations for the fields
\[
\chi_{\bmA \bmB \bmC
  \bmD}, \hspace{0.2cm}\xi_{\bmA \bmB \bmC \bmD}, \hspace{0.2cm}
\bme_{\bmA \bmB}{}^{\bm0}, \hspace{0.2cm}
\bme_{\bmA\bmB}{}^{\bmi}, \hspace{0.2cm} f_{\bmA
  \bmB}, \hspace{0.2cm}\Theta_{\bmA \bmB \bmC
  \bmD}, \hspace{0.2cm}\Theta_{\bmA \bmB}, \hspace{0.2cm} \phi_{\bmA
  \bmB \bmC \bmD}.
\]
 Explicitly, one has that
\begin{subequations}
\begin{eqnarray}
&& \partial_{\tau}e_{\bmA \bmB}{}^{\bm0}=-\chi_{(\bmA\bmB)}{}^{\bmP
    \bmQ}e_{\bmP \bmQ}{}^{\bm0}-f_{\bmA
    \bmB}, \label{EvolutionEquation3Plus1Decomposition1} \\ &&
  \partial_{\tau}e_{\bmA \bmB}{}^{\bmi}=-\chi_{(\bmA \bmB)}{}^{\bmP
    \bmQ}e_{\bmP
    \bmQ}{}^{\bmi}, \label{EvolutionEquation3Plus1Decomposition2}\\ &&
  \partial_{\tau}\xi_{\bmA \bmB \bmC \bmD}=-\chi_{(\bmA \bmB)}{}^{\bmP
    \bmQ}\xi_{\bmP \bmQ \bmC \bmD} +
  \frac{1}{\sqrt{2}}(\epsilon_{\bmA \bmB}\chi_{(\bmB \bmD)\bmP \bmQ}
  + \epsilon_{\bmB \bmD}\chi_{(\bmA\bmC)\bmP \bmQ})f^{\bmP
    \bmQ}, \label{EvolutionEquation3Plus1Decomposition3}\\ 
   && \hspace{2cm}-\sqrt{2}\chi_{\bmA \bmB (\bmC}{}^{\bmE}f_{\bmD) \bmE}
 -\frac{1}{2}\left(\epsilon_{\bmA \bmC}\Theta_{\bmB \bmD} +
  \epsilon_{\bmB \bmD}\Theta_{\bmA \bmC}\right) - \mbox{i}\Theta\mu_{\bmA
    \bmB \bmC \bmD}, \label{EvolutionEquation3Plus1Decomposition4}\\
   && \partial_{\tau}f_{\bmA \bmB}=-\chi_{(\bmA \bmB)}{}^{\bmP\bmQ}
 f_{\bmP \bmQ} + \frac{1}{\sqrt{2}}\Theta_{\bmA
    \bmB}, \label{EvolutionEquation3Plus1Decomposition5}\\ &&
  \partial_{\tau}\chi_{(\bmA \bmB)\bmC \bmD}=-\chi_{\bmA \bmB}{}^{\bmP
    \bmQ}\chi_{\bmP \bmQ \bmC \bmD}-\Theta_{\bmA \bmB \bmC \bmD} +
  \Theta \eta_{\bmA \bmB \bmC
    \bmD}, \label{EvolutionEquation3Plus1Decomposition6} \\ &&
  \partial_{\tau}\Theta_{\bmA \bmB \bmC \bmD}=-\chi_{(\bmA
    \bmB)}{}^{\bmP\bmQ}L_{\bmP \bmQ (\bmC
    \bmD)}-\dot{\Theta}\eta_{\bmA \bmB \bmC \bmD}+
  \mbox{i}d^{\bmP}{}_{(\bmA}\mu_{\bmB) \bmC \bmD
    \bmP}, \label{EvolutionEquation3Plus1Decomposition7}\\ &&
  \partial_{\tau}\Theta_{\bmA \bmB}=-\chi_{(\bmA \bmB)}{}^{\bmE
    \bmF}\Theta_{\bmE \bmF} + \sqrt{2}d^{\bmP \bmQ}\eta_{\bmA \bmB
    \bmP \bmQ}, \label{EvolutionEquation3Plus1Decomposition8}\\ &&
  \partial_{\tau}\phi_{\bmA \bmB \bmC \bmD}-
  \sqrt{2}\mathcal{D}_{(\bmA}{}^{\bmQ}\phi_{\bmB \bmC
    \bmD)\bmQ}=0. \label{EvolutionEquation3Plus1Decomposition9}
\end{eqnarray}
\end{subequations}

\medskip
The following proposition relates the discussion of the
conformal evolution equations and the full set of extended conformal
field equations given by
\eqref{frameXCEFEvanishZeroQuantities}:

\begin{lemma} [\textbf{\em propagation of the constraints and subsidiary system}] 
\label{Thm:PropagationConstraints}
Assume that the evolution equations extracted
 from equations \eqref{EvEqsZeroQuantities1}-\eqref{EvEqsZeroQuantities2}
 and the conformal Gauss gauge conditions
\eqref{ConformalGaugeConditionSpinors} hold.  Then, the
independent components of the zero quantities
\[
\hat{\Sigma}_{\bmA\bmA'}{}^{\bmB\bmB'}{}_{\bmC\bmC'}, \quad
\hat{\Xi}_{\bmA\bmB\bmC\bmC'\bmD\bmD'}, \quad
\hat{\Delta}_{\bmA\bmA'\bmB\bmB'\bmC\bmC'}, \quad \delta_{\bmA\bmA'},
\quad \gamma_{\bmA\bmA'\bmB\bmB'}, \quad \zeta_{\bmA\bmA'},
\]
 which are not determined by the evolution
equations or the gauge conditions, satisfy a symmetric hyperbolic
system of equations whose lower order terms are homogeneous in the
zero-quantities.
\end{lemma}

The proof of Lemma \ref{Thm:PropagationConstraints} can be
found in \cite{Fri95, Fri98c, CFEBook} ---see also \cite{LueVal12}
for a discussion of these equations in the presence of an
electromagnetic field.

\medskip
The most important consequence of Lemma
\ref{Thm:PropagationConstraints} is that if the zero-quantities vanish
at some initial hypersurface and the evolution equations
\eqref{EvolutionEquation3Plus1Decomposition1}-
\eqref{EvolutionEquation3Plus1Decomposition8} are satisfied, then the
\emph{full} extended conformal Einstein field equations encoded in
\eqref{uncontractedXCEFE} are satisfied in the development of the
initial data. This is a consequence of the standard uniqueness result
for \emph{homogeneous} symmetric hyperbolic systems.

\subsubsection{Controlling the gauge}\label{controllingthegauge}

The derivation of the conformal evolution equations
\eqref{EvolutionEquation3Plus1Decomposition1}-\eqref{EvolutionEquation3Plus1Decomposition8}
is based on the assumption of the 
existence of a non-intersecting congruence of conformal geodesics. To
verify this assumption one has to analyse the deviation vector of the
congruence.

\medskip
Let $\bmz$ denote the deviation vector of the congruence. One has then that
\begin{equation}
[\dot{\bmx},\bmz]=0.
\label{BracketDeviationVector}
\end{equation}
Now, let $\bmz^{\bmA \bmA'}$ denote the spinorial counterpart of the
components $z^{\bma}$ of $\bmz$ respect to a Weyl propagated frame
$\{\bme_{\bma}\}$.  Following the spirit of the space spinor formalism
one defines $z_{\bmA \bmB}\equiv \tau_{\bmB}{}^{\bmA'}z_{\bmA
  \bmA'}$. This spinor can be decomposed as
\[z_{\bmA \bmB}=\frac{1}{2}z\epsilon_{\bmA \bmB}+ z_{(\bmA \bmB)} \]
The evolution equations for the deviation vector can be readily
deduced from the commutator \eqref{BracketDeviationVector}. Expressing
the latter in terms of the fields appearing in the extended conformal
field equations one obtains
\begin{subequations}\label{ConformalDeviationEquations}
\begin{eqnarray}
&& \partial_{\tau}z=f_{\bmA \bmB}z^{(\bmA \bmB)} \label{ConformalDeviationEquations1}\\ &&
  \partial_{\tau}z_{(\bmA \bmB)}=\chi_{\bmC \bmD (\bmA \bmB)}z^{(\bmC
    \bmD)}\label{ConformalDeviationEquations2}
\end{eqnarray}
\end{subequations}
The congruence of conformal geodesics is non-intersecting as long as
$z_{(\bmA \bmB)}\neq 0$. Once one has solved equations
\eqref{EvolutionEquation3Plus1Decomposition1}-
\eqref{EvolutionEquation3Plus1Decomposition9} one can substitute
$f_{\bmA \bmB}$ and $\chi_{\bmA \bmB \bmC \bmD}$ into equation
\eqref{ConformalDeviationEquations} and analyse the evolution of the
deviation vector ---for further discussion see \cite{LueVal09}.

\subsubsection{The conformal constraint equations}
\label{ConformalConstraintEquations}

The \emph{conformal constraint equations} encode the set of
restrictions induced by the zero-quantities on the various fields on
hypersurfaces of the unphysical spacetime $(\mathcal{M},\bmg)$.  In
what follows, we will consider a setting where the 1-form $\bmf$
vanishes. Accordingly, the initial data for the \emph{extended}
conformal evolution equations
\eqref{EvolutionEquation3Plus1Decomposition1}-\eqref{EvolutionEquation3Plus1Decomposition8}
and those for the \emph{standard} conformal Einstein field equations  
are the same ---see Appendix
\ref{Appendix:CFE}.  Now, let $\mathcal{\tilde{S}}$ denote a
3-dimensional submanifold of $\tilde{\mathcal{M}}$. The metric
$\tilde{\bmg}$  induces a 3-dimensional metric $\tilde{\bmh}=
\tilde{\varphi}^{*}\tilde{\bmg}$ on $\tilde{\mathcal{S}}$, where
$\tilde{\varphi}: \tilde{\mathcal{S}} \rightarrow \tilde{\mathcal{M}}$
is an embedding.  Similarly, one can consider a 3-dimensional
submanifold $\mathcal{S}$ of $\mathcal{M}$ with induced metric
 $\bmh=\varphi^{*}\bmg$, such that
\[
\bmh = \Omega^2 \tilde{\bmh}, 
\]
 where $\Omega$ denotes the restriction of the conformal
factor to the initial hypersurface $\mathcal{S}$ ---in Section
\ref{ConformalGeodesicsAndGaussianSystems}  this restriction is denoted by $\Theta_{\star}$. 

Let $n_{a}$ and $\tilde{n}_{a}$ with
$n_{a}=\Theta\tilde{n}_{a}$ be, respectively, the $\bmg$-unit and
$\tilde{\bmg}$-unit normals, so that $ n^an_a=\tilde{n}^a\tilde{n}_a=1$
---in accordance with our signature conventions for an spacelike
hypersurface. With these definitions, the second fundamental forms
$\chi_{ab} \equiv h_{a}{}^{c}\nabla_{c}n_{b}$ and $\tilde{\chi}_{ab}
\equiv \tilde{h}_{a}{}^{c}\tilde{\nabla}_{c}\tilde{n}_{b}$ are related
 by the formula
\[ 
\chi_{ab}=\Omega(\tilde{\chi}_{ab} + \Sigma\tilde{h}_{ab})
\]
 where $\Sigma\equiv n^{a}\nabla_{a}\Omega$. 

The conformal constraint equations are conveniently expressed in terms
of a frame $\{\bme_\bmi \}$ adapted to the hypersurface
$\mathcal{S}$ ---that is, the vectors $\bme_\bmi$ span $T\mathcal{S}$
and, thus, are orthogonal to its normal. All the fields appearing in
the constraint equations are expressed in terms of this
frame. The \emph{conformal constraint equations} are then given by:
\begin{subequations}
\begin{eqnarray}
&& D_{\bmi}D_{\bmj}\Omega = -\Sigma \chi_{\bmi \bmj} -\Omega
  L_{\bmi\bmj} + s h_{\bmi \bmj}, \label{CEFEConstraints1} \\ &&
  D_{\bmi}\Sigma = \chi_{\bmi}{}^{\bmk}D_{\bmk}\Omega -\Omega
  L_{\bmi}, \label{CEFEConstraints2}\\ && D_{\bmi}s =-L_{\bmi}\Sigma -
  \Omega L_{\bmi}, \label{CEFEConstraints3} \\ && D_{\bmi}L_{\bmj
    \bmk}-D_{\bmj}L_{\bmi \bmk} =-\Sigma d_{\bmi \bmj \bmk} + d_{\bml
    \bmk \bmi \bmj} D^{\bml}\Omega - (\chi_{\bmi \bmk}L_{\bmj}
  -\chi_{\bmj\bmk}L_{\bmi}), \label{CEFEConstraints4}\\ &&
  D_{\bmi}L_{\bmj}-D_{\bmj}L_{\bmi}=d_{\bml \bmi \bmj}D^{\bml}\Omega +
  \chi_{\bmi}{}^{\bmk}L_{\bmj \bmk}-\chi_{\bmj}{}^{\bmk}L_{\bmi
    \bmk}, \label{CEFEConstraints5} \\ && D^{\bmk}d_{\bmk \bmi
    \bmj}=\chi^{\bmk}{}_{\bmi}d_{\bmj\bmk}
  -\chi^{\bmk}{}_{\bmj}d_{\bmi \bmk}, \label{CEFEConstraints6} \\ &&
  D^{\bmi}d_{\bmi \bmj}=\chi^{\bmi \bmk}d_{\bmi \bmj
    \bmk}, \label{CEFEConstraints7}\\ &&
  D_{\bmj}\chi_{\bmk\bmi}-D_{\bmk}\chi_{\bmj\bmi}=\Omega d_{\bmi \bmj
    \bmk} + h_{\bmi
    \bmj}L_{\bmk}-h_{\bmi\bmk}L_{\bmj}, \label{CEFEConstraints8}\\ &&
  l_{\bmi\bmj}=\Omega d_{\bmi \bmj} + L_{\bmi
    \bmj}-\chi_{\bmk}{}^{\bmk} (\chi_{\bmi \bmj} - \frac{1}{4}\chi
  h_{\bmi \bmj})+ \chi_{\bmk \bmi}
  \chi_{\bmj}{}^{\bmk}-\frac{1}{4}\chi_{\bmk \bml}\chi^{\bmk
    \bml}, \label{CEFEConstraints9} \\ && \lambda=6\Omega s -3\Sigma^2
  -3 D_{\bmk}\Omega D^{\bmk}\Omega, \label{CEFEConstraints10}
\end{eqnarray}
\end{subequations}
 where $D$ is the Levi-Civita connection on
$(\mathcal{S},\bmh)$, $l_{\bmi \bmj}$ is the associated Schouten
tensor,  $d_{\bmi \bmj \bmk} \equiv d_{\bmi \bm0 \bmj
  \bmk},\, d_{\bmi \bmj} \equiv d_{\bmi \bm0 \bmj
  \bm0},\, L_{\bmi} \equiv L_{\bm0 \bmi} $ and
$s$ is a scalar field on $\mathcal{S}$ ---see Appendix
\ref{Appendix:CFE} for the definition of $s$ in context of the
conformal Einstein field equations.

\subsubsection{Constraints at the conformal boundary } 
\label{Section:ConstraitsAtScri}

The conformal constraint equations simplify considerably on
hypersurfaces for which $\Omega=0$. If this is the case then equations
\eqref{CEFEConstraints1}-\eqref{CEFEConstraints9}
 reduce to
\begin{subequations}
\begin{eqnarray}
&& s h_{\bmi \bmj}=\Sigma \chi_{\bmi
    \bmj}, \label{AsymptoticConformalConstraints1}\\ 
&& D_{\bmi}\Sigma=0, \label{AsymptoticConformalConstraints2}\\ 
&& D_{\bmi}s=-L_{\bmi}\Sigma, \label{AsymptoticConformalConstraints3}\\ 
&& D_{\bmi}L_{\bmj \bmk}-D_{\bmj}L_{\bmi\bmk}=-\Sigma d_{\bmi \bmj
    \bmk}-(\chi_{\bmi \bmk}L_{\bmj}-\chi_{\bmj
    \bmk}L_{\bmi}),\label{AsymptoticConformalConstraints4} \\ 
&& D_{\bmi}L_{\bmj}-D_{\bmj}L_{\bmi}=\chi_{\bmi}{}^{\bmk} L_{\bmj
    \bmk}-\chi_{\bmj}{}^{\bmk}L_{\bmi
    \bmk}, \label{AsymptoticConformalConstraints5}\\
 && D^{\bmk}d_{\bmk
    \bmi \bmj} = \chi^{\bmk}{}_{\bmi} d_{\bmj
    \bmk}-\chi^{\bmk}{}_{\bmj}d_{\bmi
    \bmk}, \label{AsymptoticConformalConstraints6}\\ 
&& \lambda=-3\Sigma^2, \label{AsymptoticConformalConstraints7}\\ 
&&
  D^{\bmi}d_{\bmi\bmj}=\chi^{\bmi\bmk}d_{\bmi \bmj   \bmk}, 
 \label{AsymptoticConformalConstraints8}
  \\ 
&&D_{\bmj}\chi_{\bmk \bmi}-D_{\bmk}\chi_{\bmj \bmi}=h_{\bmi\bmj}
  L_{\bmk}-h_{\bmi\bmk}L_{\bmj},\label{AsymptoticConformalConstraints9}
  \\
 && l_{\bmi\bmj}=L_{\bmi \bmj}-\chi(\chi_{\bmi \bmj}
  -\frac{1}{4}\chi h_{\bmi \bmj})+\chi_{\bmk\bmi}\chi_{\bmj}{}^{\bmk}
  -\frac{1}{4}\chi_{\bmk\bml}\chi^{\bmk\bml}h_{\bmi
    \bmj}. \label{AsymptoticConformalConstraints10}
\end{eqnarray}
\end{subequations}

A procedure for obtaining a solution for these
 equations has been given in \cite{Fri86a,Fri95}. Direct algebraic
 manipulations yield
\begin{subequations} 
\begin{eqnarray}
 &\Sigma=\sqrt{\displaystyle\frac{|\lambda|}{3}}, \quad \Sigma_{\bmi}=0,
 \quad s=\Sigma\kappa, \quad \chi_{\bmi
   \bmj}=\kappa h_{\bmi \bmj}, \quad 
 L_{\bmi}=- D_{\bmi} \kappa,& \label{SolutionContraintsGeneral1}\\
&L_{\bmi
   \bmj}=l_{\bmi \bmj}+ \displaystyle\frac{1}{2} \kappa^2 h_{\bmi
   \bmj}, \quad d_{\bmi \bmj \bmk}=-\Sigma^{-1}y_{\bmi
   \bmj \bmk},& \label{SolutionContraintsGeneral2}
\end{eqnarray}
\end{subequations}
where $\kappa$ is an smooth scalar function on the initial
 hypersurface and $y_{\bmi\bmj\bmk}$ denotes the components of the
 Cotton tensor of the metric $\bmh$. The only differential
 condition that has to be solved to obtain a full solution to the
 conformal constraint equations is
\begin{equation}
\label{DivergenceElectricWeyl}
D^{\bmi}d_{\bmi \bmj}=0,
\end{equation}
 where $d_{\bmi \bmj}$ is a symmetric tracefree tensor
encoding the initial data for the electric part of the rescaled Weyl
tensor.

\subsection{The formulation of an asymptotic initial value problem}
\label{AsymptoticInitialValueProblem}

In this section we show how the conformal Gaussian gauge can
be used to formulate an asymptotic initial value problem for the
extended conformal Einstein field equations. Thus, in the sequel we
consider an initial hypersurface on which the conformal factor
vanishes so that it corresponds to the conformal boundary of an
hypothetical spacetime. Accordingly, this initial hypersurface will be
denoted by $\mathscr{I}$. 

\subsubsection{The conformal boundary}

Following Lemma \ref{LemmaCF} 
  we can set, without lost of generality, $\tau_{\star}=0$ on
  $\mathscr{I}$. Moreover, it will be assumed that $f_{\bma}$ vanishes
  initially. Accordingly, we have the initial
condition $\bmbeta_{\star}=\Theta_{\star}^{-1}\mathbf{d}\Theta_{\star}$. Recalling that
$\bmd = \Theta \bmbeta$, and $\tilde{\bmg}^\sharp = \Theta^2
\bmg^\sharp$, and using the constraints in \eqref{LemmaCFConstraint} of
Lemma \ref{LemmaCF}  it readily follows, for the asymptotically problem 
(in which ${\Theta}_{\star}=0$), that 
\[ 
\dot{\Theta}_{\star}= \sqrt{\frac{|\lambda|}{3}}.
\]
 Moreover, using again that $\bmd=\Theta\bmbeta$
and requiring $\dot{\bmx}_{\star}$ to be orthogonal to $\mathscr{I}$
(so that $\dot{\bmx}_{\star}=\bme_{0}$), we have
$d_{0\star}=\dot{\Theta}_{\star}$. It follows that 
\[
d_{0\star}=\sqrt{\frac{|\lambda|}{3}}.
\]

The coefficient $\ddot{\Theta}_{\star}$ is fixed by the requirement
$s=\Sigma \kappa$ on $\mathscr{I}$ ---see \cite{Bey07}.
 From the definition of
$s$ and $\Sigma_{\bma}\equiv \nabla_{\bma}\Theta$ it follows that
\begin{equation}
\begin{aligned}\label{FixingThetaDot}
  s_{\star} & =\left(\frac{1}{4}\nabla_{\bma}\nabla^{\bma}\Theta +
 \frac{1}{24}R\Theta\right)_{\star}  =  \frac{1}{4}(\bme_{\bma}\Sigma^{\bma})_{\star}+
 \frac{1}{4}(\Gamma_{\bma}{}^{\bma}{}_{\bmb}\Sigma^{\bmb})_{\star}   \\ 
  & =  \frac{1}{4}\eta^{\bma\bmb} (\bme_{\bma}\bme_{\bmb}\Theta)_{\star}
 +\frac{1}{4}\dot{\Theta}_{\star}(\Gamma_{\bma}{}^{\bma}{}_{\bm0})_{\star}.
\end{aligned}
\end{equation}

Taking into account that $\Theta$ and $\Sigma_{\bmi}$ vanish at
$\mathscr{I}$ we have that $\eta^{\bma
\bmb}(\bme_{\bma}\bme_{\bmb}\Theta)_{\star}=\ddot{\Theta}_{\star}$. Using
the solution to the constraints given in
\eqref{SolutionContraintsGeneral1}-\eqref{SolutionContraintsGeneral2}
and exploiting the properties of the adapted orthonormal frame we have
$(\Gamma_{\bma}{}^{\bma}{}_{\bm0})_{\star}=
(\Gamma_{\bmi}{}^{\bmi}{}_{\bm0})_{\star}=(\chi_{\bmi}{}^{\bmi})_{\star}=
\kappa \delta_{\bmi}{}^{\bmi}=3\kappa$. Substituting into
\eqref{FixingThetaDot} and using that
$s_{\star}=\dot{\Theta}_{\star}\kappa$ one gets
\[
\ddot{\Theta}_{\star}=\dot{\Theta}_{\star}\kappa. 
\]

\medskip
 Summarising, for an asymptotic initial value problem the
 conformal factor implied by the conformal Gaussian gauge is given by
\begin{equation}
\Theta(\tau)= \sqrt{\frac{|\lambda|}{3}}\tau\Big(1 + \frac{1}{2}\kappa
\tau\Big). 
\label{UniversalConformalFactor}
\end{equation}
The conformal factor given by equation
\eqref{UniversalConformalFactor} is, in a certain sense, Universal. It
does not encode any information about the particular details of the
spacetime to be evolved from $\mathscr{I}$. As such, it can be used to
analyse any spacetime with de Sitter-like Cosmological constant as
long as the spacetime has at least one component of the conformal
boundary. If $\kappa \neq 0$ the conformal boundary has
two components located at
\[ 
\tau =0  \hspace{0.5cm} \text{and} \hspace{0.5cm}
\tau=-\frac{2}{\kappa}.
\]
The first zero corresponds to the initial hypersurface
$\mathscr{I}$. The physical spacetime corresponds to the region where
$\Theta\neq0$.  Therefore, the roots of $\Theta$ render two different
regions of $({\mathcal{M}},{\bmg})$ corresponding to two different
conformal representation of $(\tilde{\mathcal{M}},\tilde{\bmg})$. One
of these representations corresponds to the region covered by the
conformal geodesics with $\tau \in[- {2}/{|\kappa|},0]$ 
or $\tau \in [0,2/|\kappa|]$ and other
corresponds to the region  covered by the
conformal geodesics with $\tau \in [0, \infty)$ or $\tau \in
  (-\infty, 0]$ depending on the sign of $\kappa$. 

\begin{remark}
{\em The discussion of the previous paragraphs is
formal: the component of the conformal boundary given by $\tau
=-2/\kappa$ may not be realised in a specific spacetime. This is, in
particular, the case of the extremal and hyperextremal
Schwarzschild-de Sitter spacetimes in which the singularity precludes
reaching the second conformal infinity ---see Figure \ref{fig:eSdSDiagram}.}
\end{remark}

\subsubsection{Exploiting the conformal gauge freedom}
\label{ExploitConformalGauge}

The conformal freedom of the setting allows us to further simplify the
solution to the conformal constraint equations at $\mathscr{I}$. Given
a solution to the conformal Einstein field equations associated to a
metric $\bmg$, it follows from the conformal covariance of the
equations and fields that the conformally related metric $\bmg'
\equiv \vartheta^2\bmg$ for some $\vartheta$ is also a solution. On an
initial hypersurface $\mathcal{S}$ the latter implies implies $\bmh'=
\vartheta^2_{\star} \bmh$. From the definition of the field $s$ ---see
Appendix \ref{Appendix:CFE}--- and the conformal transformation rule
for the Ricci scalar one has that
 \[ 
s'_{\star}= \vartheta^{-1}_{\star}s_{\star} +
\vartheta^{-2}_{\star}(\nabla_{\bmc}\vartheta)_{\star}(\nabla^{\bmc}\Theta)_{\star}.
\]
Thus, the condition $s'=0$ can be solved locally for
$\vartheta_{\star}$. Accordingly, one chooses $\vartheta_{\star}$ so
that $\kappa=0$.  In this gauge $\chi'_{\bmi \bmj}$ and $L'_{\bmi}$
vanish and $L'_{\bmi \bmj}=l'_{\bmi \bmj}$ at $\mathscr{I}$. In
addition, the conformal factor reduces to
\[
\Theta(\tau)= \sqrt{\frac{|\lambda|}{3}}\tau.
\]
In this representation $\Theta$ has only one zero and the second
component of the conformal boundary (if any) is located at an infinite
distance with respect to the parameter $\tau$.

\subsection{The general structure of the conformal evolution equations}
\label{generalsetting}

One of the advantages of the hyperbolic reduction of the extended
conformal Einstein field equations by means of 
conformal Gaussian systems is that it provides \emph{a priori}
knowledge of the location of the conformal boundary of the solutions
to the conformal field equations. Following the discussion in Section
\ref{ConformalGaussSystems}, the conformal geodesics fix the gauge
through equations \eqref{conformalGaugeConditionsFrameVersion} and
\eqref{GaugeAdaptTimelikeLeg}.  The last condition corresponds to the
requirement on the spacetime to possess a congruence of conformal
geodesics and a Weyl propagated frame ---i.e. equations
\eqref{protoGauge} and \eqref{weylPropagationFrame} are satisfied. 
As already mentioned, the system of evolution equations
\eqref{EvolutionEquation3Plus1Decomposition1}-
\eqref{EvolutionEquation3Plus1Decomposition8}
constitutes a symmetric hyperbolic system. This is the key property
for analysing the existence and stability of perturbations of suitable
 spacetimes using the extended conformal Einstein field equations.

\medskip

To discuss the structure of the conformal evolution system in more
detail, let $\bme$ denote the components of the frame $\bme_{\bmA
  \bmB}$, $\mathbf{\Gamma}$ the independent components of $\chi_{\bmA
  \bmB \bmC \bmD}$ and $\xi_{\bmA \bmB \bmC \bmD}$, and $\bm\phi$ the
independent components of the rescaled Weyl spinor $\phi_{\bmA \bmB
  \bmC \bmD}$. Then the evolution equations
\eqref{EvolutionEquation3Plus1Decomposition1}-\eqref{EvolutionEquation3Plus1Decomposition8}
can be written as
\begin{subequations}
\begin{eqnarray} 
&&\partial_{\tau}\bm\upsilon=\mathbf{K} \bm\upsilon +
  \mathbf{Q}(\bm\Gamma)\bm\upsilon + \mathbf{L}
  (x)\bm\phi \label{structureEvolutionEquations1}, 
\\ && (\mathbf{I} +
  \mathbf{A}^{0}(\bme))\partial_{\tau}\bm\phi +
  \mathbf{A}^{\bmi}\partial_{\bmi}\bm\phi=
  \mathbf{B}(\bm\Gamma), \label{structureEvolutionEquations2}
\end{eqnarray}
\end{subequations}
where $\bm\upsilon$ represents the independent components of the
spinors in the conformal evolution equations except for the rescaled
Weyl spinor whose components are represented by $\bm\phi$. In
addition, $\mathbf{I}$ is the $5\times 5$ identity matrix, $\mathbf{K}$ is a
constant matrix, $\mathbf{Q}$, $\mathbf{A}^{0}$, $\mathbf{A}^{\bmi}$,
and $\mathbf{B}$ are smooth matrix valued functions of its arguments
and $\mathbf{L}(x)$ is a matrix valued function depending on the
coordinates. To have an even more compact notation let
$\mathbf{u}\equiv(\bm\upsilon, \bm\phi)$.
Consistent with this notation, let  $\mathring{\mathbf{u}}$
 denote a solution to the evolution equations
\eqref{structureEvolutionEquations1}-\eqref{structureEvolutionEquations2}
arising from data $\mathring{\mathbf{u}}_{\star}$ prescribed on an hypersurface
$\mathcal{S}$. The solution $\mathring{\mathbf{u}}$
 will be regarded as the \emph{reference solution}.
Consider a general perturbation 
succinctly written as $ \mathbf{u}= \ring{\mathbf{u}} +
\breve{\mathbf{u}}$. Equivalently, one considers
\begin{equation}
\label{splitBackgroundPerturbation}
 \bme = \ring{\bme}+
 \breve{\bme}, \qquad \mathbf{\Gamma}= \mathbf{\ring{\Gamma}}
 + \mathbf{\breve{\Gamma}}, \qquad
 \bm\phi= \ring{\bm\phi} + \breve{\bm\phi}.
\end{equation}
Recalling that $\ring{\mathbf{u}}$ is a solution to the conformal
evolution equations
\eqref{structureEvolutionEquations1}-\eqref{structureEvolutionEquations2}
and making use of the split \eqref{splitBackgroundPerturbation} one
obtains that
\begin{subequations} 
\begin{eqnarray}
&& \partial_{\tau}\breve{\bm\upsilon}=\mathbf{K} \breve{\bm\upsilon} +
  \mathbf{Q}(\ring{\bm\Gamma} + \breve{\bm\Gamma})\breve{\bm\upsilon} +
  \mathbf{Q}(\breve{\bm\Gamma})\ring{\bm\upsilon} +
  \mathbf{L}(x)\breve{\bm\phi}, \label{formPerturbationEquations1}
\\ && (\mathbf{I}+
  \mathbf{A}^{0}(\ring{\bme}+\breve{\bme}))\partial_{\tau}\breve{\bm\phi}
  + (\mathbf{I}+
  \mathbf{A}^{0}(\ring{\bme}+\breve{\bme}))\partial_{\tau}\ring{\bm\phi} +
 \mathbf{A}^{\bmi}(\ring{\bme} +
  \breve{\bme})\partial_{\bmi}\breve{\bm\phi} +
\nonumber \\ &&
     \mathbf{A}^{\bmi}(\ring{\bme} +
  \breve{\bme})\partial_{\bmi}\ring{\bm\phi}=\mathbf{B}
(\ring{\bm\Gamma}+\breve{\bm\Gamma})\breve{\bm\phi} + \mathbf{B}
(\ring{\bm\Gamma}+\breve{\bm\Gamma})\ring{\bm\phi}.\label{formPerturbationEquations2}
\end{eqnarray}
\end{subequations}
Equations \eqref{formPerturbationEquations1} and
\eqref{formPerturbationEquations2} are read as equations for the
components of the perturbed fields $\breve{\bm \upsilon}$ and
$\breve{\bmphi}$. These equations are in a form
  where the theory of first order symmetric hyperbolic systems in
  \cite{Kat75} can be applied to obtain a existence and stability
  result for small perturbations of the initial data
  $\mathring{\mathbf{u}}_{\star}$.  This requires however, the introduction
  of the appropriate norms measuring size of the perturbed initial data
  $\breve{\mathbf{u}}_{\star}$.  This general discussion will not be
  developed further, instead, we particularise this discussion in
  Section \ref{PertsOfeSdS} introducing the appropriate norms
  required to analyse the Schwarzschild-de Sitter spacetime as an
  asymptotic initial value problem.

\section{The Schwarzschild-de Sitter spacetime and its conformal structure}

In this section we briefly review  general properties of
the  Schwarzschild-de Sitter spacetime that will be relevant for 
the main analysis of this article.

\subsection{The Schwarzschild-de Sitter spacetime}
\label{Sec:eSdSIntro}

 The
 \emph{Schwarzschild-de Sitter spacetime} is the
spherically symmetric solution to the Einstein field equations
\begin{equation}
\tilde{R}_{ab}= \lambda \tilde{g}_{ab}
\label{EinsteinFieldEquations}
\end{equation}
 with, in the signature conventions of this article, a \emph{negative}
Cosmological constant given in \emph{static coordinates} $(t,r,\theta,\varphi)$ by
\begin{equation}  
\label{SdSmetric}
\tilde{\bmg}_{SdS}=F(r)\mathbf{d} t \otimes \mathbf{d} t -
F(r)^{-1}\mathbf{d} r \otimes \mathbf{d} r - r^2 \bmsigma,
\end{equation}  
where the function $F(r)$ is given by
\begin{equation}\label{Ffor-SdS}
F(r)\equiv 1 - \frac{2m}{r} + \frac{1}{3}\lambda r^2,
 \end{equation}
and $\bmsigma$ is the standard metric on the 2-sphere $\mathbb{S}^2$
\[ 
\bm\sigma \equiv  \mathbf{d}\theta \otimes \mathbf{d}\theta + \sin^2\theta 
\mathbf{d}\varphi \otimes \mathbf{d}\varphi, 
\] 
with $t\in(-\infty,\infty),\; r\in(0,\infty), \;
\theta \in[0,\pi], \;\varphi\in [0,2\pi)$. This solution
reduces to the \emph{de Sitter spacetime} when $m=0$ and to the
\emph{Schwarzschild solution} when $\lambda=0$. 

\begin{remark}
{\em In the following,
we will only consider the case $m>0$ and we will always assume a 
de Sitter-like value for the cosmological constant $\lambda$.}
\end{remark}

 The location of the roots
of  the polynomial $ r - 2m + \tfrac{1}{3}\lambda r^3 $ are determined by
the relation between $m$ and $\lambda$; whenever $0<9m^2|\lambda|<1$
 this polynomial
has two distinct positive roots $r_{b},r_{c}$ 
and a negative root $r_{-}$  located at 
\begin{eqnarray*}
r_{b} &\equiv&\frac{2}{\sqrt{ |\lambda|}}\cos \left( \frac{\alpha}{3} + \frac{4\pi}{3}\right), \\
r_{c}&\equiv&\frac{2}{\sqrt{ |\lambda|}}\cos \left( \frac{\alpha}{3} \right),\\
r_{-} & \equiv & \frac{2}{\sqrt{ |\lambda|}}\cos \left( \frac{\alpha}{3} + \frac{2\pi}{3}\right), 
\end{eqnarray*}
where $\cos \alpha =-3m\sqrt{|\lambda|}$. The positive roots 
$0<r_{b}\leq r_{c}$ correspond, respectively, to a black hole-like
horizon and a Cosmological-like horizon.  
One can classify this 2-parameter family of
 solutions to the Einstein field equations 
 depending on the relation
between the parameters $m$ and $\lambda$ . 
 The \emph{subextremal Schwarzschild-de Sitter}
spacetime arises when the relation between $m$ and $\lambda$ satisfies
\begin{equation}\label{subextSdSCondition}
0<9 m^2 |\lambda|<1.
\end{equation}
If condition \eqref{subextSdSCondition} holds, one can verify that
$F(r)>0$ for $r_{b}<r<r_{c}$ while $F(r)<0$ in the regions $0\leq r<r_{b}$ and
$r>r_{c}$. Consequently, the solution is static for $r_{b}<r<r_{c}$
---see \cite{BicPod95}.  The \emph{extremal Schwarzschild-de Sitter}
spacetime is obtained by setting
\begin{equation}\label{eSdSCondition}
|\lambda|=1/9m^2.
\end{equation}
If the extremal condition \eqref{eSdSCondition} holds, then the black
hole and Cosmological horizons degenerate into a \emph{single Killing
  horizon} at $r=3m$.  Moreover, one has that $F(r)<0$ for $0 \leq r
<\infty$ so that the hypersurfaces of constant coordinate $r$ are
spacelike while those of constant $t$ are timelike and there are no
static regions.  In the extremal case the function $F(r)$ can be
factorised as
\begin{equation}
\label{Ffor-eSdS}  
F(r)=-\frac{(r-3m)^2(r+6m)}{27m^2r}.
\end{equation} 
In the \emph{hyperextremal Schwarzschild-de Sitter} spacetime one
considers
 \begin{equation}\label{hSdSCondition}
9m^2|\lambda| >1.
\end{equation}
In this case one has again $F(r)<0$ for $0 \leq r <\infty$ so that similar remarks
as those for the extremal case hold. The crucial difference with the
extremal case is that in the hyperextremal case there are no horizons.
 Finally, at $r=0$ it can be verified that the spacetime
has a \emph{curvature singularity} irrespective of the relation
between $m$ and $\lambda$ ---in particular, the scalar
$\tilde{C}_{abcd}\tilde{C}^{abcd}$, with $\tilde{C}^a{}_{bcd}$ the
Weyl tensor of the metric $\tilde{\bmg}_{SdS}$, blows up.

\begin{figure}[t]
\centering
\includegraphics[width=1\textwidth]{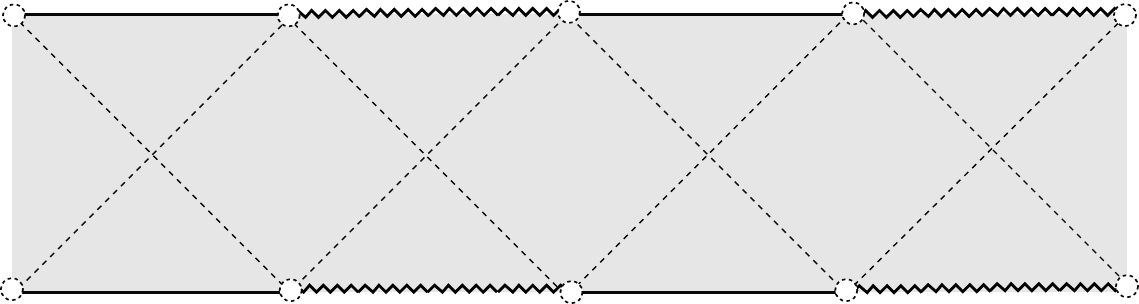}
\put(-390, -8){$\mathscr{I}^{-} (r=\infty)$}
\put(-270, -8){$r=0$}
\put(-185, -8){$\mathscr{I}^{-} (r=\infty)$}
\put(-70, -8){$r=0$}
\put(-390, 114){$\mathscr{I}^{+} (r=\infty)$}
\put(-270, 114){$r=0$}
\put(-185, 114){$\mathscr{I}^{+} (r=\infty)$}
\put(-70, 114){$r=0$}
\put(-418, -8){$\mathcal{Q}$}
\put(-318, -8){$\mathcal{Q'}$}
\put(-213, -8){$\mathcal{Q}$}
\put(-113, -8){$\mathcal{Q'}$}
\put(-10, -8){$\mathcal{Q}$}
\put(-418, 114){$\mathcal{Q}$}
\put(-318, 114){$\mathcal{Q'}$}
\put(-213, 114){$\mathcal{Q}$}
\put(-113, 114){$\mathcal{Q'}$}
\put(-10, 114){$\mathcal{Q}$}
\put(-354, 80){$\mathcal{H}_{c}$}
\put(-388, 80){$\mathcal{H}_{c}$}
\put(-254, 80){$\mathcal{H}_{b}$}
\put(-288, 80){$\mathcal{H}_{b}$}
\put(-150, 80){$\mathcal{H}_{c}$}
\put(-185, 80){$\mathcal{H}_{c}$}
\put(-45, 80){$\mathcal{H}_{b}$}
\put(-76, 80){$\mathcal{H}_{b}$}
\put(-354, 25){$\mathcal{H}_{c}$}
\put(-388, 25){$\mathcal{H}_{c}$}
\put(-254, 25){$\mathcal{H}_{b}$}
\put(-288, 25){$\mathcal{H}_{b}$}
\put(-150, 25){$\mathcal{H}_{c}$}
\put(-185, 25){$\mathcal{H}_{c}$}
\put(-45, 25){$\mathcal{H}_{b}$}
\put(-76, 25){$\mathcal{H}_{b}$}
\caption{Penrose diagram for the subextremal Schwarzschild-de Sitter
  spacetime.  The excluded points $\mathcal{Q}$, $\mathcal{Q'}$
  represent asymptotic regions where the Cosmological horizon appear
  to meet $\mathscr{I}$. As discussed in Section \ref{Sec:eSdSIntro}
  this region of the spacetime does not belong to $\mathscr{I}$.  }
\label{fig:SubSdSDiagram}
\end{figure}

\begin{figure}[t]
\includegraphics[width=1\textwidth]{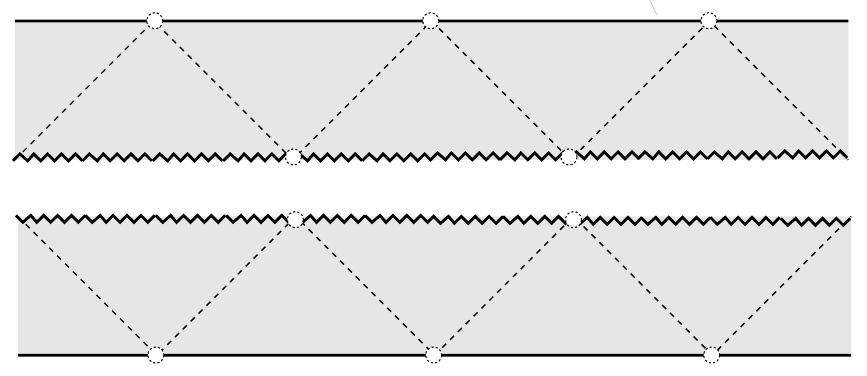}
\put(-400,180){(a)}
\put(-350, 159){$\mathcal{Q}$}
\put(-216, 159){$\mathcal{Q'}$}
\put(-80, 160){$\mathcal{Q}$}
\put(-300,175){$\mathscr{I} \quad (r=\infty)$}
\put(-180,175){$\mathscr{I} \quad (r=\infty)$}
\put(-400,130){$\mathcal{H}$}
\put(-320,130){$\mathcal{H}$}
\put(-242,130){$\mathcal{H}$}
\put(-165,130){$\mathcal{H}$}
\put(-108,130){$\mathcal{H}$}
\put(-28,130){$\mathcal{H}$}
\put(-282, 112){$\mathcal{P}$}
\put(-148, 112){$\mathcal{P}$}
\put(-230,98){$r=0$}
\put(-400,85){(b)}
\put(-230,82){$r=0$}
\put(-282, 63){$\mathcal{P}$}
\put(-148, 63){$\mathcal{P}$}
\put(-300,0){$\mathscr{I} \quad (r=\infty)$}
\put(-180,0){$\mathscr{I} \quad (r=\infty)$}
\put(-348, 15){$\mathcal{Q}$}
\put(-212, 15){$\mathcal{Q'}$}
\put(-78, 15){$\mathcal{Q}$}
\put(-385,30){$\mathcal{H}$}
\put(-317,30){$\mathcal{H}$}
\put(-244,30){$\mathcal{H}$}
\put(-177,30){$\mathcal{H}$}
\put(-112,30){$\mathcal{H}$}
\put(-40,30){$\mathcal{H}$}
\caption{Penrose diagrams for the extremal Schwarzschild-de Sitter
  spacetime. Case (a) corresponds to a \emph{white hole} which evolves
  towards a de Sitter final state while
  case (b) is a model of a black hole with a future singularity. The
  continuous black line denotes the conformal boundary; the serrated
  line denotes the location of the singularity; the dashed line shows
  the location of the Killing   horizons $\mathcal{H}$ at $r=3m$ .
 The excluded points $\mathcal{Q}$, $\mathcal{Q'}$ and
  $\mathcal{P}$ represent asymptotic regions of the spacetime that do not
  belong to $\mathscr{I}$ or the singularity $r=0$. } 
\label{fig:eSdSDiagram}
\end{figure}

\begin{figure}[t]
\centering
\includegraphics[width=0.3\textwidth,angle=90]{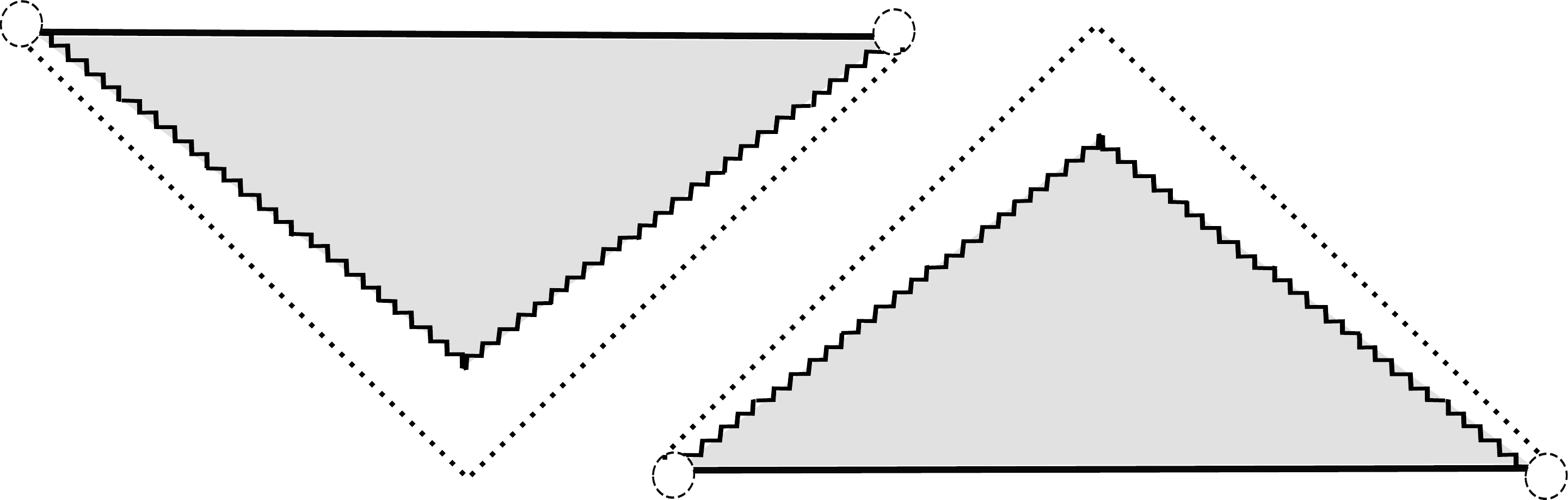}
\put(-410,140){a)}
\put(-396,130){$\mathcal{Q}$}
\put(-175,130){$\mathcal{Q'}$}
\put(-305, 125){$\mathscr{I} \hspace{3mm} (r=\infty)$}
\put(-255, 60){\rotatebox{40}{$r=0$}}
\put(-330, 80){\rotatebox{-38}{$r=0$}}
\put(-160,140){b)}
\put(-230,-10){$\mathcal{Q}$}
\put(-10,-10){$\mathcal{Q'}$}
\put(-145, -10){$\mathscr{I} \hspace{3mm} (r=\infty)$}
\put(-170, 40){\rotatebox{40}{$r=0$}}
\put(-85, 55){\rotatebox{-40}{$r=0$}}
\caption{Penrose diagram for the hyperextremal Schwarzschild-de Sitter
 spacetime. The singularity is of spacelike nature. Dotted lines at
 45 $^{\circ}$ and 135 $^{\circ}$ have been included for
 visualisation.  Case (a) corresponds to a white hole which evolves to a
 final de-Sitter state. Case (b) corresponds to a black hole with a
 future spacelike singularity.  }
\label{fig:HypSdSDiagram}
\end{figure}

\subsection{The $\mathbb{S}^3\backslash\{\mathcal{Q},\mathcal{Q'}\}$-representation}
\label{ConformalStructureSdS}
\medskip
The basic conformal structure of the subextremal and extremal
Schwarzschild-de Sitter spacetimes has already been discussed in
\cite{BazFer85, BicPod95} and \cite{Pod99} respectively. Coordinate
and Penrose diagrams have been also provided in \cite{Gey80} for the
subextremal, extremal and hyperextremal cases.  In this section we
present a concise discussion, adapted to our conventions, of the
conformal structure of the Schwarzschild-de Sitter spacetime in the
subextremal, extremal and hyperextremal cases.  We start our
discussion showing that irrespective of the relation of $m$ and
$\lambda$ the induced metric at the conformal boundary for the
Schwarzschild de Sitter spacetime 
can be identified with the standard metric
on $\mathbb{S}^3$. As discussed in more detail
in Section \ref{MetricOnScri}, this construction
 depends on the particular conformal representation being considered.  In the subextremal case one cannot obtain
simultaneously an analytic
extension regular near both $r_{b}$ and $r_{c}$---see
\cite{BazFer85}.  Since we are interested only in the asymptotic
region, in this section we will consider the region $r>r_{c}$. For the
extremal and hyperextremal cases such considerations are not necessary.

In the following we introduce the null coordinates
\[
u\equiv\sqrt{|\lambda|}(t-\mathfrak{r}), \qquad
v\equiv \sqrt{|\lambda|}(t+\mathfrak{r}),
\]
where $\mathfrak{r}$ is a \emph{tortoise} coordinate given by
\begin{equation}
\label{IntegralTortoise}
\mathfrak{r}\equiv\int \frac{1}{F(r)}dr.
\end{equation}
This integral can be computed explicitly ---see
\cite{BazFer85,BicPod95}.  The particular form of $\mathfrak{r}$
depends on the relation between $\lambda$ and $m$. As
discussed in \cite{BicPod95, Pod99} the integration constant can
always be chosen so that $\mathfrak{r} \rightarrow 0$ as $r
\rightarrow \infty$. Defining $\tan U\equiv u, \tan V\equiv v$, with $U,V \in
            [-\tfrac{\pi}{2},\tfrac{\pi}{2}]$ one gets the line
            element
\begin{equation}\label{MetricWellBehavedHorizon}
\tilde{\bmg}_{SdS}=
\frac{1}{2}\frac{F(r)}{|\lambda|}\sec^2{U}\sec^2{V}\left(\mathbf{d}U\otimes
\mathbf{d}V + \mathbf{d}V \otimes \mathbf{d}U\right)- r^2\bmsigma.
\end{equation}

As discussed in \cite{BicPod95,BazFer85}, one can construct
  \emph{Kruskal type} coordinates covering the black hole horizon
   by choosing appropriately the integration constant in
  equation \eqref{IntegralTortoise}.  Analogously, choosing a
  different integration constant,
  one can construct Kruskal type coordinates covering the cosmological
  horizon.  Nevertheless in the subextremal case,
  as emphasised in \cite{BazFer85}, it
  is not possible to construct Kruskal type coordinates covering
  simultaneously both horizons. To construct the Penrose diagram
  for this spacetime, one considers as building blocks the Penrose
  diagrams for the regions $0 \leq r \leq r_{b}$, $r_{b}\leq r \leq r_{c}$ and
  $r_{c} \leq r < \infty$ which are then glued together using the
  corresponding Kruskal type coordinates to cross each horizon ---see
  \cite{BazFer85,Gey80} for a detailed discussion on the construction
  the Penrose diagram and Kruskal type coordinates in the
Schwarzschild-de Sitter spacetime. Consistent with the above
discussion and given that 
we are only interested in the asymptotic region,
we restrict our attention, in the subextremal case, to 
$r > r_{c}$. In
the extremal case one has, however, that $r_{b}=r_{c}=3m$ and one can
verify that
\[
\lim_{r \to 3m}\frac{\cos U}{r-3m}=\lim_{r \to 3m}\frac{\cos
  V}{r-3m}=C,
\]
where $C\neq 0$ is a constant depending on $m$ and the integration
constant chosen in the definition of $\mathfrak{r}$. Consequently, in
the extremal case, the metric \eqref{MetricWellBehavedHorizon} is well
defined for the whole range of the coordinate $r$: $0<r<\infty$ ---see
\cite{Pod99}. Introducing the coordinates $(\bar{U},\bar{V})$ defined
via
\[
 \tan U \equiv \ln\tan \bigg(\frac{\pi}{4} + \frac{\bar{U}}{2}\bigg), \qquad
 \tan V \equiv \ln\tan \bigg(\frac{\pi}{4} + \frac{\bar{V}}{2}\bigg)
\]
one obtains
\[ 
\tilde{\bmg}_{SdS}=
\frac{1}{2}\frac{F(r)}{|\lambda|}\sec{\bar{U}}\sec{\bar{V}}\left(\mathbf{d}\bar{U}\otimes
\mathbf{d}\bar{V} + \mathbf{d}\bar{V} \otimes
\mathbf{d}\bar{U}\right)- r^2\bmsigma.
\]
Recalling that in the subextremal case $F(r)\leq 0$ for $r\geq r_{c}$
while for the extremal and hyperextremal cases $F(r)\leq 0$ for $0 < r
<\infty$, one identifies the conformal factor
\[
\Xi^{2}=\frac{|\lambda|}{|F(r)|}\cos \bar{U} \cos \bar{V}.
\]
Therefore, we can identify the conformal metric $\bmg_{SdS}= \Xi^2
\tilde{\bmg}_{SdS}$ with 
\begin{equation}
\bmg_{SdS} =-\frac{1}{2}\left(\mathbf{d}\bar{U}\otimes
\mathbf{d}\bar{V} + \mathbf{d}\bar{V} \otimes \mathbf{d}\bar{U}\right)
- \frac{|\lambda|r^2}{|F(r)|}\cos\bar{ U} \cos\bar{ V} \bm\sigma.
\end{equation}
Introducing the coordinates
\[ 
T \equiv \bar{U}+\bar{V}, \qquad \Psi\equiv \bar{V}-\bar{U},
\]
one gets
\[
\bmg_{SdS}= \frac{1}{4}\left ( \mathbf{d}\Psi\otimes \mathbf{d}\Psi -
\mathbf{d}T \otimes \mathbf{d}T \right)
-\frac{|\lambda|r^2}{|F(r)|}\cos \frac{1}{2}\left(T
+\Psi\right)\cos \frac{1}{2}\left(T
-\Psi\right)\bm\sigma.
\]
The analysis in \cite{BazFer85} shows that the conformal factor
$\Xi$ tends to zero as $r \rightarrow \infty$.  Hence, to identify
the induced metric at $\mathscr{I}$ it is sufficient to analyse such
limit.  Noticing that
\[
\mathfrak{r}= \frac{1}{2 \sqrt{|\lambda|}}(v-u) =\frac{1}{2
  \sqrt{|\lambda|}} \ln \left( \frac {\tan(\pi/4 +
  \bar{V})}{\tan(\pi/4 + \bar{U})} \right)
\]
and recalling that
 \[\lim_{r \rightarrow \infty} \mathfrak{r} = 0,\]
 one concludes that  $r\rightarrow \infty$ implies $ \Psi =0$ as
 long as $\bar{U}\neq \pm \tfrac{1}{2}\pi$ and $\bar{V} \neq \pm\tfrac{1}{2}\pi$. 
 Using  equation \eqref{Ffor-SdS} one can verify that
\[
\lim_{r \rightarrow \infty }\frac{|\lambda|r^2}{|F(r)|} = 1.
\]
Consequently, the induced metric on $\mathscr{I}$ is given by
\[
\bmh=-\frac{1}{4}\mathbf{d}T\otimes \mathbf{d}T -
\cos^{2}\frac{T}{2}\bm\sigma
\]
which can be written in a more recognisable form introducing
$\xi\equiv \tfrac{1}{2}(T+\pi)$ so that
\begin{equation}\label{MetricOnS3}
\bmhbar= -\mathbf{d}\xi\otimes \mathbf{d}\xi -\sin^2 \xi \bm\sigma.
\end{equation}
The metric $\bmhbar$ is the standard metric on $\mathbb{S}^3$.
Observe that the excluded points in the discussion of this section
$(\bar{U},\bar{V})=(\pm\tfrac{1}{2}\pi,\pm\tfrac{1}{2}\pi)$ correspond
 to $\xi=0$ and
$\xi=\pi$ ---the North and South pole of $\mathbb{S}^3$.  The Penrose
diagram of the subextremal, extremal and hyperextremal Schwarzschild-de
Sitter spacetime is given in Figure \ref{fig:eSdSDiagram} (a). The
conformal boundary $\mathscr{I}$ of the (subextremal, extremal and
hyperextremal) Schwarzschild-de Sitter spacetime, defined by the
condition $\Xi=0$, is \emph{spacelike} consistent with the fact that
the Cosmological constant of the spacetime is de Sitter-like ---see
e.g. \cite{PenRin86,Ste91}. Moreover, the singularity at $r=0$ is of a
\emph{spacelike nature} ---see \cite{Gey80,Pod99}.  As pointed out in
\cite{BazFer85,GriPod09}, the Schwarzschild-de Sitter spacetime can be
interpreted as the model of a \emph{white hole} singularity towards a
final de Sitter state. Alternatively, making use of a reflection
\[
u \mapsto -u, \qquad v \mapsto -v,
\] 
one obtains a model of a \emph{black hole} with a future singularity
---see Figures \ref{fig:SubSdSDiagram}, \ref{fig:eSdSDiagram} and
\ref{fig:HypSdSDiagram}.

\medskip
In what follows, we adopt the white hole point of
view for the extremal and hyperextremal cases so that $\mathscr{I}$
corresponds to future conformal infinity and we will consider
 a \emph{ backward
  asymptotic initial value problem}. Consistent with this point of
view, for the subextremal case we consider asymptotic initial data on
$\mathscr{I}^{+}$ and study the development of such data towards the
curvature singularity located at $r=0$ ---see Figure \ref{fig:SdSAsymptoticPerturbation}.

\subsection{The $\mathbb{R}\times\mathbb{S}^2$-representation}
\label{TheRcrossS2Representation}

In Section \ref{ConformalStructureSdS} we have shown
  that there exist a conformal representation in which the induced
  metric on the conformal boundary corresponds to the standard metric
  on $\mathbb{S}^3$. A quick inspection shows that the metric \eqref{MetricOnS3} is
\emph{conformally flat}. In this section we put this observation in a
wider perspective and show that the induced metric on $\mathscr{I}$ of
a spherically symmetric spacetime with spacelike $\mathscr{I}$ is
necessarily conformally flat. In addition, a conformal representation
in which the induced metric at the conformal boundary corresponds to
 the standard metric on $\mathbb{R}\times\mathbb{S}^2$ is discussed.
This conformal representation will be of particular importance in 
the subsequent analysis.

\subsubsection{The conformal boundary of spherically symmetric and asymptotically de Sitter spacetimes} 
\label{MetricOnScri}
  Following an argument similar to the
one given in \cite{LueVal14a} we have the following construction for a
spherically symmetric spacetime with spacelike conformal boundary: if
a spacetime $(\tilde{\mathcal{M}},\tilde{\bmg})$ is spherically
symmetric then the metric $\tilde{\bmg}$ can be written in a warped
product form
\begin{equation}
\label{WarpedProductForm}
 \tilde{\bmg}= \tilde{\bmgamma}-\tilde{\rho}^2\bmsigma,
\end{equation}
 where $\tilde{\bmgamma}$ is the 2-metric on the quotient
manifold $\tilde{\mathcal{Q}} \equiv \mathcal{\tilde{M}}/SO(3)$,
$\bmsigma$ is the standard metric of $\mathbb{S}^2$ and $\tilde{\rho}:
\tilde{\mathcal{Q}}\rightarrow \mathbb{R}$. If $\bmg$ and
$\tilde{\bmg}$ are conformally related, $\bmg=\Theta^2 \tilde{\bmg}$, then
the spherical symmetry condition for $\bmg$ is translated into the
requirement that $\bmg$ can be written in the form
\[ 
\bmg=\bmgamma - \rho^2 \bmsigma,
\]
 where $\bmgamma\equiv \Theta \tilde{\bmgamma}$ and $\rho\equiv \Theta
\tilde{\rho}$, where $\Theta$ does not depend on the coordinates on
$\mathbb{S}^2$. Near $\mathscr{I}$ let us introduce local coordinates
$(\Theta,\psi)$ on the quotient manifold $\mathcal{Q} \equiv
\mathcal{M}/SO(3)$ so that $\Theta=0$ denotes the locus of
$\mathscr{I}$. Since the conformal boundary is spacelike we have that
$\bmg(\mathbf{d} \Theta,\mathbf{d} \Theta)>0$. Therefore, the metric
induced on $\mathscr{I}$ by $\bmg$ has the form
\[
\bmh = -A(\psi)\mathbf{d}\psi \otimes \mathbf{d} \psi -
 \rho^2(\psi) \bmsigma,
\]
where $A(\psi)$ is a positive function. Redefining the
coordinate $\psi$ we can rewrite $\bmh$ as
\[ 
\bmh = -\rho^2(\psi)(\mathbf{d} \psi \otimes \mathbf{d} \psi +
\bmsigma).
\]
It can be readily verified ---say, by calculating the Cotton tensor of
$\bmh$--- that the metric $\bmh$ is conformally flat.
In Section \ref{IdentifyAsympRegularData} it will be shown that, in view of the
conformal freedom of the setting, a convenient choice is to consider a conformal representation
in which the the 3-metric on $\mathscr{I}$ is given by
\begin{equation}\label{metricTorus}
 \bmh = -\mathbf{d}\psi \otimes \mathbf{d}\psi - \bm\sigma.
\end{equation}
This metric is the standard metric of the cylinder $\mathbb{R}\times
\mathbb{S}^2$ with $\psi \in (-\infty,\infty)$. It can be verified that this conformal representation
is related to the one discussed in Section 
 \ref{ConformalStructureSdS} via $\bmh=\omega^2\bm\hbar$,
where the conformal factor $\omega$ and
the relation between the coordinates are given by
\begin{equation}\label{ConformalFactorSphereToCylinder}
\psi(\xi)=\psi_{\star}-\ln|\csc \xi + \cot \xi|, \qquad
\omega(\xi)=\csc(\xi).
\end{equation}
Equivalently, one has that 
\[
\xi(\psi)=\arccos \left(
\frac{e^{2(\psi_{\star}-\psi)}-1}{e^{2(\psi_{\star}-\psi)}+1}\right),
\qquad \omega(\psi) = \frac{e^{\psi}}{2e^{\psi_{\star}}} (
e^{2\psi_{\star}} + e^{2\psi}),
\]
where $\psi_{\star}$ is a constant of integration. We can directly
observe that in this representation $\xi=0$ and $\xi=\pi$ correspond
to $\psi= -\infty$ and $\psi= \infty$, respectively.

\subsubsection{The extrinsic curvature of the conformal boundary
 in the $\mathbb{R}\times \mathbb{S}^2$ representation}
\label{PhysicalParameters}

A particularly
simple conformal representation for the Schwarzschild-de
Sitter spacetime can be obtained using the discussion of Section
\ref{MetricOnScri}. Accordingly, take the metric of the 
Schwarzschild-de Sitter spacetime as written in equation \eqref{SdSmetric} with $F(r)$  
as  given by the relation \eqref{Ffor-SdS} and consider the conformal factor
$\wideparen{\Xi}\equiv 1/r$. Introducing the coordinates
$\varrho\equiv {1}/{r}$ and  $\zeta\equiv\sqrt{|\lambda|/3}t $, the conformal metric
 \[
\wideparen{\bmg}\equiv \wideparen{\Xi}^2{}\tilde{\bmg}_{eSdS}
\] 
is given by
\[
\wideparen{\bmg} = \frac{3}{|\lambda|}\Big(\varrho^2 -2m\varrho^3 -
\frac{1}{3}|\lambda| \Big) \mathbf{d} \zeta \otimes \mathbf{d} \zeta -
\Big(\varrho^2 -2m\varrho^3 - \frac{1}{3}|\lambda|\Big)^{-1}
\mathbf{d} \varrho \otimes \mathbf{d}\varrho - \bmsigma.
\]
The induced metric on the hypersurface described by the condition
$\wideparen{\Xi}=0$ is given by
\[ 
\wideparen{\bmh}=- \mathbf{d} \zeta \otimes \mathbf{d}\zeta -
\bmsigma.
\]

It can be verified that \emph{$\wideparen{\bmg}$
satisfies a conformal gauge for which the
conformal boundary has vanishing extrinsic curvature}. To see this, 
consider a $\wideparen{\bmg}$-orthonormal coframe $\{ \bmomega^\bma\}$ 
with 
\[    
\bm\omega^0 = \sqrt{\frac{3}{|\lambda|}}
\left(\varrho^2 -2m\varrho^3
 -\frac{1}{3}|\lambda|\right)^{1/2}\mathbf{d}\zeta 
, \qquad \bm\omega^3 =  \left(\varrho^2 -2m\varrho^3 
-\frac{1}{3}|\lambda|\right)^{-1/2}\mathbf{d}\varrho,
\] 
and $\{\bmomega^1,\,\bmomega^2\}$ a $\bmsigma$-orthonormal
coframe. Denote by $\{ \bme_\bma \}$ the corresponding dual
frame. Using this frame we can directly compute the Friedrich scalar 
 $\wideparen{s} \equiv
\tfrac{1}{4}\wideparen{\nabla}^{\bmc}\wideparen{\nabla}_{\bmc}\wideparen{\Xi}
+ \tfrac{1}{24}\wideparen{R}\hspace{1mm}\wideparen{\Xi}$ ---see
Appendix \ref{Appendix:CFE}.  The computation of the Ricci scalar
yields 
\begin{equation}
\label{RicciEdgarStyleConformalFactor}
 \wideparen{R}=-12m\varrho.
\end{equation}
A direct calculation using 
\[
\wideparen{\nabla}_{\mu}\wideparen{\nabla}^{\mu}\Xi=
\frac{1}{\sqrt{-\det \wideparen{\bmg}}}\partial_{\mu}
(\sqrt{-\det \wideparen{\bmg}} \hspace{1mm} \wideparen{g}^{\mu\nu}\partial_{\nu}\Xi)
\]
shows that $\wideparen{\nabla}_{a}\wideparen{\nabla}^{a} \Xi =
6m\varrho^2 -2\varrho$. Consequently, the scalar $\wideparen{s}$
vanishes at the hypersurface defined by
$\wideparen{\Xi}=\varrho=0$. Contrasting this result with the solution
to the conformal constraints given in equations
\eqref{SolutionContraintsGeneral1}-\eqref{SolutionContraintsGeneral2}
we conclude that in this representation the hypersurface described by
$\wideparen{\Xi}=0$ has vanishing extrinsic curvature as claimed.

\begin{remark}
{\em Notice that, in this
  representation the curvature singularity, located $r=0$, corresponds
  to $\varrho=\infty$.  Consequently, $\mathscr{I}$ is at an infinite
  distance from the conformal boundary.}
\end{remark}

\medskip
Observe that,
 the components of the Weyl tensor with respect to the orthonormal
frame $\{ \bme_\bma \}$ as described
 above are given by
\[ 
C_{1212}=-2m\varrho, \hspace{0.2cm} C_{1313}=m\varrho, \hspace{0.2cm}
C_{1010}=-m\varrho, \hspace{0.2cm} C_{2323}=m\varrho, \hspace{0.2cm}
C_{2020}=-m\varrho, \hspace{0.2cm}C_{3030}=2m\varrho.
\]
This information will be required in the discussion of the initial data
for the rescaled Weyl tensor ---see Section 
\ref{Section:InitialDataRescaledWeylFrame}. Using now that $d_{\bma \bmb \bmc \bmd}= \Xi^{-1}C_{\bma \bmb \bmc
  \bmd}$ with $\wideparen{\Xi}=\xi$ and exploiting the fact that the
computations have been carried out in an orthonormal frame so that
$C^{\bma}{}_{\bmb \bmc \bmd}=\eta^{\bma \bmf}C_{\bmf \bmb \bmc \bmd}$,
we get
\[ 
d_{1212}=-2m, \hspace{0.5cm} d_{1313}=m, 
\hspace{0.5cm} d_{1010}=-m, \hspace{0.5cm} d_{2323}=m,
\hspace{0.5cm} d_{2020}=-m, \hspace{0.5cm}d_{3030}=2m. 
\]
 Finally, considering $d_{ij}\equiv d_{i 0 j 0}$ we have 
\begin{equation}
\label{ComparisonInitialDatumAstar}
d_{11}=-m, \hspace{0.3cm} d_{22}=-m, \hspace{0.3cm} d_{33}=2m.
\end{equation}

\subsection{Identifying asymptotic regular data }
\label{IdentifyAsympRegularData}

As discussed in Section \ref{Sec:eSdSIntro}, 
there is a conformal representation in which the induced 
 metric on the
conformal boundary of the Schwarzschild-de Sitter is the standard
metric $\bmhbar$ on $\mathbb{S}^3$. Nevertheless, the asymptotic points
 $\mathcal{Q}$ and $\mathcal{Q'}$, as depicted in the Penrose diagram of Figure
\ref{fig:eSdSDiagram}, are associated to the behaviour of those
timelike geodesics which never cross the horizon ---see Appendix
\ref{AsymptoticPointsQQprime}.  Despite that, from
the point of view of the intrinsic geometry of $\mathscr{I}$ these
asymptotic regions ---corresponding to the North and South poles of
$\mathbb{S}^3$--- are regular, from a spacetime point of view they are
not. This issue will be further discussed Section
\ref{Section:InitialDataRescaledWeylFrame} where it will be shown that
the initial data for the electric part of rescaled Weyl tensor is
singular at $\mathcal{Q}$ and $\mathcal{Q'}$.  Fortunately, as
exposed in Section \ref{ExploitConformalGauge} one can exploit the
inherent conformal freedom of the setting to select any representative
of the conformal class $[\bmhbar]$ to construct a solution to the
conformal constraint equations.  Taking into account the previous
remarks it will be convenient to choose the
  conformal representation discussed in Section \ref{TheRcrossS2Representation},
  $\bmh=\omega^2\bmhbar$ with $\omega$ and $\bmh$ given in equations
\eqref{metricTorus} and \eqref{ConformalFactorSphereToCylinder},
 in which the points $\mathcal{Q}$ and
  $\mathcal{Q'}$ are at infinity respect to the metric $\bmh$.

\subsubsection{A frame for the induced metric at $\mathscr{I}$}
\label{FrameInitialData}

  Consistent with the discussion of the last section,
 on $\mathscr{I}$ one considers an adapted
frame $\{\bml,\, \bmm,\bar{\bmm}\}$ such that the metric
\eqref{metricTorus} can be written
in the form
\[
\bmh = -(  \bml\otimes \bml +\bmsigma)
\]
where 
\[
\bml=\mathbf{d}\psi, \qquad \bm\sigma= \frac{1}{2}
 ( \bmm \otimes \bar{\bmm} + \bar{\bmm}\otimes \bmm ).
\]
In terms of abstract index notation we have 
\begin{equation}
\label{metricTorusFrame}
h_{ij} = -l_{i}l_{j} - 2m_{(i}\bar{m}_{j)}.
 \end{equation}
The frame $\{\bml,\, \bmm,\bar{\bmm}\}$
satisfies the pairings 
\begin{equation} 
\label{NormalisationConditions}
l_{j}l^{j}=-1, \qquad  m_{j}\bar{m}^{j}=-1, \qquad  l_jm^j=
l_j\bar{m}^j = m_jm^j=\bar{m}_j\bar{m}^j=0.
\end{equation}

\subsubsection{Initial data for the rescaled Weyl tensor}
\label{Section:InitialDataRescaledWeylFrame}

The procedure for the construction of a solution to the conformal
constraints at the conformal boundary requires, in particular, a solution to the
divergence equation \eqref{DivergenceElectricWeyl} for the electric
part of the rescaled Weyl tensor.  The requirement of spherical
symmetry of the spacetime can be succinctly incorporated using the
results in \cite{Pae14}.  If the unphysical spacetime
$(\mathcal{M},\bmg)$ possesses a Killing vector $\bmX$ then the initial
data encoded in the symmetric tracefree tensor $d_{ij}$ must satisfy
the condition
\begin{equation}
\label{KIDcondition}
 \pounds _{\bmX} d_{ij}=0,
\end{equation}
 where $ \pounds_{\bmX}$ denotes the Lie derivative in the direction
 of $\bmX$ on the initial hypersurface. The only  symmetric tracefree
 tensor $d_{\bmi\bmj}$ compatible with the above requirement is given
 by
\begin{equation}\label{formTTtensorSphericallySymmetric}
 d_{ij}= \frac{1}{2}\varsigma (3l_{i}l_{j}+ h_{ij}).
\end{equation}
where $\varsigma=d_{ij}l^{i}l^{j}$.

\medskip
\noindent
\textbf{TT-tensors on $\mathbb{R}^3$.} The general form of symmetric,
tracefree and divergence-free tensors (i.e. \emph{TT-tensors}) in a
conformally flat setting are well-known ---see
e.g. \cite{BeiOMu96,DaiFri01}. For convenience of the reader, in this
short paragraph, we adapt the conventions and discussion given in the
latter references to the present setting. The general the solutions to the equation
\begin{equation}\label{Divergence}
\grave{D}^{\bmi}\grave{d}_{\bmi\bmj}=0,
\end{equation}
 where $\grave{\bmh} \equiv-\bmdelta$ is the flat metric has been given in
 \cite{DaiFri01}.  One can introduce Cartesian coordinates $(x^{\bmalpha})$
 with the origin of $\mathbb{R}^3$ located at a fiduciary position
 $\mathcal{O}$. Additionally, we introduce polar coordinates defined via
 $\rho=\delta_{\bmalpha\bmbeta}x^{\bmalpha}x^{\bmbeta}$.  The flat metric
 in these coordinates
 reads
\begin{equation}\label{FlatMetric}
\grave{\bmh} = -\mathbf{d}\rho \otimes \mathbf{d}\rho -\rho^2 \bmsigma.
\end{equation}
Using this notation and taking into account the requirement of
spherical symmetry encoded in equation \eqref{KIDcondition} the flat
space counterpart of the required solution is
\[
\grave{\bmd} = \frac{A_{\star}}{\rho^3}\left(3 \mathbf{d}\rho \otimes
\mathbf{d}\rho + \grave{\bmh} \right),
\]
where $A_{\star}$ is a constant. In order to obtain an analogous
solution in conformally related 3-manifolds one can exploit the conformal
properties of equation \eqref{Divergence} using the following:

\begin{lemma}\label{Lemma:TTtensorTransformation}
Let $\bar{d}_{ij}$ be a tracefree symmetric solution to
$\bar{D}^{i}\bar{d}_{ij}=0$ where $\bar{D}$ is the
Levi-Civita connection of $\bar{\bmh}$. Let $\bmh=\omega^2\bar{\bmh}$,
then $d_{ij}=\omega^{-1}\bar{d}_{ij}$ is a symmetric
tracefree solution to $D^{i}d_{ij}=0$ where $D$ is the
Levi-Civita connection of $\bmh$.
\end{lemma}
This lemma can be found in \cite{DaiFri01}. Here we have adapted the
statement to agree with the conventions of this article.

\medskip
\noindent
\textbf{TT-tensors on $\mathbb{S}^3$ and
  $\mathbb{R}\times\mathbb{S}^2$.} One can exploit Lemma \ref{Lemma:TTtensorTransformation} to derive
spherically symmetric solutions of the divergence equation
\eqref{Divergence} in conformally flat 3-manifolds. In particular, the
metrics $\bmhbar$ and $\grave{\bmh}$ as given in equations
\eqref{MetricOnS3} and \eqref{FlatMetric} are related via
\[
\bmhbar=\omega^2 \grave{\bmh},
\]
where 
\begin{equation}\label{PlaneToSphere1}
\rho(\xi)= \cot(\xi/2), \qquad \omega(\xi)=2\sin^2(\xi/2),
\end{equation}  
The coordinate transformation $\rho(\xi)$ corresponds to the
stereographic projection in which the origin $\mathcal{O}$
of $\mathbb{R}^3$ is mapped to the South pole on $\mathbb{S}^3$.
Alternatively, one can also derive
 \begin{equation}\label{PlaneToSphere2}
\rho(\xi)=\tan(\xi/2), \qquad \omega(\xi)=2\cos^2(\xi/2),
\end{equation}
corresponding to the stereographic projection in which the
 origin of $\mathbb{R}^3$
is mapped to the North pole of $\mathbb{S}^3$. 
Using Lemma \ref{Lemma:TTtensorTransformation} with equations
 \eqref{PlaneToSphere1} or \eqref{PlaneToSphere2}
 one obtains 
\begin{equation} \label{TTtensorSphere}
\bm{\dbar}=\frac{A_{\star}}{2\sqrt{1-\omega^2(\xi)}}\left( 3\mathbf{d}\xi
 \otimes \mathbf{d}\xi + \bmhbar \right).
\end{equation}
Observe that $\dbar_{\bmi\bmj}$ is singular when $\omega(\xi)=1$ which
 corresponds to
$\xi=0$ and $\xi=\pi$ according to equations \eqref{PlaneToSphere1} 
and \eqref{PlaneToSphere2}, respectively.
Therefore, in this conformal representation the electric part of the
 rescaled Weyl tensor
is singular at the North and South poles of $\mathbb{S}^3$.
Proceeding in a analogous way as in the previous paragraphs
 one can observe that the metrics $\bmh$ and $\grave{\bmh}$
given in equations \eqref{metricTorus} and \eqref{FlatMetric}
are related via
\[
\bmh=\omega^2\grave{\bmh}
\]
where
\[
\rho(\psi)= e^{\psi}, \qquad \omega(\psi)=e^{-\psi}.
\]
A straightforward computation using Lemma \ref{Lemma:TTtensorTransformation}
 renders
\begin{equation} \label{TTtensorCylinder}
\bm{d}=A_{\star}\left( 3\mathbf{d}\psi \otimes \mathbf{d}\psi + \bmh \right).
\end{equation}
Moreover, since $D^{i}d_{ij}=3 A_{\star} D^{i}(l_{i}l_{j})$, it
follows that verifying that
$d_{ij}$ satisfies the condition \eqref{KIDcondition} reduces to the
computation of $\omega_i \equiv  \pounds_{\bmX}l_{i}$ and  showing that the
components of $\omega_i$ along any leg of the frame vanishes ---that
is
\[
l^i\omega_i =0, \qquad m^i\omega_i=0, \qquad
\bar{m}^i\omega_i=0.
\]
The latter can easily be done using the Killing equation
$\pounds_{\bmX}h_{ij}=2D_{(i}X_{j)}=0$ along with equations
\eqref{metricTorusFrame} and \eqref{NormalisationConditions}. Finally,
comparing expression \eqref{TTtensorCylinder} with equation
\eqref{ComparisonInitialDatumAstar} we can recognise that
$A_{\star}=m$.  Observe that this identification is irrespective of
the extrinsic curvature of $\mathscr{I}$.

\subsection{Asymptotic initial data for the Schwarzschild-de Sitter spacetime}
\label{sec:InitialData}

In the last section it was shown that the
$\mathbb{R}\times \mathbb{S}^2$-conformal  representation 
 leads to regular asymptotic data for the rescaled Weyl tensor. In this
  section we complete the discussion the asymptotic initial data for the
Schwarzschild-de Sitter spacetime 
in this conformal representation.  
To do so, we make use of the procedure to solve the
  conformal constraints at the conformal boundary as discussed in
  Section \ref{Section:ConstraitsAtScri} and the specific properties
  of the Schwarzschild-de Sitter spacetime.  

\subsubsection{Initial data for the Schouten tensor}
\label{section:intialDataSchoutenFrame}

Computing the Schouten tensor $\Schouten[\bmh]$ of $\bmh$ we get that
\[
\Schouten[\bmh]= -\frac{1}{2}\mathbf{d}\psi \otimes \mathbf{d}\psi + \frac{1}{2}\bmsigma.
\] 
 Equivalently, in abstract index notation one writes
\[
l_{ij}=-l_{i}l_{j} - \frac{1}{2}h_{ij}.
\]
 Thus, recalling the solution
to the conformal constraints given in equation
\eqref{SolutionContraintsGeneral2} we get,
\[
L_{ij}=-l_il_j-\frac{1}{2}(1-\kappa^2)h_{ij}. 
\]

\subsubsection{Initial data for the connection coefficients}
\label{section:initialdataFrameConnection}

In order to compute the connection coefficients associated with the
coframe $\{ \bmomega_\bmi \}$ recall that
$\bmomega^3=\mathbf{d}\psi$ and $\{\bmomega^1, \bmomega^2\}$ are
$\bmsigma$-orthonormal. Equivalently, one has that $\{ \bme_\bmi \}
=\{\bmpartial_{\psi},\bme_{\bm1}, \bme_{\bm2} \}$ with 
\[
\bme_{\bm1}=
\frac{1}{\sqrt{2}}(\bmm + \bar{\bmm}), \hspace{1cm}\bme_{\bm2}=
\frac{\mbox{i}}{\sqrt{2}}(\bmm -\bar{\bmm}),
\] 
where $\bmsigma= \bmm \otimes \bar{\bmm} + \bar{\bmm} \otimes \bmm$, so that
\[ 
\bmh =  -\bm\omega^1 \otimes \bm\omega^1 -\bm\omega^2 \otimes \bm\omega^2 
-\bm\omega^3 \otimes \bm\omega^3. 
\]

The connection coefficients can be obtained using the first structure
equation \eqref{FirstCartanStructureEquationFrame} given in Appendix
\ref{CartanFrame}. Proceeding in this manner, by a straightforward
computation, one can show that the only non-zero connection coefficient
is $\gamma_{2}{}^{2}{}_{1}$. In terms of the Ricci-rotation
coefficients, the latter corresponds to
$2\sqrt{2}\operatorname{Re}(\alpha_{\star})$ where $\alpha_{\star}
=-\frac{1}{2}\bar{m}^a\bar{\delta}m_{a}$ in the standard NP
notation ---see \cite{Ste91}. Therefore, the only no-trivial initial data for
the connection coefficients is \[ \gamma_{2}{}^{2}{}_{1}=
\sqrt{2}(\alpha_{\star}+ \bar{\alpha}_{\star}).
\]

\medskip
\noindent
\textbf{Remark 4.} The frame over the cylinder $\mathbb{R}\times
\mathbb{S}^2$ introduced in this section is not a global
one. Nevertheless, it is possible to construct an atlas covering
$\mathbb{R}\times\mathbb{S}^2$ such that one each of the charts one
has a well defined frame of the required form.

\subsubsection{Spinorial initial data}
\label{Sec:SpinorInitialData}

 In this section we discuss the spinorial counterpart of the
 asymptotic initial data computed in the previous sections.

\subsubsection{Spin connection coefficients}\label{Section:Spinorialconnection}

The spinorial counterpart of the asymptotic initial data constructed
in the previous sections is readily obtained by suitable
contraction with the spatial Infeld-van der Waerden symbols ---see Appendix
 \ref{CartanSpacespinor}. 
Following the discussion of Section
\ref{section:initialdataFrameConnection}, let $\bmomega^3=
 \mathbf{d}\psi$ and let $\{ \bmomega^1, \bmomega^2 \}$ denote an
 $\bmsigma$-orthonormal coframe. Using equations
 \eqref{Infeldxyz2} of Appendix \ref{CartanSpacespinor} we have that
 the spinorial coframe is given by
\begin{equation}
\label{spinorCoframexyz}
 \bmomega^{\bmA\bmB}= \sigma_{\bmi}{}^{\bmA \bmB}\bmomega^{\bmi} =
 (y^{\bmA\bmB} + z^{\bmA \bmB})\bmomega^1 + \mbox{i}(y^{\bmA
   \bmB}-z^{\bmA \bmB})\bmomega^2 - x^{\bmA \bmB}\bmomega^3.
\end{equation}
Alternatively, one has that the \emph{spinorial frame} is given by 
\[
\bme_{\bmA \bmB}=x_{\bmA \bmB}\bme_{x}{}^3\bmpartial_{\psi} + \sqrt{2}
y_{\bmA \bmB}e_{y}{}^{+} \bar{\bmm}^{\flat}+ \sqrt{2} z_{\bmA
  \bmB}e_{z}{}^{-}\bmm^{\flat}
\]
 where  $\bme_{x}{}^{3}$, $\bme_{y}{}^{+}$, $\bme_{z}{}^{-}$ denote the
 only non-vanishing frame coefficients. Equation \eqref{spinorCoframexyz} allow
us to compute the reduced connection coefficients
$\gamma_{\bmA}{}^{\bmB}{}_{\bmC\bmD}$ using the first Cartan structure
equation \eqref{FirstCartanStructureEquationSpinor} in Appendix
\ref{CartanSpacespinor}. Alternatively, one can
use the results of Section \ref{section:initialdataFrameConnection} and the 
spatial Infeld-van der Waerden symbols to compute
\[
\gamma_{\bmA \bmB}{}^{\bmC \bmD}{}_{\bmE \bmF} \equiv \gamma_{\bmi}{}^{\bmj}{}_{\bmk}
\sigma_{\bmA \bmB}{}^{\bmi}\sigma^{\bmC \bmD}{}_{\bmj} \sigma_{\bmE
  \bmF}{}^{\bmk},
\] 
where
\[
\gamma_{\bmi}{}^{\bmj}{}_{\bmk}=
\delta_{\bmi}{}^{2}\delta_{1}{}^{\bmj}\delta_{\bmk}{}^{2}\gamma_{2}{}^{1}{}_{2}
+\delta_{\bmi}{}^{2}\delta_{2}{}^{\bmj}\delta_{\bmk}{}^{1}\gamma_{2}{}^{2}{}_{1},
 \]
with 
\[
\gamma_{2}{}^{1}{}_{2}=-\sqrt{2}(\alpha_{\star}+
\bar{\alpha}_{\star}), \qquad \gamma_{2}{}^{2}{}_{1}=
\sqrt{2}(\alpha_{\star}+ \bar{\alpha}_{\star}).
\] 
Using the identities \eqref{Infeldxyz1}-\eqref{Infeldxyz2} in Appendix \ref{CartanSpacespinor} one obtains
\[ 
\gamma_{\bmA \bmB}{}^{\bmC \bmD}{}_{\bmE \bmF}= 2 \sqrt{2}(\alpha_{\star}+ \bar{\alpha}_{\star}) 
(y_{\bmA \bmB}-z_{\bmA \bmB})
(y_{\bmE \bmF}z^{\bmC \bmD}-y^{\bmC \bmD}z_{\bmE \bmF}). 
\]
Thus, the reduced connection coefficients are given by
\begin{equation} 
\label{ReducedConnection}
 \gamma_{\bmA \bmB}{}^{\bmD}{}_{\bmF} \equiv \frac{1}{2}\gamma_{\bmA
   \bmB}{}^{\bmC \bmD}{}_{\bmC \bmD} = (\alpha_{\star}+ \bar{\alpha}_{\star})
 x^{\bmD}{}_{\bmF}(y_{\bmA \bmB}-z_{\bmA \bmB}).
\end{equation}
By computing the spinor version of the \emph{connection form}
$\bmgamma^{\bmD}{}_{\bmF} \equiv \gamma_{\bmA
\bmB}{}^{\bmD}{}_{\bmF}\bmomega^{\bmA \bmB}$ using equations
\eqref{ReducedConnection} and \eqref{spinorCoframexyz} one can readily
verify that the first structure equation is satisfied. Additionally, using
the reality conditions, 
\[
x_{\bmA\bmB}{}^{\dagger}=-x_{\bmA \bmB}, \qquad y_{\bmA
  \bmB}{}^{\dagger}=z_{\bmA \bmB}, \qquad z_{\bmA
  \bmB}{}^{\dagger}=y_{\bmA \bmB}
\]
 we can verify that $\gamma_{\bmA \bmB \bmC \bmD}$ is an imaginary
spinor ---as is to be expected from the space spinor formalism. The
field  $\gamma_{\bmA \bmB \bmC \bmD}$ represents the initial data for
the field $\xi_{\bmA \bmB \bmC \bmD}$ ---the imaginary part of the
reduced connection coefficient  $\Gamma_{\bmA \bmB \bmC\bmD}$.  The
real part of $\Gamma_{\bmA \bmB \bmC\bmD}$ corresponds to the Weingarten spinor
$\chi_{\bmA \bmB \bmC \bmD}$ which, in accordance with equation
\eqref{SolutionContraintsGeneral1}, is given initially by
\[ 
\chi_{\bmA \bmB \bmC \bmD}= \kappa h_{\bmA \bmB \bmC \bmD}.
\]
 Rewriting the reduced connection coefficients
\eqref{ReducedConnection} in terms of the basic valence-4 spinors
introduced in Section \ref{sec:SphericalSymmetrySpinors} we get for
$\xi_{\bmA \bmB \bmC \bmD}=\gamma_{\bmA \bmB \bmC \bmD}$ the explicit expression
\begin{eqnarray*}
&& \xi_{\bmA \bmB \bmC \bmD}= -(\alpha_{\star}+ 
\bar{\alpha}_{\star})
(\epsilon^{1}{}_{\bmA \bmB \bmC \bmD} + \epsilon^{3}{}_{\bmA \bmB \bmC
  \bmD}) \\
&& \hspace{3cm}+ \frac{1}{2\sqrt{2}}(\alpha_{\star}+ \bar{\alpha}_{\star})\epsilon_{\bmA \bmC}
(y_{\bmB \bmD} + z_{\bmB \bmD}) + \frac{1}{2\sqrt{2}}(\alpha_{\star}+ \bar{\alpha}_{\star})
\epsilon_{\bmB \bmD} (y_{\bmA \bmC}+ z_{\bmA \bmC}). 
\end{eqnarray*}

\subsubsection{Spinorial counterpart of the Schouten tensor}
\label{Section:InitialDataSpinorSchouten}

The spinorial counterpart of the Schouten tensor $l_{ij}$ can be
directly read from the expressions in Section
\ref{section:intialDataSchoutenFrame}. Observe that the elementary
spinor $x^{\bmA \bmB}$ corresponds to the components of $l_{i}$ with
respect to the coframe \eqref{spinorCoframexyz} since
\[ 
\bmomega^{\bmA \bmB}x_{\bmA \bmB}=-x^{\bmA \bmB}x_{\bmA \bmB}\bmomega^3=
\bmomega^3=\mathbf{d}\psi = \bml. 
\]
 Replacing $h_{\bmi\bmj}$ by its space spinor
counterpart $h_{\bmA \bmB \bmC \bmD}$ we obtain
\[
l_{\bmi \bmj} \mapsto  l_{\bmA \bmB \bmC \bmD}= -x_{\bmA \bmB}x_{\bmC \bmD}
-\frac{1}{2}h_{\bmA \bmB \bmC \bmD}. 
\]
Equivalently, recalling that the space spinor counterpart of
the tracefree part of a tensor $l_{ \{ \bmi\bmj\}} \equiv l_{\bmi
  \bmj}-\tfrac{1}{3}lh_{\bmi \bmj}$ corresponds to the totally
symmetric spinor $l_{(\bmA \bmB \bmC \bmD)}$ it follows then from
\[ 
l_{\bmi \bmj}= l_{\{ \bmi \bmj\}} + \frac{1}{3} l h_{\bmi \bmj}, 
\]
that
 \[
l_{\bmA \bmB \bmC \bmD} =l_{(\bmA \bmB \bmC \bmD)} + 
\frac{1}{3}l h_{\bmA \bmB \bmC \bmD}.
\]
Thus, using that for the metric \eqref{metricTorus}  one has
$r=-2$ and that $l \equiv h^{\bmi\bmj}l_{\bmi\bmj}= \frac{1}{4}r$, it
follows that  $l= -\frac{1}{2}$ and $l_{(\bmA \bmB
  \bmC \bmD)}=-x_{(\bmA \bmB}x_{\bmC \bmD)}=-2\epsilon^{2}_{\bmA \bmB
  \bmC \bmD}$.  Therefore, we get
\begin{equation}
\label{intialSchoutenIntrinsic}
l_{\bmA\bmB \bmC \bmD}= -2\epsilon^{2}{}_{\bmA \bmB \bmC \bmD} -
\frac{1}{6}h_{\bmA \bmB \bmC \bmD}.
\end{equation}
 Finally, recalling the expressions for the components of the spacetime
 Schouten tensor given in \eqref{SolutionContraintsGeneral2} we conclude
\[
L_{\bmA \bmB \bmC \bmD}= -2\epsilon^{2}{}_{\bmA \bmB \bmC \bmD} 
- \frac{1}{6}(1-3\kappa^2)h_{\bmA \bmB \bmC \bmD}.
\]

\subsubsection{Initial data for the rescaled Weyl spinor}
\label{Section:InitialWeylSpinor}

Following the approach employed in last section, the spinorial
counterpart of \eqref{TTtensorCylinder} is given by
\[ 
d_{\bmA \bmB \bmC \bmD}= A_{\star}(3l_{\bmA \bmB}l_{\bmC \bmD} 
+ h_{\bmA \bmB \bmC \bmD}).
\]

However, the trace-freeness condition simplifies the last
expression since $d^{\bmi}{}_{\bmi}=0$ implies that $d_{\bmi \bmj}=d_{
  \{ \bmi \bmj\}}$. Therefore $d_{\bmA \bmB \bmC \bmD}= d_{(\bmA \bmB
  \bmC \bmD)}=3 A_{\star}l_{(\bmA \bmB}l_{\bmC \bmD)}$. As the
elementary spinor $x_{\bmA \bmB}$ can be  associated to the
components of $\bml$ respect to the co-frame
\eqref{spinorCoframexyz} one gets that 
\[ 
d_{\bmA \bmB \bmC \bmD}=3A_{\star}x_{(\bmA \bmB}x_{\bmC \bmD)}.
\]
This last expression can be equivalently written in terms of
the basic valence-4 space spinors of Section
\ref{sec:SphericalSymmetrySpinors} as
\[
\phi_{\bmA \bmB \bmC \bmD}= 6m\epsilon^{2}{}_{\bmA \bmB \bmC \bmD}. 
\]
where, in the absence of a magnetic part, we
 have identified $\phi_{\bmA \bmB \bmC \bmD}$
initially with $d_{\bmA \bmB \bmC \bmD}$. Observe that  have set 
$A_{\star}=m$ consistent with the discussion of Section
  \ref{Section:InitialDataRescaledWeylFrame}.

\section{The solution to the asymptotic initial value problem for the 
Schwarzschild-de Sitter spacetime and perturbations}

As already discussed in the introductory section, recasting explicitly the
Schwarzschild-de Sitter spacetime as a solution to the system of
conformal evolution equations
\eqref{EvolutionEquation3Plus1Decomposition1}-\eqref{EvolutionEquation3Plus1Decomposition9}
requires solving, in an explicit manner, the conformal geodesic
equations. This, as discussed in Appendix
  \ref{eSdS:AsymptoticPoints}, is not possible in general. Instead,
an alternative approach is to study directly the conformal evolution
equations
\eqref{EvolutionEquation3Plus1Decomposition1}-\eqref{EvolutionEquation3Plus1Decomposition9}
making explicit the spherical symmetry of the solution and the
asymptotic initial data corresponding to the Schwarzschild-de Sitter
spacetime. This approach does not only extract the required
information about the reference solution ---in the conformal Gaussian
gauge--- but, in addition, is a model for the general structure of the
conformal evolution equations. The relevant analysis is
  discussed in Sections \ref{sec:SphericalSymmetrySpinors} and
  \ref{sec:CoreAnalysis}.   As a complementary analysis,
  we study the the formation of singularities in the evolution
  equations.  In order to have a more compact discussion leading to the Main
  Result, the analysis of the formation of singularities is presented
  in  Appendix
  \ref{FormationOfSingularitiesAndReparametrisations}. Finally, in Section \ref{PertsOfeSdS}, we use the
  theory of symmetric hyperbolic systems contained in \cite{Kat75} to
  obtain a existence and stability result for the development of small
  perturbations to the asymptotic initial data of the Schwarzschild-de
  Sitter spacetime.

\subsection{The spherically symmetric evolution equations}
\label{sec:SphericalSymmetrySpinors}

Hitherto, the discussion of the extended conformal Einstein field
equations and the conformal constraint equations has been completely
general. Since we are interested in analysing the Schwarzschild-de
Sitter spacetime as a solution to the conformal field equations one
has to incorporate specific properties of this spacetime. The most
important assumption for our analysis is that of the \emph{spherical
  symmetry} of the spacetime. Under this assumption, a generalisation
of Birkhoff's theorem for vacuum spacetimes with de Sitter-like
Cosmological constant shows that the spacetime must be locally
isometric to either the Nariai or the Schwarzschild-de Sitter
solutions ---see \cite{Sta98}. As the Nariai solution is known to not
admit a smooth conformal boundary \cite{Bey09a,Fri14b}, then the
formulation of an asymptotic initial value problem readily selects the
Schwarzschild-de Sitter spacetime.

\medskip
To incorporate the assumption of spherical symmetry
 into the conformal field equations encoded in the spinorial zero-quantities
\eqref{spacetimeSpinorXCEFE1}-\eqref{spacetimeSpinorXCEFE4} one has to
reexpress the requirement of spherical symmetry in terms of the
space spinor formalism. In order to ease the presentation we simply
introduce a consistent 
Ansatz for spherical symmetry ---a similar approach has been
taken in \cite{LueVal14a}. More precisely, we set 
\begin{subequations}
\begin{eqnarray}
&&\phi_{\bmA \bmB \bmC \bmD} = \phi_{2}\,\epsilon^{2}{}_{\bmA \bmB \bmC \bmD}, \label{AnsatzFirst}\\ 
&&  \Theta_{\bmA \bmB} =   \sqrt{2}\Theta_{x}{}^{T}\,x_{\bmA \bmB}, \\
&&  \Theta_{\bmA \bmB \bmC \bmD}=   \Theta_{2}{}^{S}\,
\epsilon^{2}{}_{\bmA \bmB \bmC \bmD} +
  \frac{1}{3}\Theta_{h}{}^{S}\,h_{\bmA \bmB \bmC \bmD} ,\\ 
&& \xi_{\bmA \bmB \bmC \bmD}=\xi_{1}\,\epsilon^{1}{}_{\bmA
    \bmB \bmC \bmD} + \xi_{2}\,\epsilon^{2}{}_{\bmA \bmB \bmC \bmD} +
  \xi_{3}\,\epsilon^{3}{}_{\bmA \bmB \bmC \bmD}  + \frac{1}{3}\xi_{h} \,h_{\bmA
    \bmB \bmC \bmD} \nonumber \\ 
&&\hspace{2cm}+\frac{\xi_x}{\sqrt{2}}(x_{\bmB \bmD} \epsilon_{\bmA
      \bmC} +x_{\bmA \bmC}\epsilon_{\bmB \bmD})
 + \frac{\xi_y}{\sqrt{2}}(y_{\bmB
      \bmD}\epsilon_{\bmA \bmC} +y_{\bmA \bmC}\epsilon_{\bmB
      \bmD}) \nonumber \\ 
&& \hspace{2cm} + \frac{\xi_z}{\sqrt{2}}(z_{\bmB \bmD}\epsilon_{\bmA
      \bmC} +z_{\bmA \bmC}\epsilon_{\bmB \bmD}),\\ 
&&\chi_{\bmA \bmB \bmC \bmD}=\chi_{2}\,\epsilon^{2}_{\bmA \bmB \bmC
    \bmD} + \frac{1}{3}\chi_h\, h_{\bmA \bmB \bmC \bmD}, \label{AnsatzChi}\\ 
&& e^{0}{}_{\bmA \bmB}=e^{0}_{x}\,x_{\bmA \bmB}, \qquad e^{3}_{\bmA \bmB}=
  e^{3}_{x}\,x_{\bmA \bmB}, \qquad e^{+}_{\bmA \bmB}=
  e^{+}_{y}\,y_{\bmA \bmB}, \qquad e^{-}_{\bmA \bmB}=
  e^{-}_{z}\,z_{\bmA \bmB}, \\
&& f_{\bmA \bmB}=f_{x}\,x_{\bmA \bmB}, \\
&&d_{\bmA \bmB}=d_{x}\,x_{\bmA\bmB}. \label{AnsatzLast}
\end{eqnarray}
\end{subequations}
The \emph{elementary spinors} $x_{\bmA\bmB}$, $y_{\bmA\bmB}$,
$z_{\bmA\bmB}$, $\epsilon^2_{\bmA\bmB\bmC\bmD}$ and
$h_{\bmA\bmB\bmC\bmD}$ used in the above Ansatz are defined in
Appendix \ref{Appendix:SpaceSpinorFormalism}.  For further details on
the construction of a general spherically symmetric Ansatz see
\cite{Fri98c,Val12}.  Alternatively, one can follow a procedure
similar to that of Section \ref{Section:Spinorialconnection} ---by
writing a consistent spherically symmetric Ansatz for the orthonormal
frame one can identify the non-vanishing components of the required
tensors. The transition to the spinorial version of such Ansatz can be
obtained by contracting appropriately with the Infeld-van der Waerden
symbols taking into account equations
\eqref{Infeldxyz1}-\eqref{Infeldxyz2},
\eqref{xyzProducts1}-\eqref{xyzProducts4} and
\eqref{UsefulIdentities1}-\eqref{UsefulIdentities3}.

\medskip

The Ansatz for spherical symmetry encoded in equations
\eqref{AnsatzFirst}-\eqref{AnsatzLast} combined with the evolution
equations
\eqref{EvolutionEquation3Plus1Decomposition1}-\eqref{EvolutionEquation3Plus1Decomposition9}
leads, after suitable contraction with the elementary spinors
introduced in Section \ref{sec:SphericalSymmetrySpinors}, to a set of
evolution equations for the fields
 \[ 
\phi_{2}, \; \Theta_{x}{}^{T},\;
\Theta_{2}{}^{S},\; \Theta_{h}{}^{S},\xi_{1},\;  \xi_3
,\;  \xi_{x},
\; \xi_{y},\; \xi_{z},\; 
 \; e_{x}^{0}, \; e_{x}^{3}, \; e_{z}^{+},
 e_{y}^{-},\; f_{x}.
\]
 This lengthy computation has been
 carried out using the suite {\tt xAct } for tensor
 and spinorial manipulations  in {\tt Mathematica} ---see
 \cite{GarMar12}. At the end of the day one obtains the  following
 evolution equations:
\begin{subequations}
\begin{eqnarray} 
&&  \partial_{\tau}e_x^0 =\tfrac{1}{3} \chi_2 e_x^0 -  \tfrac{1}{3} \chi_h e_x^0 -  f_x, \label{EvEq1}\\
&& \partial_{\tau}e_x^3 = \tfrac{1}{3} \chi_2 e_x^3 -  \tfrac{1}{3} \chi_h e_x^3, \label{EvEq2} \\
&& \partial_{\tau}e_y^+ =  - \tfrac{1}{6} \chi_2 e_y^+ -  \tfrac{1}{3} \chi_h e_y^+, \label{EvEq3} \\
&&  \partial_{\tau}e_z^-= - \tfrac{1}{6} \chi_2 e_z^- -  \tfrac{1}{3} \chi_h e_z^-, \label{EvEq4} \\
&& \partial_{\tau}f_x = \tfrac{1}{3} \chi_2 f_x -  \tfrac{1}{3} \chi_h f_x + \Theta_x^T, \label{EvEq5} \\
&& \partial_{\tau}\chi_2=\tfrac{1}{6} \chi_2^2 -  \tfrac{2}{3} \chi_2 \chi_h - \Theta_2^S - \Theta \phi_2, \label{EvEq6}\\
&&  \partial_{\tau}\chi_h= - \tfrac{1}{6} \chi_2^2 -  \tfrac{1}{3} \chi_h^2 -  \Theta_h^S, \label{EvEq7}\\
&& \partial_{\tau}\xi_3 = \tfrac{1}{12} \chi_2 \xi_3 -  \tfrac{1}{3} \chi_h \xi_3 -  \tfrac{1}{2} \chi_2 \xi_y,\label{EvEq8} \\
&& \partial_{\tau}\xi_1 = \tfrac{1}{12} \chi_2 \xi_1 -  \tfrac{1}{3} \chi_h \xi_1 - \tfrac{1}{2} \chi_2 \xi_z,\label{EvEq9} \\
&& \partial_{\tau}\xi_x =  - \tfrac{1}{2} \chi_2 f_x -  \Theta_x^T -  \tfrac{1}{6} \chi_2 \xi_x - \tfrac{1}{3} \chi_h \xi_x, \label{EvEq10}\\
&& \partial_{\tau}\xi_y = - \tfrac{1}{8} \chi_2 \xi_3 + \tfrac{1}{12} \chi_2 \xi_y -  \tfrac{1}{3} \chi_h \xi_y,\label{EvEq11} \\
&&  \partial_{\tau}\xi_z =- \tfrac{1}{8} \chi_2 \xi_1 + \tfrac{1}{12} \chi_2 \xi_z -  \tfrac{1}{3} \chi_h \xi_z, \label{EvEq12}\\
&&  \partial_{\tau}\Theta_x^T = \tfrac{1}{3} \chi_2 \Theta_x^T -  \tfrac{1}{3} \chi_h \Theta_x^T  + \tfrac{1}{3} d_x \phi_2, \label{EvEq13}\\
&&   \partial_{\tau} \Theta_2^S =  \tfrac{1}{6} \chi_2 \Theta_2^S -  \tfrac{1}{3} \chi_h \Theta_2^S - 
 \tfrac{1}{3} \chi_2 \Theta_h^S + \dot{\Theta} \phi_2, \label{EvEq14}\\
&& \partial_{\tau}\Theta_h^S = - \tfrac{1}{6} \chi_2 \Theta_2^S -  \tfrac{1}{3} \chi_h \Theta_h^S, \label{EvEq15}
\\
&& \partial_{\tau}\phi_2 = - \tfrac{1}{2} \chi_2 \phi_2 - \chi_h  \phi_2.\label{EvEq16}
\end{eqnarray}
\end{subequations}

The results of the analysis of Sections
\ref{Section:Spinorialconnection}, \ref{Section:InitialDataSpinorSchouten}
and \ref{Section:InitialWeylSpinor} provide the asymptotic initial data
for the above spherically symmetric evolution equations. The resulting expressions
are collected in the following lemma:

\begin{lemma}
\label{Lemma:SdSAsymptotic InitialData}
There exists a conformal gauge in which asymptotic initial data for the 
Schwarzschild-de Sitter spacetime can be expressed, in terms of the
fields defined by the Ansatz \eqref{AnsatzFirst}-\eqref{AnsatzLast}, as 
\begin{align*}
 \phi_{2} = & 6m, & \Theta_{x}{}^{T} = & 0, &  \Theta_{2}{}^{S} = &-2 ,
&  \Theta_{h}{}^{S} = &-\frac{1}{2}(1-3\kappa^2), \\   \xi_{1} = &
 -(\alpha_{\star}+ \bar{\alpha}_{\star}), &
 \xi_3 = & -(\alpha_{\star}+ \bar{\alpha}_{\star}), & \xi_{x} = &
\frac{1}{2\sqrt{2}}(\alpha_{\star}+ \bar{\alpha}_{\star}), & 
  \xi_{y} = & \frac{1}{2\sqrt{2}}(\alpha_{\star}+ \bar{\alpha}_{\star}),
\\  \xi_{z} = &
\frac{1}{2\sqrt{2}}(\alpha_{\star}+ \bar{\alpha}_{\star}), & \chi_{2}= & 0,  &
 \chi_{h}= & 3\kappa,& \chi_{x}= & 0, \\ 
  e_{x}^{0}= & 0, & e_{x}^{3}= & 1, & e_{z}^{+} = & 1,
 & e_{y}^{-} = & 1, \\   f_{x} = & 0. 
\end{align*}
\end{lemma}

\subsection{The Schwarzschild-de Sitter spacetime
 in the conformal Gaussian gauge}
\label{sec:CoreAnalysis}

  In this section we analyse in some detail the spherically symmetric evolution
  equations derived in the previous section. In particular, we show
  that there is a subsystem of equations that decouples from the rest
  ---which we call the \emph{core system}--- and controls the
  essential dynamics of the system \eqref{EvEq1}-\eqref{EvEq16}.

  As  the Schwarzschild-de Sitter spacetime possess a curvature
  singularity at $r=0$, one expects, in general, the conformal
  evolution equations to develop singularities. Moreover, since the
  two essential parameters appearing in the initial data 
  given in Lemma \ref{Lemma:SdSAsymptotic InitialData} are $m$ and
  $\kappa$ ---the function $\alpha_{\star}$ only encodes the connection on
  $\mathbb{S}^2$--- one expects, in general, that the congruence of
  conformal geodesics reaches the curvature singularity at
  $\tau=\tau_{\lightning}(m,\kappa)$. 
 Nevertheless, numerical evaluations
  suggest that for $\kappa=0$ the core system does not develop any
  singularity ---observe that this is consistent with the remark made
  in the discussion of Section \ref{PhysicalParameters}. Furthermore,  an estimation for the time of
  existence $\tau_{\circledcirc}$ of the solution to the conformal
  evolution equations \eqref{EvEq1}-\eqref{EvEq16} with initial data
  in the case $\kappa=0$ is given.
  A discussion of the mechanism for the formation of
  singularities in the core system ($\kappa \neq 0$) and the role of
  the parameter $\kappa$ is given in Appendix
  \ref{FormationOfSingularitiesAndReparametrisations}.  

\subsubsection{The core system}
Inspection of the system \eqref{EvEq1}-\eqref{EvEq16} reveals that
there is a subsystem of equations that decouple from the rest.  In the
sequel we will refer to these equations as the \emph{core
  system}. Defining the fields
\begin{equation}
\chi \equiv \frac{1}{3}\bigg(\frac{1}{2}\chi_{2}+ \chi_{h}\bigg),
 \hspace{1cm} L \equiv -\frac{1}{3}\bigg(\frac{1}{2}\Theta_{2}{}^{S} +
 \Theta_{h}{}^{S}\bigg),
 \hspace{1cm} \phi \equiv \frac{1}{3}\phi_{2}, \label{coreVar}
\end{equation}
the system \eqref{EvEq16}-\eqref{EvEq1} can be shown to imply the
equations
\begin{subequations}
\begin{eqnarray}\label{eq:CoreSystem}
&& \dot{\phi}=-3\chi \phi, \label{Core1} \\ && \dot{\chi} =-\chi^2 + L
  - \frac{1}{2}\Theta \phi, \label{Core2} \\ && \dot{{L}} =-\chi
  L-\frac{1}{2}\dot{\Theta}\phi,\label{Core3}
\end{eqnarray} 
\end{subequations} 
where the overdot denotes differentiation with respect to $\tau$ and
 \[
\Theta(\tau)= \sqrt{\frac{|\lambda|}{3}}\tau\left(1+\frac{1}{2} \kappa
\tau\right), \qquad \text{ } \qquad
\dot{\Theta}=\sqrt{\frac{|\lambda|}{3}}(1+ \kappa \tau).
\]  

The initial data for this system is given by
 \begin{equation}
\label{InitialDataCoreSpherical}
\phi(0)=2m, \qquad \qquad \chi(0) = \kappa, \qquad \qquad L(0) =
\frac{1}{2}(1-\kappa^2).
\end{equation}

As it will be seen in the remainder of this article, equations
\eqref{Core1}-\eqref{Core3} with initial data
\eqref{InitialDataCoreSpherical} govern the dynamics of the complete
system \eqref{EvEq1}-\eqref{EvEq16}. The evolution of the remaining
fields can be understood once the core system has been investigated.


\subsubsection{Analysis of the Core System}
\label{AnalysisOfTheCoreSystem}

This section will be concerned with an analysis of the initial value
problem for the core system \eqref{Core1}-\eqref{Core3} with initial
data given by \eqref{InitialDataCoreSpherical}.  As it will be seen in the
following, the essential feature driving the dynamics of the core system
\eqref{Core1}-\eqref{Core3} is the fact that the function $\chi$
satisfies a Riccati equation coupled to two further fields. One also
has the following:

\medskip
\noindent 
\textbf{Observation 1.}  The core equation \eqref{Core1} can be
formally integrated to yield

\begin{equation}
\label{FormalSolutionE}
\phi(\tau)= 2m \exp\left({-3\int^\tau _0 \chi(\mbox{s})\mbox{d}\mbox{s}}\right).
\end{equation}
Hence, $\phi(\tau)>0$ if $m \neq 0$.

\medskip
In the remaining of this section, we analyse the
behaviour of the core system in the case where the extrinsic curvature
of $\mathscr{I}$ vanishes.

\medskip
As discussed in Section \ref{AsymptoticInitialValueProblem} in the
 case $\kappa=0$  the conformal factor reduces to
$\Theta(\tau)=\sqrt{|\lambda|/3}\tau$ ---thus,  one
has only one root corresponding to the initial
hypersurface $\mathscr{I}$. To simplify the
notation recall that $\dot{\Theta}_{\star}=\sqrt{|\lambda|/3}$ 
so that $\Theta(\tau)= \dot{\Theta}_{\star}\tau$. Accordingly,  the core
system \eqref{Core1}-\eqref{Core3} can be rewritten as
\begin{subequations}
\begin{eqnarray}
&& \dot{\phi}=-3\chi \phi, \label{CoreSym1} \\ 
&& \dot{\chi} =-\chi^2 + L - \frac{1}{2}\dot{\Theta}_{\star}\tau
  \phi, \label{CoreSym2} \\ 
&& \dot{{L}} =-\chi
  L-\frac{1}{2}\dot{\Theta}_{\star} \phi.\label{CoreSym3}
\end{eqnarray} 
\end{subequations} 
Moreover,  the initial data  reduces to 
\[
\chi(0)=0, \qquad  L(0)=\frac{1}{2}, \qquad\phi(0)=2m.
\] 

\medskip
\noindent
\textbf{Observation 2.} A direct inspection shows that equations
\eqref{CoreSym1}-\eqref{CoreSym3} imply that
\[
\chi(\tau)=\tau L(\tau).
\]
 This relation can be easily verified by direct substitution into
 equations \eqref{CoreSym2} and \eqref{CoreSym3}. Observe that
 $L(\tau)=\chi(\tau)/\tau$ is well defined at $\mathscr{I}$ where
 $\tau=0$ and $\chi(0)=0$ since the initial conditions ensure that
\[
\lim_{\tau \rightarrow 0}\frac{\chi(\tau)}{\tau}=\frac{1}{2}.
\]
  
\medskip
Taking into account the above observation the core system reduces to 
\begin{subequations}
\begin{eqnarray}
&& \dot{L}=-\tau L^2 -\frac{1}{2}\dot{\Theta}_{\star} \phi \label{CoreSymmetric1} \\ 
&& \dot{\phi}= -3 \tau L
  \phi \label{CoreSymmetric2}
\end{eqnarray} 
\end{subequations} 
 with initial data
\begin{equation}\label{intialDataTimeSym}
L(0)=\frac{1}{2},   \qquad  \phi(0)=2m.
\end{equation}

\medskip
\noindent
\textbf{Observation 3.} One can integrate
 \eqref{CoreSymmetric2} to 
\begin{equation}
\label{phiPositiveTimeSymmetric} 
 \phi(\tau)=2m \exp\left(-\int_{0}^{\tau} \mbox{s}L(\mbox{s})\mbox{ds}\right) 
\end{equation}
 and conclude that $\phi(\tau)>0$ for $\tau>0$.

\medskip
To prove the boundedness of the solutions to the core system we begin
by proving some basic estimates:

\begin{lemma} 
\label{SymEstimateL}
If  $\kappa=0$, then the solution of \eqref{Core1}-\eqref{Core3} with initial data  \eqref{InitialDataCoreSpherical} satisfies the bound
 \[
L(\tau) \geq
\phi(\tau)\left(\frac{1}{4m}-\frac{1}{2}\dot{\Theta}_{\star}\tau\right) \qquad
\text{
for}\qquad \tau \geq 0.
\]
\end{lemma}

\begin{proof}
 Using equations \eqref{CoreSymmetric1} and \eqref{CoreSymmetric2} we
 obtain the  expression
\begin{equation}
\phi \dot{L}-L\dot{\phi}= 2 \tau L^2 \phi
-\frac{1}{2}\dot{\Theta}_{\star}\phi^2 \geq
-\frac{1}{2}\dot{\Theta}_{\star}\phi^2 \qquad \text{for} \qquad \tau
\geq 0. \label{eqQuotientRule}
\end{equation}

Since $\phi(\tau)>0$ we can consider
the derivative of ${L}/{\phi}$. Notice that
\[ 
\phi^2 \frac{\mbox{d}}{\mbox{d}\tau}\left(\frac{L}{\phi}\right)= \phi
\dot{L}-L\dot{\phi}.
\] 
This observation and inequality \eqref{eqQuotientRule}
gives
\[ 
\frac{\mbox{d}}{\mbox{d}\tau}\left(\frac{L}{\phi}\right) \geq 
-\frac{1}{2}\dot{\Theta}_{\star}  \qquad \text{for} \qquad \tau
\geq 0.
\]
 Integrating the last differential inequality from $\tau=0$
to $\tau>0$ taking into account the initial conditions leads to
\[
L(\tau) \geq \phi(\tau)\left(\frac{1}{4m}-\frac{1}{2}
\dot{\Theta}_{\star}\tau\right) \qquad \text{for} \qquad \tau
\geq 0. 
\]
\end{proof}

Observe that the last estimate ensures that $L(\tau)$ is
non-negative for $\tau \in [0,8m/\dot{\Theta}_{\star}]$.
It turns out that finding an upper bound for $L(\tau)$
 is relatively simple:

\begin{lemma}\label{lemmaUpperBoundLTimeSym}
If  $\kappa=0$ then, for the solution of \eqref{Core1}-\eqref{Core3}
 with initial data
  \eqref{InitialDataCoreSpherical}, one has that
\[
L(\tau) \leq \frac{2}{\tau^2+ 4} \qquad \mbox{for} \qquad \tau\geq 0.
\] 
\end{lemma}

\begin{proof}
Assume $\tau \geq 0$. Using that $\phi(\tau)>0$ and equation
\eqref{CoreSymmetric1} one obtains the differential inequality
\[
\dot{L}(\tau) \leq - \tau L^2(\tau).
\]
Using that $L(\tau)> 0$ for $\tau \geq 0$ one gets
\[ 
\frac{\dot{L}(\tau)}{L^2(\tau)} \leq -\tau.
\]
 The last expression can be integrated 
 giving an upper bound for $L(\tau)$:
\[
L(\tau)\leq \frac{2}{\tau^2+ 4}. 
\]
\end{proof}

A simple bound on a finite interval can be found for the field $\phi(\tau)$
as follows:
\begin{lemma}
\label{NoBlowUpForFiniteTime}
If  $\kappa=0$ then, for the solution of \eqref{Core1}-\eqref{Core3} with initial data  \eqref{InitialDataCoreSpherical} and for  $0 \leq \tau \leq 1/(2\sqrt[3]{\dot{\Theta}_{\star}m})$, the field $\phi(\tau)$ is bounded by above. 
\end{lemma}

\begin{proof}
Assume $\tau \geq 0$. From the estimate of Proposition
 \ref{SymEstimateL} one has that
\[
L \geq -\frac{1}{2}\dot{\Theta}_{\star}\tau\phi.
\]  
Therefore 
\[ 
-3 \tau L\phi \leq \frac{3}{2}\dot{\Theta}_{\star}\tau^2 \phi^2 .  
\]
Using equation \eqref{CoreSymmetric2} one obtains the 
differential inequality
\[ 
\dot{\phi} \leq   \frac{3}{2}\dot{\Theta}_{\star}\tau^2 \phi^2.
\]
 Since $\phi(\tau)>0$ the last expression can be integrated
 to yield,
\[
\phi(\tau) \leq \frac{2m}{1-\dot{\Theta}_{\star}m\tau^3}.
\]
 Therefore, for  $0<\tau < 1/\sqrt[3]{\dot{\Theta}_{\star}m}$,
the field, $\phi(\tau)$ is bounded by above. Consequently,
one can take  $0 \leq \tau \leq 1/(2\sqrt[3]{\dot{\Theta}_{\star}m})$.
\end{proof}

The results of Lemmas \ref{SymEstimateL},
  \ref{lemmaUpperBoundLTimeSym} and \ref{NoBlowUpForFiniteTime} can be
  summarised in the following: 

\begin{lemma}  
\label{LemmaBoundednessCoreSystem}
The solution to the core system \eqref{Core1}-\eqref{Core3} with
initial data \eqref{InitialDataCoreSpherical}, in the case $\kappa=0$,
is bounded for $0 \leq \tau \leq \tau_{\bullet}$, where
\begin{equation}
\tau_{\bullet} \equiv \textnormal{min} \;\Big\{
\frac{8m}{\dot\Theta_{\star}},
\frac{1}{2\sqrt[3]{\dot{\Theta}_{\star}m}} \Big\}.
\end{equation}
\end{lemma}

\begin{remark}
{\em A plot of the numerical evaluation of the solutions
to the core system \eqref{Core1}-\eqref{Core3} with initial data
  \eqref{InitialDataCoreSpherical} in the case $\kappa=0$ is shown
  in Figure \ref{Figure:Kappa0}.}
\end{remark}

\begin{figure}[t]
\begin{center}
\includegraphics[width=0.65\textwidth]{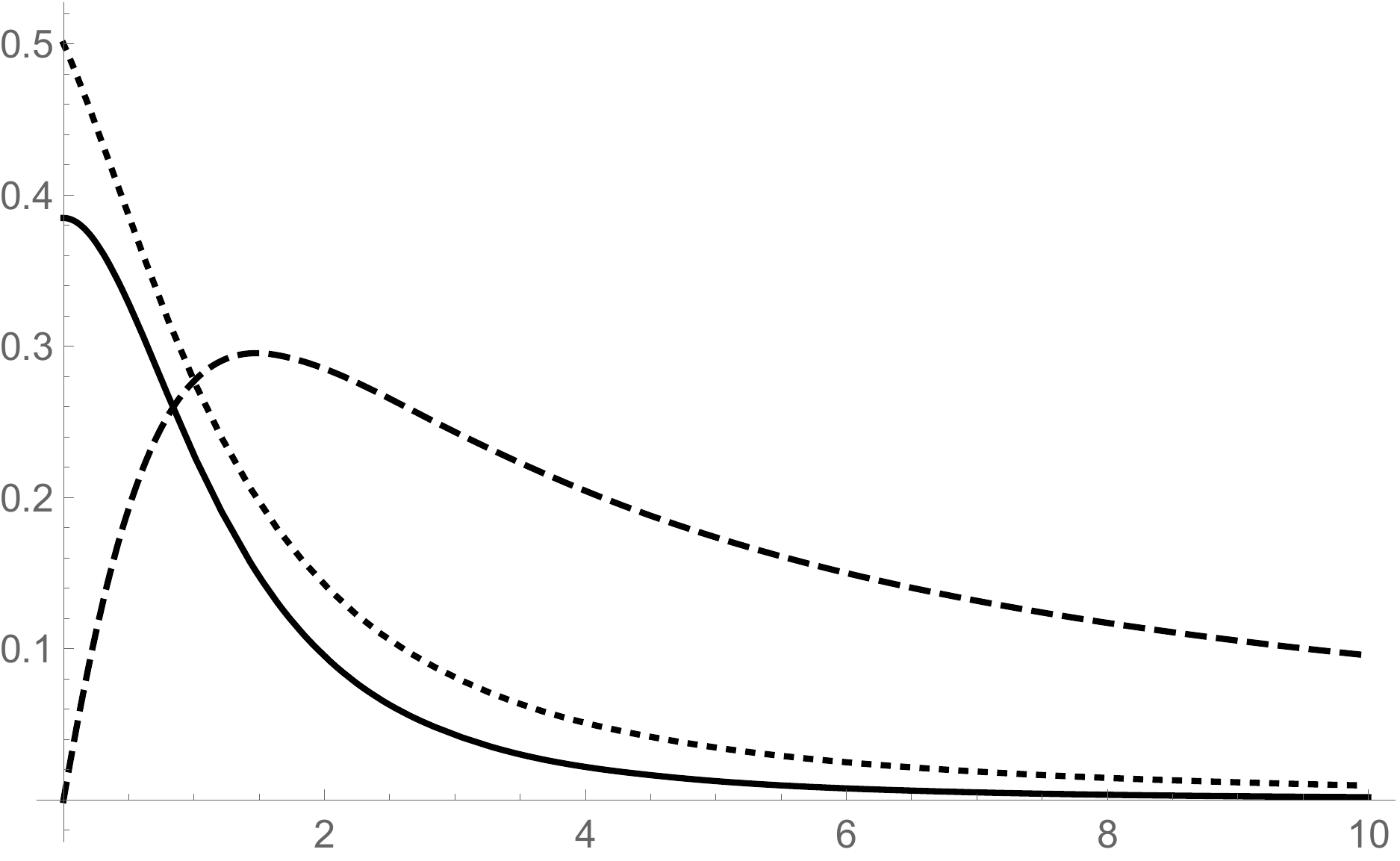}
\end{center}
\caption{Numerical solution of the core system in the $\kappa=0$ case
  with $|\lambda|=3$ and $m=1/3\sqrt{3}$.  The solid line corresponds
  to $\phi$, the dashed line to $\chi$ and the dotted line to
  $L$. Observe that in contrast to the $\kappa>1$ and $\kappa<-1$
cases,  numeric 
evaluations suggest that in the case $\kappa=0$
the fields of the core system are bounded for all times  ---see Figures \ref{fig:figPos} and \ref{fig:Negative} of Appendix
  \ref{FormationOfSingularitiesAndReparametrisations}.}
\label{Figure:Kappa0}
 \end{figure}

\subsubsection{Behaviour of the remaining fields in the
 conformal evolution equations}
\label{DiscussionConformalEvolution}

In this section we complete the analysis of the conformal evolution equations.
In particular, we show that 
the dynamics of the whole evolution equations is driven by the core
  system. To this end, we introduce the fields
\[
\bar{\chi}\equiv\frac{1}{3}(\chi_{2}-\chi_{h}), \qquad
 \bar{L}\equiv\frac{1}{3}(\Theta_{2}^{S}-\Theta_{h}^{S}). 
\]
 The
 evolution equations for these variables are
\begin{subequations}
\begin{eqnarray}
&&\dot{\bar{\chi}}=\bar{\chi}^2 -\bar{L}-\Theta \phi,\label{SuplementarySystem1} \\
&&\dot{\bar{L}} = \bar{\chi}\bar{L} + \dot{\Theta}\phi,\label{SuplementarySystem2}
\end{eqnarray}
\end{subequations}
with initial data
\[
\bar{\chi}(0)=-\kappa,\qquad \bar{L}(0)=-\frac{1}{2}(1+\kappa^2).
\]
Notice that despite these equations resemble those of the core system,
the field $\phi$ is not determined by the equations
\eqref{SuplementarySystem1}-\eqref{SuplementarySystem2} ---thus, we
call this subsystem the \emph{supplementary system}.  Once the core
system has been solved, $\phi$ can be regarded as a source term for
the system \eqref{SuplementarySystem1}-\eqref{SuplementarySystem2}.
If $\bar{\chi}$ and $\bar{L}$ are known then one can write the
remaining unknowns in quadratures. More precisely, defining
\begin{eqnarray*}
\xi_{y3}^{+}\equiv \xi_{y}+\frac{1}{2}\xi_{3}, \qquad
\xi_{y3}^{-}\equiv \xi_{y}-\frac{1}{2}\xi_{3}, \\ \xi_{z1}^{+}\equiv
\xi_{z}+\frac{1}{2}\xi_{1}, \qquad \xi_{z1}^{-}\equiv
\xi_{z}-\frac{1}{2}\xi_{1},
\end{eqnarray*}
one finds that the equations for these fields can be formally solved
to give
\begin{eqnarray*}
\xi_{y3}^{+}(\tau)=\xi_{y3}^{+}(0)\exp\left(-\int_{0}^{\tau}\chi(s)\mbox{ds}\right),
\qquad
\xi_{y3}^{-}=\xi_{y3}^{-}(0)\exp\left(-\int_{0}^{\tau}\bar{\chi}(s)\mbox{ds}\right),
\\ \xi_{z1}^{+}(\tau)=\xi_{z1}^{+}(0)\exp\left(-\int_{0}^{\tau}\chi(s)\mbox{ds}\right),
\qquad\xi_{z1}^{-}=\xi_{z1}^{-}(0)\exp\left(-\int_{0}^{\tau}\bar{\chi}(s)\mbox{ds}\right).
\end{eqnarray*}
The role of the the subsystem formed by $\Theta_{x}^{T}$, $f_{x}$ and $e_{x}^{3}$ 
is analysed in the following result.

\begin{lemma}
 \label{ForeSdSfVanishes}
Given asymptotic initial data for the  Schwarzschild-de Sitter
spacetime, if $\partial_{\psi}\kappa =0$ on $\mathscr{I}$ then
\[ f_{x}(\tau)=e_{x}{}^{0}(\tau)=\Theta_{x}{}^{T}(\tau)=0.\]
\end{lemma}

\begin{proof}
This result follows directly from equations \eqref{EvEq2},\eqref{EvEq5},
\eqref{EvEq13} and the initial data given in Lemma \ref{Lemma:SdSAsymptotic InitialData}.  To see this, first recall
that
 \[
d_{x} \equiv x^{\bmA \bmB} e_{\bmA \bmB}{}^{\bmi}\bme_{\bmi}
(\Theta) = e_{x}{}^{0}\partial_{0}\Theta +
e_{x}{}^{3}\partial_{3}\Theta. 
\]
 Assuming then that
$\bme_{\bm3}(\kappa)=0$ one has that $\bme_{\bm3}(\Theta)=0$ and therefore
\[
d_{x} =\sqrt{2}
e_{x}{}^{0}\dot{\Theta}.
\]
 Observing that equations
\eqref{EvEq2},\eqref{EvEq5}, \eqref{EvEq13} form an homogeneous system
of equations for the fields ${e_{x}{}^{0},f_{x},\Theta_{x}{}^{T}}$ with
vanishing initial data then, using a standard existence and
uniqueness argument for ordinary differential equations, it follows
 that the unique solution to this
subsystem is the trivial solution, namely
\[
f_{x}(\tau)=e_{x}{}^{0}(\tau)=\Theta_{x}{}^{T}(\tau)=0.
\]
\end{proof}
\noindent  Using the result of Lemma  \ref{ForeSdSfVanishes} 
one can formally integrate equation \eqref{EvEq10} to yield
\[\xi_{x}(\tau)= \xi_{x}(0)\exp\left( -\int_{0}^{\tau}\chi(s)\mbox{ds}\right). \]
The frame coefficients can also be found by quadratures
\begin{equation*} 
e_{x}^{0}(\tau)=e_{x}^{0}(0)\exp\left(\int_{0}^{\tau}\bar{\chi}(s)\mbox{ds}\right) ,
\qquad e_{y}^{+}(\tau)=e_{y}^{+}(0)\exp\left(-\int_{0}^{\tau}\chi(s)\mbox{ds}\right),
\end{equation*}
\begin{equation*}
\qquad e_{z}^{+}(\tau)=e_{z}^{+}(0)\exp\left(-\int_{0}^{\tau}\chi(s)\mbox{ds}\right).
\end{equation*}
Since we can write
\begin{align*}
\chi_{2} = & 2(\chi + \bar{\chi}), & \chi_{h}= & 2\chi-\bar{\chi}, &
\Theta_{2}^{S} = & 2(\bar{L}-L), & \Theta_{h}^{S} =
&-\bar{L}-2L,  \\ \xi_{y}= &\frac{1}{2}(\xi_{y3}^{+} +\xi_{y3}^{-}), &
\xi_{z}=&\frac{1}{2}(\xi_{z1}^{+} +\xi_{z1}^{-}), &
\xi_{1}=&2(\xi_{z1}^{+}-\xi_{z1}^{-}),
&\xi_{3}=&2(\xi_{y3}^{+}-\xi_{y3}^{-}).
\end{align*}
then, it only remains to study the behaviour of $\bar{\chi}$ and $\bar{L}$
to completely characterise the evolution equations
\eqref{EvEq1}-\eqref{EvEq16}.

\begin{remark}
{\em In the analysis of the core system of Appendix
  \ref{FormationOfSingularitiesAndReparametrisations} we identify the
mechanism for the formation of singularities at finite time
 in the case $\kappa
\neq 0$. Since $\phi$ acts as a source term for the supplementary
system \eqref{SuplementarySystem1}-\eqref{SuplementarySystem2} one expects
the solution to this system to be singular at finite time if the
solutions to the core system develop a singularity. Clearly, the
behaviour of the core system is independent from the behaviour of the
supplementary system. Consequently, the fact that $\phi$ diverges at
finite time or not is irrespective of the behaviour of $\bar{L}$ and
$\bar{\chi}$. }
\end{remark}

\subsubsection{Deviation equation for the congruence}
\label{DeviationEquationForTheCongruence}

As discussed in Section \ref{controllingthegauge}, the evolution
equations
\eqref{EvolutionEquation3Plus1Decomposition1}-\eqref{EvolutionEquation3Plus1Decomposition8}
are derived under the assumption of the existence of a
non-intersecting congruence of conformal geodesics. In this section we
analyse the solutions to the deviation equations.

\medskip
As a consequence of Lemma \ref{ForeSdSfVanishes} we have $f_{\bmA
  \bmB}=0$. Following the spirit of the space spinor formalism, the
deviation spinor $z_{\bmA \bmB}$ can be written in terms of elementary
valence 2 spinors as
\[
 z_{(\bmA \bmB)}=z_{x}x_{\bmA \bmB}+z_{y}y_{\bmA \bmB}+z_{z}z_{\bmA
   \bmB}.
\]
Substituting expression \eqref{AnsatzChi} into equation
\eqref{ConformalDeviationEquations2} and using the identities given in
equation \eqref{UsefulIdentities4} one obtains
\[
\partial_{\tau}z_{x}=0, \qquad \partial_{\tau}z_{z}=0, \qquad
\partial_{\tau}z_{y}
=-\frac{1}{12}\chi_{2}z_{y}-\frac{1}{6}\chi_{h}z_{y}.
\]
One can formally integrate these equations to obtain
\[
 z_{x}(\tau)=z_{x\star}, \qquad z_{z}(\tau)=z_{z\star}, \qquad
 z_{y}(\tau) =z_{y\star}\exp \left(
 -\frac{1}{2}\int_{0}^{\tau}\chi(s)\mbox{ds}\right).
\]
In the last equation, $z_{x\star},z_{y\star}$ and $z_{z\star}$ denote
the initial value of $z_{x}(\tau),z_{y}(\tau)$ and $z_{z}(\tau)$
respectively.  It follows that the deviation vector is non-zero and
regular as long as the initial data $z_{x\star},z_{y\star}$ and
$z_{z\star}$ are non-vanishing and $\chi(\tau)$ is
regular. Accordingly, \emph{the congruence of conformal geodesics will be
non-intersecting.}

\subsubsection{Analysis of the supplementary system}
\label{AnalysisSuplementarySystem}
As in the case of the core system, the supplementary system is simpler in
the gauge in which $\kappa=0$. In such case, direct inspection shows
that equations \eqref{SuplementarySystem1}-\eqref{SuplementarySystem2}
imply 
\[ 
\bar{\chi}=-\tau\bar{L}.
\]
This can be verified by direct substitution into equations
\eqref{SuplementarySystem1} and \eqref{SuplementarySystem2}.  Notice
that $\bar{L}(\tau)$ is well defined at $\mathscr{I}$ where $\tau=0$
and $\bar{\chi}(0)=0$ since the initial conditions ensure that
\[
\lim_{\tau\rightarrow 0}\frac{\bar{\chi}(\tau)}{\tau}=\frac{1}{2}. 
\]
Taking into account this observation, the system
\eqref{SuplementarySystem1}-\eqref{SuplementarySystem2} reduces to the
equation
\begin{equation} 
\label{SuplementaryZeroKappa}
\dot{\bar{L}}=-\tau\bar{L}^2 + \dot{\Theta}_{\star}\phi,
\end{equation}
with initial data
\begin{equation}
\label{SuplementaryZeroKappaData}
\bar{L}(0)=-\frac{1}{2}.
\end{equation}
Using that $\phi$ is only determined by the core system, together with the
analysis of Section \ref{AnalysisOfTheCoreSystem} one obtains the
following result:
\begin{lemma}
\label{Lemma:NoSingularitySupplementarySystem}
 The solution to equation \eqref{SuplementaryZeroKappa} 
with initial data \eqref{SuplementaryZeroKappaData} is bounded 
for $0 \leq \tau \leq \tau_{\circledcirc}$ with
\begin{equation}\label{NonInfinitesimalExistenceTime} 
\tau_{\circledcirc}\equiv
 \text{min}\big\{ \tau_{\circ}, \tau_{\bullet}
\big\}, \qquad
where \qquad \tau_{\circ} \equiv \sqrt{\dot{\Theta}_{\star}^{-1/2} \left( \frac{\pi}{2}+
2\arctan \left(\frac{1}{2}\dot{\Theta}_{\star}^{-1/2}\right)\right)}
\end{equation}
\end{lemma}

\begin{proof}
To prove that $\bar{L}(\tau)$ is bounded from above we proceed by
contradiction.  Assume that $\bar{L}\rightarrow \infty$ for some
finite $\tau_{\lightning} \in [0,\tau_{\bullet}]$,
 then $\dot{\bar{L}} \rightarrow \infty$ at
$\tau_{\lightning}$.  Now, equation \eqref{SuplementaryZeroKappa}
can be rewritten as 
\[
\dot{\bar{L}} +\tau\bar{L}^2 = \dot{\Theta}_{\star}\phi.
\]

Therefore, since $\tau \geq 0$, the last expression implies that
 $\phi\rightarrow \infty$ at $\tau_{\lightning}$.  However, in
Section \ref{AnalysisOfTheCoreSystem} we showed that $\phi$ is
 finite for $\tau \in [0, \tau_{\bullet}]$.
 This is a contradiction, and one cannot have
$\bar{L} \rightarrow \infty$ at 
$\tau_{\lightning} \in [0, \tau_{\bullet}]$.
 Consequently $L(\tau)$ is
bounded from above for $0 \leq \tau \leq \tau_{\bullet}$.
 To show that $\bar{L}(\tau)$ is bounded
from below, for $0 \leq \tau\leq \tau_{\circ}$
  with $\tau_{\circ}$ as given by 
 relation \eqref{NonInfinitesimalExistenceTime},
  observe that $\phi(\tau)>-\tau$ for $\tau \geq 0$ since $\phi(\tau)>0$.
 Using this observation, equation \eqref{SuplementaryZeroKappa} implies the
differential inequality
\[
\dot{\bar{L}} \geq -\tau (\bar{L}^2 + \dot{\Theta}_{\star}).
   \]
Since $\dot{\Theta}_{\star}>0$ one has that $(\bar{L}^2 + \dot{\Theta}_{\star})
>0$.  Thus, one can rewrite the last inequality as
\[
\frac{\dot{\bar{L}}}{(\bar{L}^2 + \dot{\Theta}_{\star}) } \geq -\tau , 
\]
which can be integrated using the initial data
\eqref{SuplementaryZeroKappaData} to give
\[ 
L(\tau) \geq -\sqrt{\Theta_{\star}}\tan\left(
\frac{1}{2}\sqrt{\dot{\Theta}_{\star}}\tau^2 +\arctan \left(
\frac{1}{2\sqrt{\dot{\Theta}_{\star}}} \right)\right).
\]
Since the function $\tan$ is bounded if its argument lies in $[0,\pi/4]$ one
concludes that $L(\tau)$ is bounded from below for $0 \leq \tau \leq
\tau_{\circ}$. Finally, taking the minimum of $\tau_{\bullet}$ and 
$\tau_{\circ}$ one obtains the result.
\end{proof}

\begin{remark}
{\em Numerical evaluations of the solutions to the
supplementary system show that it should be possible to improve Lemmas
 \ref{LemmaBoundednessCoreSystem} 
and \ref{Lemma:NoSingularitySupplementarySystem}  and conclude that the
solutions do not blow up in finite time. These results, however, will
not be required to formulate the Main Result of this article. }
\end{remark}

\subsection{Perturbations of the Schwarzschild-de Sitter spacetime}
\label{PertsOfeSdS}

In the sequel, we consider perturbations of the Schwarzschild-de Sitter
spacetime which can be covered by a congruence of conformal geodesics
so that Lemma \ref{LemmaCF} can be applied. In particular, this means
that the functional form of the conformal factor is the same for for
both the background and the perturbed spacetime.

The discussion of Section \ref{IdentifyAsympRegularData}
 brings to the
foreground the difficulties in setting up an asymptotic initial value
problem for the Schwarzschild-de Sitter spacetime in a representation
in which the initial hypersurface contains the asymptotic points
$\mathcal{Q}$ and $\mathcal{Q'}$: on the one hand, the initial data
for the rescaled Weyl tensor is singular at both $\mathcal{Q}$ and
$\mathcal{Q'}$; and, on the other hand, the curves in a congruence of
timelike conformal geodesics become asymptotically null as they
approach $\mathcal{Q}$ and $\mathcal{Q'}$ ---see Appendix \ref{AsymptoticPointsQQprime}.

Consistent with the above remarks, the analysis of the conformal
evolution equations
\eqref{EvolutionEquation3Plus1Decomposition1}-\eqref{EvolutionEquation3Plus1Decomposition8}
has been obtained in a conformal representation in which the metric on
$\mathscr{I}$ is the standard one on $\mathbb{R}\times
\mathbb{S}^2$. In this particular conformal representation the
asymptotic points $\mathcal{Q}$ and $\mathcal{Q'}$ are at infinity
respect to the 3-metric of $\mathscr{I}$ and the initial data for the
Schwarzschild-de Sitter spacetime is homogeneous.  In this section we
analyse nonlinear perturbations of the Schwarzschild-de Sitter
spacetime by means of suitably posed initial value problems. More
precisely, we analyse the development of perturbed initial data close
to that of the Schwarzschild-de Sitter spacetime in the above
described conformal representation. Then, using the conformal
evolution equations
\eqref{EvolutionEquation3Plus1Decomposition1}-\eqref{EvolutionEquation3Plus1Decomposition8}
and the theory of first order symmetry hyperbolic systems contained in
\cite{Kat75} we obtain a existence and stability result for a
reference solution corresponding to the asymptotic region of the
Schwarzschild-de Sitter spacetime ---see Figure
\ref{fig:SdSAsymptoticPerturbation}.

\subsubsection{Perturbations of asymptotic data for the 
  Schwarzschild-de Sitter spacetime}
\label{perturbationseSdS}

In what follows, let $\mathcal{S}$ denote a 3-dimensional manifold with
$\mathcal{S}\approx \mathbb{R}\times\mathbb{S}{}^{2}$. By assumption,
there exists a diffeomorphism $\psi: \mathcal{S}\rightarrow
\mathbb{R}\times\mathbb{S}^2$ which can used to pull-back a coordinate
system $x=(x^{\alpha})$ on $\mathbb{R}\times\mathbb{S}^2$ to obtain a
coordinate system on $\mathcal{S}$ ---i.e. $ \wideparen{x}= \psi^{*}
x= x\circ \psi$. Exploiting the fact that $\psi$ is a diffeomorphism
we can define not only the pull-back $\psi^{*}:
T^{*}(\mathbb{R}\times\mathbb{S}^2 )\rightarrow T^{*}\mathcal{S} $ but
also the push-forward of its inverse $ ( \psi^{-1})_{*} :
T(\mathbb{R}\times\mathbb{S}^2) \rightarrow T\mathcal{S}$. Using this
mapping, we can push-forward vector fields $\bmc_{\bmi}$ on
$T(\mathbb{R}\times\mathbb{S}^2)$ and pull-back their covector fields
$\bmalpha^{\bmi}$ on $T^{*}\mathcal{S}$ via
\[ 
\wideparen{\bmc}_{\bmi}=(\psi^{-1})_{*}\bmc_{\bmi}, \qquad
\wideparen{\bmalpha}^{\bmi}=\psi^{*}\bmalpha^{\bmi}.    
\]
In a slight abuse of notation, the fields $\wideparen{\bmc}_{\bmi}$
and $\wideparen{\bmalpha}^{\bmi}$ will be simply denoted by
$\bmc_\bmi$ and $\bmalpha^\bmi$.

\medskip
In the following, we will refer to all the fields discussed
previously for the exact  Schwarzschild-de Sitter spacetime as
the \emph{background solution} and distinguish them with a 
$\mathring{\phantom{X}}$ over the
Kernel letter ---e.g. $\ring{\bmh}$ will denote the standard metric on
$\mathbb{R}\times\mathbb{S}^2$ given in equation
\eqref{metricTorus}. Similarly, the perturbation to the corresponding
field will be identified with a $\breve{\phantom{X}}$ over the Kernel letter. 
Notice
that although the frame $\{\bmc_{\bmi}\}$ is
$\ring{\bmh}$-orthonormal, it is not necessarily orthogonal respect to
the intrinsic 3-metric $\bmh$ on $\mathcal{S}$.

\medskip
Let
$\{\bme_{\bmi} \}$ denote a $\bmh$-orthonormal frame over
$T\mathcal{S}$ and let $\{ \bmomega^{\bmi}\}$ be the associate
cobasis. Assume that there exist vector fields $\{
\breve{\bme}_\bmi\}$ such that an $\bmh$-orthonormal frame
$\{\bme_{\bmi}\}$ is related to an $\ring{\bmh}$-orthonormal frame
$\{\bmc_\bmi\}$ through the relation
 \[
\bme_{\bmi}=\bmc_{\bmi}+ \breve{\bme}_{\bmi}.
\]
 This last requirement is equivalent to introducing coordinates on
$\mathcal{S}$ such that
\begin{equation}
\bmh = \ring{\bmh} + \breve{\bmh}.
\label{hperturbation}
\end{equation}
Now, consider a solution 
\[
(h_{\bmi\bmj}, \;\chi_{\bmi
  \bmj},\; L_{\bmi},\; L_{\bmi \bmj},\; d_{\bmi\bmj\bmk},\; d_{\bmi\bmj})
\]
 to the asymptotic conformal constraint equations
\eqref{AsymptoticConformalConstraints1}-\eqref{AsymptoticConformalConstraints9}
which is, in some sense to be determined, close to initial data for
the Schwarzschild-de Sitter spacetime so that one can write
\begin{eqnarray*}
&
  h_{\bmi\bmj}|_{\mathcal{S}}=\ring{h}_{\bmi\bmj}|_{\mathcal{S}}+\breve{h}_{\bmi\bmj}|_{\mathcal{S}},
 \qquad \chi_{\bmi\bmj}|_{\mathcal{S}}=\ring{\chi}_{\bmi\bmj}|_{\mathcal{S}}+\breve{\chi}_{\bmi\bmj}|_{\mathcal{S}},
 \qquad  L_{\bmi}|_{\mathcal{S}} =\ring{L}_{\bmi}|_{\mathcal{S}}
 + \breve{L}_{\bmi}|_{\mathcal{S}}&\\ 
& L_{\bmi\bmj}|_{\mathcal{S}}
 =\ring{L}_{\bmi\bmj}|_{\mathcal{S}}+\breve{L}_{\bmi\bmj}|_{\mathcal{S}},
\qquad  d_{\bmi\bmj\bmk}|_{\mathcal{S}}
=\ring{d}_{\bmi\bmj\bmk}|_{\mathcal{S}}
+\breve{d}_{\bmi\bmj\bmk}|_{\mathcal{S}},
\qquad d_{\bmi\bmj}|_{\mathcal{S}}
=\ring{d}_{\bmi\bmj}|_{\mathcal{S}}+\breve{d}_{\bmi\bmj}|_{\mathcal{S}}. & 
\end{eqnarray*}

A spinorial version of these data can be obtained using the spatial 
Infeld-van der Waerden symbols. Accordingly, one writes
\begin{subequations}
\begin{eqnarray}
& \eta_{\bmA \bmB \bmC \bmD}|_{\mathcal{S}} = \ring{\eta}_{\bmA \bmB
    \bmC \bmD}|_{\mathcal{S}} + \breve{\eta}_{\bmA \bmB \bmC
    \bmD}|_{\mathcal{S}}, \qquad  \mu_{\bmA \bmB \bmC
    \bmD}|_{\mathcal{S}} = \breve{\mu}_{\bmA \bmB \bmC
    \bmD}|_{\mathcal{S}}, & \label{spinorInitialDataPertrubations1}
  \\
&L_{\bmA \bmB \bmC
    \bmD}|_{\mathcal{S}}=\ring{L}_{\bmA\bmB\bmC\bmD}|_{\mathcal{S}}+\breve{L}_{\bmA
    \bmB \bmC \bmD}|_{\mathcal{S}}, \qquad  \xi_{\bmA \bmB \bmC
    \bmD}|_{\mathcal{S}}=\ring{\xi}_{\bmA\bmB\bmC\bmD}|_{\mathcal{S}}
  +
  \breve{\xi}_{\bmA\bmB\bmC\bmD}|_{\mathcal{S}}, &
 \label{spinorInitialDataPertrubations2}
\\ 
& L_{\bmA \bmB}|_{\mathcal{S}} = \breve{L}_{\bmA
    \bmB}|_{\mathcal{S}}, \qquad  \chi_{\bmA \bmB \bmC
    \bmD}|_{\mathcal{S}}=\ring{\chi}_{\bmA \bmB \bmC
    \bmD}|_{\mathcal{S}}+\breve{\chi}_{\bmA\bmB\bmC\bmD}|_{\mathcal{S}}, &
\label{spinorInitialDataPertrubations3}
  \\ 
& \bme_{\bmA\bmB}|_{\mathcal{S}}=
  \ring{\bme}_{\bmA\bmB}|_{\mathcal{S}}+
  \breve{\bme}_{\bmA\bmB}|_{\mathcal{S}}, \qquad f_{\bmA
    \bmB}|_{\mathcal{S}}= \breve{f}_{\bmA\bmB}|_{\mathcal{S}}. &
\label{spinorInitialDataPertrubations4}
\end{eqnarray}
\end{subequations}
Observe that all the objects appearing in expressions
\eqref{spinorInitialDataPertrubations1}-\eqref{spinorInitialDataPertrubations4}
are scalars.

\subsubsection{Controlling the size of the perturbation}
\label{ControllingTheSizeOfThePerturbation}

In this subsection we introduce the necessary notions
and definitions to measure the size of the perturbation of the initial
data. Let $\mathcal{A} \equiv \{ (\phi_{1},\mathcal{U}_{1}),(\phi_{2},\mathcal{U}_{2})\}$  with 
$\phi_{1}:\mathcal{U}_{1}\rightarrow \mathbb{R}^3$ and $\phi_{2}:\mathcal{U}_{2}\rightarrow \mathbb{R}^3$
be an Atlas for $\mathbb{R}\times \mathbb{S}^2$.  Let $\mathcal{V}_{1} \subset \mathcal{U}_{1}$,
$\mathcal{V}_{2}\subset \mathcal{U}_{2}$ be closed sets such that $\mathbb{R}\times
\mathbb{S}^2 \subset \mathcal{V}_{1} \cup \mathcal{V}_{2}$. In addition, define
the functions
\begin{equation}
\eta_{1}(x) = \begin{cases} 1 & \textrm{ $x \in \phi_{1}(\mathcal{V}_{1})$ }
  \\ 0 & \textrm{ $x \in \mathbb{R}^3/\phi_{1}(\mathcal{V}_{1})$} \\
   \end{cases}, 
\qquad \eta_{2}(x) = \begin{cases} 1 & \textrm{ $x \in
    \phi_{2}(\mathcal{V}_{2})$ } \\ 0 & \textrm{ $x \in
    \mathbb{R}^3/\phi_{2}(\mathcal{V}_{2})$} \\
   \end{cases} .
\end{equation}
Observe that any point $p\in \mathcal{S}$  is described in local coordinates by
$x_{p}=(\phi_{i} \circ \psi )(p)$ with $x_{p} \in \phi(\mathcal{U})$
 where $\psi$ is
the diffeomorphism defined in Section \ref{perturbationseSdS} 
and $(\phi,\mathcal{U}) \in \mathcal{A}$.
Consequently, any smooth function $Q:
\mathcal{S}\rightarrow \mathbb{C}^{N}$ can be regarded in
local coordinates as
$Q(x):\phi(\mathcal{U}) \rightarrow \mathbb{C}^N$. Let $Q_{i}(x)$
denote the restriction of $Q(x)$ to one the open sets
$\phi_{i}(\mathcal{U}_{i})$ for ${i=1,2 }$. 
Then, we define the norm of $Q$ as
\[
\parallel Q \parallel_{\mathcal{S},m} \equiv \parallel
\eta_{1}(x)Q_{1}(x) \parallel_{\mathbb{R}^3,m} + \parallel
\eta_{2}(x)Q_{2}(x) \parallel_{\mathbb{R}^3,m}
\]
where
\[
\parallel {Q}\parallel_{\mathbb{R}^3,m}= \left(
\sum_{l=0}^{m}\sum_{\alpha_1,...,\alpha_l}^{3}\int_{\mathbb{R}^3}
(\partial_{\alpha_1}...\partial_{\alpha_l}Q)^2
\mbox{d}^3x\right)^{1/2}.
\]

\medskip
Now, we use these notions to define  Sobolev norms for
any quantity $Q_{\mathcal{K}}$ with $_\mathcal{K}$
 being an arbitrary string of frame spinor indices as 
\[
\parallel{Q}_{\mathcal{K} }\parallel_{\mathcal{S},m}\equiv
\sum_{{\kappa}}^{{}}\parallel {Q}_{\kappa}\parallel_{\mathcal{S},m}.
\]

In the last expression $m$ is a positive integer and the 
sum is carried over all the independent components of
$Q_{\mathcal{K}}$ which have been denoted by $Q_{\kappa}$.

\subsubsection{Formulation of the evolution problems}
\label{FormulationEvolutionProblems}

Consistent with the split
\eqref{spinorInitialDataPertrubations1}-\eqref{spinorInitialDataPertrubations4}
for the initial data, we look
for solutions to the conformal evolution equations
\eqref{structureEvolutionEquations1}-\eqref{structureEvolutionEquations2}
of the form
\begin{subequations} 
\begin{eqnarray}
&\eta_{\bmA \bmB \bmC \bmD} = \ring{\eta}_{\bmA \bmB \bmC \bmD} +
  \breve{\eta}_{\bmA \bmB \bmC \bmD}, \qquad  \mu_{\bmA \bmB
    \bmC \bmD} = \breve{\mu}_{\bmA \bmB \bmC \bmD},& \label{perturbedSolution1} \\
&L_{\bmA \bmB
    \bmC \bmD}=\ring{L}_{\bmA\bmB\bmC\bmD}+\breve{L}_{\bmA \bmB \bmC
    \bmD}, \qquad  \xi_{\bmA \bmB \bmC
    \bmD}=\ring{\xi}_{\bmA\bmB\bmC\bmD} +
  \breve{\xi}_{\bmA\bmB\bmC\bmD}, &\label{perturbedSolution2}\\ 
& L_{\bmA \bmB} = \breve{L}_{\bmA
    \bmB}, \qquad  \chi_{\bmA \bmB \bmC \bmD}
  =\ring{\chi}_{\bmA \bmB \bmC \bmD} +\breve{\chi}_{\bmA\bmB\bmC\bmD}, &
\label{perturbedSolution3}\\ 
& \bme_{\bmA\bmB} = \bmc_{\bmA\bmB} +
  \breve{\bme}_{\bmA\bmB}, \qquad  f_{\bmA \bmB}=
  \breve{f}_{\bmA\bmB}.\label{perturbedSolution4}& 
\end{eqnarray}
\end{subequations}

\medskip
Using the notation introduced
 in Section \ref{generalsetting}, the initial data
\eqref{spinorInitialDataPertrubations1}-\eqref{spinorInitialDataPertrubations4}
 will be represented as
$\mathbf{u}_{\star}$. The perturbed initial data will be assumed to be small in
the sense that given some $\varepsilon > 0$ one has
\begin{eqnarray*}
&& \parallel \breve{\mathbf{u}}_{\star} \parallel_{\mathcal{S},m} \equiv \parallel
  \breve{\chi}_{\bmA \bmB \bmC \bmD}\parallel_{\mathcal{S},m}
  + \parallel \breve{\xi}_{\bmA \bmB \bmC
    \bmD}\parallel_{\mathcal{S},m} + \parallel \breve{L}_{\bmA \bmB
    \bmC \bmD}\parallel_{\mathcal{S},m} + \parallel \breve{L}_{\bmA
    \bmB}\parallel_{\mathcal{S},m} \\ && \hspace{2cm} + \parallel
  \breve{\bme}_{\bmA \bmB }\parallel_{\mathcal{S},m} + \parallel
  \breve{f}_{\bmA \bmB}\parallel_{\mathcal{S},m} + \parallel
  \breve{\phi}_{\bmA \bmB\bmC \bmD}\parallel_{\mathcal{S},m} <
  \varepsilon.
\end{eqnarray*}

\begin{remark}
{\em Notice that, 
 as a consequence of the conformal representation
being considered, the above smallness requirement on the
 perturbed initial data constraints  the possible behaviour
of the perturbation near the asymptotic points $\mathcal{Q}$ and
$\mathcal{Q'}$. To see this in more detail let
$\breve{\phi}$ denote a perturbation of the initial data for some
component the rescaled Weyl spinor.  For simplicity, assume that in
some local coordinates $(\psi,\theta,\varphi)$
for $\mathbb{R}\times \mathbb{S}^2$, the perturbed field $\breve{\phi}$ is
independent of $(\theta, \varphi)$.
In such case, if $\breve{\phi}\in L^{2}(\mathbb{R})$ one has that
\begin{equation}\label{ObservationDecay}
\breve{\phi}= \mathcal{O}(\psi^{-\beta})
\end{equation}
with $\beta >1/2$. 
Consequently,  in the $\mathbb{R}\times\mathbb{S}^2$-conformal representation
 the perturbations must decay  at infinity, i.e. as
 they approach $\mathcal{Q}$ and $\mathcal{Q'}$. 
Under the conformal transformation
$\bmg=\varpi^{2}\acute{\bmg}$ the components of the rescaled Weyl
spinor transform as ${\phi}_{\bmA \bmB \bmC
  \bmD}=\varpi^{-3}\acute{\phi}_{\bmA \bmB \bmC \bmD}$.
This last expression is consistent with the
frame version of the conformal transformation rule
given in Lemma \ref{Lemma:TTtensorTransformation}. Taking into
account the discussion of Section \ref{MetricOnScri} 
and equation \eqref{ObservationDecay} one concludes that for the
corresponding perturbation in the $\mathbb{S}^3$-conformal representation
 one has 
\[
\breve{\phibar}= \mathcal{O}(\xi^{-3}(\ln|\xi|)^{-\beta})
\]
near the South pole $\xi=0$. Consequently, initial data on
$\mathbb{R}\times\mathbb{S}^2$ satisfying $L^2$-decay conditions near
infinity correspond, in general, to data which is singular in other conformal
representations. In other words,
the class of perturbation data that we can consider can be, in principle,
singular at both the North and South poles in the $\mathbb{S}^3$-conformal
 representation.}
\end{remark}

\begin{remark}
{\em An explicit class of perturbed
asymptotic initial data sets can be constructed, keeping the initial metric fixed to
be standard one on $\mathbb{R}\times\mathbb{S}^2$, using the analysis
of \cite{DaiFri01} as follows: introduce Cartesian coordinates
$(x^{\bmalpha})$ in $\mathbb{R}^3$ with origin located at a fiduciary
position $\mathcal{Q}$ and define a polar coordinate via
$\rho\equiv\delta_{\bmalpha\bmbeta}x^{\bmalpha}x^{\bmbeta}$. The
general solution of the equation
\[
\acute{D}^{i}\acute{d}_{ij}=0,
\]
 where $\acute{D}^{i}$
is the Levi-Civita connection on $\mathbb{R}^3$, can be
parametrised as
\[
\acute{d}_{ab}=\acute{d}_{ab}^{(P)}+\acute{d}_{ab}^{(J)}+\acute{d}_{ab}^{(A)}+\acute{d}_{ab}^{(Q)}+\acute{d}_{ab}^{(\Lambda)}.
\]
The terms $\acute{d}_{ab}^{(P)}$, $\acute{d}_{ab}^{(J)}$
$\acute{d}_{ab}^{(A)}$, $\acute{d}_{ab}^{(Q)}$ are divergent at
$\mathcal{Q}$ and have been explicitly derived in
\cite{DaiFri01}. Given any smooth function $\Lambda(x^{\bmalpha})$ on
$\mathbb{R}^3$ the term $\acute{d}_{ab}^{(\Lambda)}$ can be obtained
using the operators $\eth$ and $\bar{\eth}$ ---see \cite{Ste91} for
definitions.  This term can have, in general, any behaviour near
$\mathcal{Q}$ ---see \cite{DaiFri01}.  However, setting
$\Lambda=\mathcal{O}(\rho^n)$ with $n\geq 3$ the term
$\acute{d}_{ab}^{(\Lambda)}$ is regular near $\mathcal{Q}$.  Using the
frame version of the conformal transformation rule of Lemma
\ref{Lemma:TTtensorTransformation} and either equation
\eqref{PlaneToSphere1} or \eqref{PlaneToSphere2} one can verify that
the corresponding term in the $\mathbb{S}^3$-representation is
$\hspace{1mm}\dbar_{\bma\bmb}^{(\Lambda)}=\mathcal{O}(\rho^{n+3})$.
Similarly, using the conformal transformation formulae, given in
Section \ref{MetricOnScri}, relating the
$\mathbb{S}^3$ and $\mathbb{R}\times\mathbb{S}^2$-representations of
the initial data, one obtains
$d_{\bma\bmb}^{(\Lambda)}=\mathcal{O}(\rho^{n+6})$.  We observe that regular behaviour of perturbed
initial data in the $\mathbb{R}\times\mathbb{S}^2$-representation does
not necessarily correspond to regular behaviour in the
$\mathbb{S}^3$-representation nor in the
$\mathbb{R}^3$-representation.}
\end{remark}

\subsection{The main result}
\label{Section:AsymptoticIVP}

\begin{figure}[t]
\centering
\includegraphics[width=0.93\textwidth]{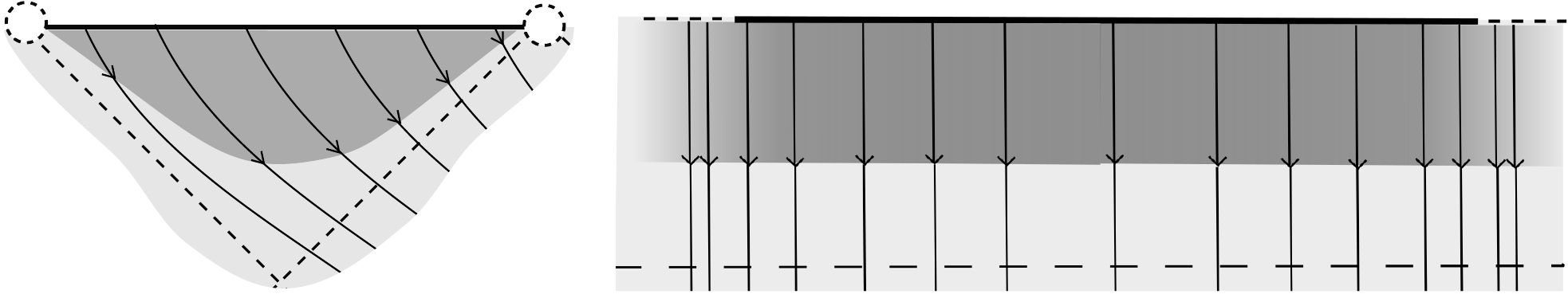}
\put(-400,90){(a)}
\put(-330,72){$\mathscr{I}$}
\put(-245,90){(b)}
\put(-120,72){$\mathscr{I} $}
\put(-10,70){$\tau=0$}
\put(-10,32){$\tau=\tau_{\bullet}$}
\put(-10,10){$\tau=\tau_{\mathcal{CH}}$}
\caption{ Schematic depiction of the development of 
perturbed initial data for the Schwarzschild-de Sitter spacetime and the
 congruence of conformal geodesics.  In \emph{(a)} the evolution of
asymptotic initial data is depicted in the conformal representation in
which the asymptotic points $\mathcal{Q}$ and $\mathcal{Q'}$ are at a
finite distance respect to the metric on $\mathscr{I}$.  Figure
\emph{(b)} shows a schematic depiction of the evolution of asymptotic
initial data in the conformal representation in which Theorem
\ref{ExistenceCauchySatbility-eSdS1} has been formulated. In contrast
to the conformal representation leading to Figure \emph{(a) }, the
initial data is homogeneous and formally identical for the
subextremal, extremal or hyperextremal cases. 
In both diagrams, the
dashed line corresponds to the location of an hypothetical Cauchy horizon of the development.}
\label{fig:ComparisonGradient}
\end{figure}
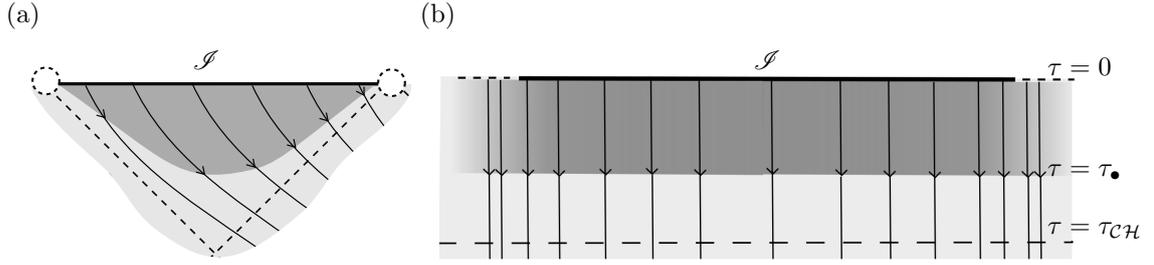

The main
analysis of the background solution in Section \ref{sec:CoreAnalysis}
was performed in a conformal representation in which the asymptotic
initial data is homogeneous and the extrinsic
  curvature of $\mathscr{I}$ vanishes ---i.e. $\kappa =0$.  The general
evolution equations
\eqref{structureEvolutionEquations1}-\eqref{structureEvolutionEquations2}
consist of transport equations for ${\bm\upsilon}$ coupled with a
system of partial differential equations for ${\bm\phi}$. However, as
shown in Section \ref{sec:CoreAnalysis}, the assumption of
\emph{spherical symmetry} implies that the only independent component
of the spinorial field $\phi_{\bmA \bmB \bmC \bmD}$ is $\phi_{2}$. Consequently,
the system
\eqref{structureEvolutionEquations1}-\eqref{structureEvolutionEquations2}
reduces, for the background fields
$\mathring{\mathbf{u}}=(\mathring{\bm \upsilon},\mathring{\bm \phi})$,
to a system of ordinary differential equations. The
\emph{Piccard-Lindel{\"o}f theorem} can be applied to discuss local
existence of the latter system. However, one does not have, \emph{a
  priori}, control on the smallness of the existence time. To obtain
statements concerning \emph{the existence time of the perturbed
  solution}, we recall that the discussion of the evolution
  equations of Section \ref{sec:CoreAnalysis} shows that the
  components of solution $\ring{\mathbf{u}}$ are regular for $ \tau
  \in [0, \tau_{\circledcirc}]$ with $\tau_{\circledcirc}$ as given
in equation \eqref{NonInfinitesimalExistenceTime}, so that the
guaranteed existence time is not arbitrarily small.

\medskip
The analysis of the core system in Section \ref{sec:CoreAnalysis} was
  restricted to the case $\kappa=0$, in which the conformal boundary
  has vanishing extrinsic curvature. In this case, we obtained an
  explicit existence time $\tau_{\circledcirc}$ for the solution to
  the conformal evolution equations.  In contrast, the analysis given
  in Appendix \ref{FormationOfSingularitiesAndReparametrisations}
  shows that in general, for $\kappa \neq 0$, the core system develops
  a singularity at finite $\tau_{\lightning}$.  Since the results
  given in Section \ref{DeviationEquationForTheCongruence} for the
  conformal deviation equations hold not only for $\kappa =0$, but
  for any $\kappa$ as long as $\partial_{\psi}\kappa =0$, one has that
  the congruence of conformal geodesics is non-intersecting in the
  $\kappa \neq 0$ case as well. This shows that, the singularities in
  the core system in the case $\kappa \neq 0$ are not gauge
  singularities.  The estimation for the existence time
  $\tau_{\circledcirc}$ in the $\kappa=0$ case along with the
  discussion of the reparametrisation of conformal geodesics given in
  Appendix \ref{ExploitingTheConformalGaugeFreedom} can, in principle,
  be used to obtain an estimation for the existence time
  $\tau_{\otimes}$ in the case $\kappa \neq 0$. 

In this section it is shown how one can exploit these
  observations, together with the theory for symmetric hyperbolic
systems, to prove the existence of solutions to the general conformal
evolution equations with the same existence time
  $\tau_{\circledcirc}$ for small perturbations of asymptotic initial
data close to that of the Schwarzschild-de Sitter reference
solution. By construction, the development of this perturbed data will
be contained in the domain of influence which corresponds, in this
case, to the asymptotic region of the spacetime ---see Figure
\ref{fig:ComparisonGradient}.

\medskip
Taking into account the above
remarks and using the theory of symmetric hyperbolic
systems contained in \cite{Kat75} one can formulate the following
existence and Cauchy stability result:

\begin{theorem}
[\textbf{\em existence and Cauchy stability for
    perturbations of asymptotic initial data for the 
    Schwarzschild-de Sitter spacetime}]
\label{ExistenceCauchySatbility-eSdS1}
  Let $\mathbf{u}_{\star}=\ring{\mathbf{u}}_{\star} +
  \breve{\mathbf{u}}_{\star}$ denote asymptotic initial data for the
  extended conformal Einstein field equations on a 3-dimensional
  manifold $\mathcal{S}\approx  \mathbb{R}\times\mathbb{S}^2$ where
  $\ring{\mathbf{u}}_{\star}$ denotes the asymptotic initial data for
  the  Schwarzschild-de Sitter spacetime (subextremal, extremal and
  hyperextremal cases) with $\kappa=0$ 
 in which the asymptotic
points $\mathcal{Q}$ and $\mathcal{Q'}$ are at infinity. Then, for $m \geq 4$
  and $\tau_{\circledcirc} $ as given in equation
 \eqref{NonInfinitesimalExistenceTime}, there exists
  $\varepsilon>0$ such that:

\begin{itemize}
\item[(i)] for
  $||\breve{\mathbf{u}}_{\star}||_{\mathcal{S},m}<\varepsilon$ , there
  exist a unique solution $\breve{\mathbf{u}}$ to the conformal
  evolution equations
  \eqref{formPerturbationEquations1}-\eqref{formPerturbationEquations2}
  with a minimal existence interval $[0,\tau_{\circledcirc}]$ and
\[
\breve{\mathbf{u}} \in C^{m-2}([0,\tau_{\circledcirc}] 
\times \mathcal{S}, \mathbb{C}^{N}),
\]
and the associated congruence of conformal geodesics
 contains no conjugate points in $[0,\tau_{\circledcirc}]$;

\item[(ii)] given a sequence of perturbed data
  $\{\breve{\mathbf{u}}_\star^{(n)}\}$ such that
\[
\parallel \breve{\mathbf{u}}_\star^{(n)} \parallel_{\mathcal{S},m}
\rightarrow 0 \qquad \mbox{as} \qquad n\rightarrow \infty,
\]
then the corresponding solutions $\{\breve{\mathbf{u}}^{(n)}\}$ have a
minimum existence interval $[0,\tau_{\circledcirc}]$ and it holds that
\[
\parallel \breve{\mathbf{u}}^{(n)} \parallel_{\mathcal{S},m}
\rightarrow 0 \qquad \mbox{as} \qquad n\rightarrow \infty
\]
uniformly in $\tau \in[0,\tau_\circledcirc ]$ as $n\rightarrow \infty$;

\item[(iii)] the solution $\mathbf{u}=\ring{\mathbf{u}}+
  \breve{\mathbf{u}}$ is unique in $[0,\tau_\circledcirc ]\times
  \mathcal{S}$ and implies a $C^{m-2}$ solution
  $(\tilde{\mathcal{M}}_{\tau_\circledcirc},\tilde{\bmg})$ to the Einstein
  vacuum equations with the same de Sitter-like Cosmological constant
  as the background solution where
\[ 
\tilde{\mathcal{M}}_{\tau_\bullet} \equiv
    (0,\tau_{\circledcirc})\times \mathcal{S}.
\] 
Moreover, the hypersurface $\mathscr{I}\equiv
    \{0\}\times \mathcal{S}$ represents the conformal boundary of the
    spacetime.
\end{itemize}
\end{theorem}

\begin{proof}
Points \emph{(i)} and \emph{(ii)} are a direct application of the theory contained
in \cite{Kat75} where it is used that the background solution
$\ring{\mathbf{u}}$ is regular on $\tau \in[0, \tau_{\circledcirc}]$.
The initial data for the Schwarzschild-de Sitter spacetime
encoded in $\mathbf{u}_{\star}$ is in a representation in which the
points $\mathcal{Q}$ and $\mathcal{Q'}$ are at infinity. 
  Observe that the asymptotic initial data, as derived
 in Section \ref{sec:InitialData},  for the
subextremal, extremal and hyperextremal cases are formally
 the same ---in particular,
notice that the initial data for the electric part of the rescaled Weyl tensor
contains information about the mass $m$ while the conformal factor
$\Theta$ carries information about $\lambda$. The arguments in the
analysis of Section \ref{sec:CoreAnalysis}
 are irrespective of the relation
between $\lambda$ and $m$. The key
observation in the proof is that one can apply the general theory of
symmetric hyperbolic systems of \cite{Kat75} for each
open set and chart of an atlas for $\mathbb{R}\times \mathbb{S}^2$;
then, these local solutions can be patched together to obtain the
required global solution over $[0,\tau_\circledcirc]\times
\mathcal{S}$ ---it is sufficient to cover
$\mathbb{R}\times\mathbb{S}^2$ with finitely many patches (two) as
discussed in Section
\ref{ControllingTheSizeOfThePerturbation}. Details of a similar
construction in the context of characteristic problems can be found in
\cite{Fri91}. To prove point \emph{(iii)} first observe that from
Lemma \ref{Thm:PropagationConstraints} the solution to the conformal
evolution system
\eqref{formPerturbationEquations1}-\eqref{formPerturbationEquations2}
implies a solution $\mathbf{u}=\ring{\mathbf{u}}+ \breve{\mathbf{u}}$
to the extended conformal Einstein field equations on
$[0,\tau_\circledcirc]\times\mathcal{S}$ if
$\mathbf{u}_\star=\ring{\mathbf{u}}_\star+ \breve{\mathbf{u}}_\star$
solves the conformal constraint equations on the initial hypersurface.
This solution implies, using Lemma \ref{lemma:XCEFE-EFE}, a solution
to the Einstein field equations whenever the conformal factor is not
vanishing. General results of the theory of asymptotics implies then
that the initial hypersurface $\mathcal{S}$ can be interpreted as the
conformal boundary of the physical spacetime
$(\tilde{\mathcal{M}}_{\tau_\bullet},\tilde{\bmg})$ ---see
\cite{Ste91,CFEBook}.
\end{proof}


\section{Conclusions}
\label{Sec:Conlcusions}

In this article we have studied the Schwarzschild-de Sitter family of
spacetimes as a solution to the extended conformal Einstein field
equations expressed in terms of a conformal Gaussian system. Given that,
in principle, it is not possible to explicitly express the spacetimes
in this gauge, we have adopted the alternative strategy of formulating
an asymptotic initial value problem for a spherically symmetric
spacetime with a de Sitter-like Cosmological constant. The
generalisation of Birkhoff's theorem to vacuum spacetimes with
Cosmological constant then ensures that the resulting solutions are
necessarily a member of the Schwarzschild-de Sitter spacetime.

As part of the formulation of an asymptotic initial value problem for
the Schwarzschild-de Sitter spacetime we needed to specify suitable
initial data for the conformal evolution equations. The rather simple
form that the conformal constraint equations acquire in the framework
considered in this article allows to study in detail the conformal
properties of the Schwarzschild-de Sitter spacetime at the conformal
boundary and, in particular, at the asymptotic points where
the conformal boundary \emph{meets} the horizons. The key observation from
this analysis is that the conformal structure is singular at these
points and cannot be regularised in an obvious manner. Accordingly,
any satisfactory formulation of the asymptotic initial value problem
will exclude these points.

An interesting property of the conformal evolution equations under the
assumption of spherical symmetry is that the system reduces to a
set of transport equations along the conformal geodesics covering the
spacetime. The essential dynamics, and in particular the formation of
singularities in the solutions to this system, is governed by a
\emph{core system} of three equations ---one of them a Riccati
equation. As discussed in Appendix
  \ref{FormationOfSingularitiesAndReparametrisations}, this core
system provides a mechanism for the formation of singularities in the
exact solution.  The analysis of the
core system allows not only to study the properties on the
Schwarzschild-de Sitter spacetime expressed in terms of a conformal
Gaussian gauge system, but also to understand the effects that the
\emph{gauge data} has on the properties of the conformal
representation arising as a solution to the conformal evolution
equations.  It is of
interest to explore the idea of whether the
mechanisms identified in the analysis of the core system could be used
to analyse the formation of singularities in more complicated
spacetimes ---say, in the developments of perturbations of asymptotic
initial data for the Schwarzschild-de Sitter spacetime.

The conformal representation of the Schwarzschild-de Sitter spacetime
obtained in this article has been used to show that it is possible to
construct, say, \emph{future asymptotically de Sitter
} solutions to the Einstein vacuum Einstein with a minimum existence time ---as measured by
the proper time of the conformal geodesics used to construct the gauge
system--- which can be understood as perturbations of a member of the
Schwarzschild-de Sitter family of spacetimes. As already mentioned in the main text, it
is an interesting problem to determine the maximal Cauchy development to these
spacetimes. In order to obtain the maximal Cauchy development of
suitably small perturbations of asymptotic data for the
Schwarzschild-de Sitter one would require the use of more refined
methods of the theory of hyperbolic partial differential equations as
one is, basically, confronted with \emph{global existence problem} for
the conformal evolution equations. In this respect, we conjecture that
the \emph{time symmetric conformal representation} in which $\kappa=0$
together with the \emph{global stability} methods of \cite{KreLor98}
should allow us to make inroads into this issue. Closely related to the construction of the maximal
development of perturbations of asymptotic initial data of the
Schwarzschild-de Sitter spacetime is the question whether there is a
Cauchy horizon associated to the boundary of this development. If this is the
case, one would like to investigate the properties of this
horizon. Intuitively, the answer to these issues should
depend on the relation between the asymptotic points $\mathcal{Q}$ and
$\mathcal{Q}'$ and the conformal structure of the spacetime. In
particular, one would like to know whether the singularities of
the rescaled Weyl tensor at these points generically propagate along
the boundary of the perturbed solution ---notice, that they do not
for the background solution. If one were able to use the
$\mathbb{R}\times \mathbb{S}^2$-representation of the conformal
boundary of  perturbations of asymptotic initial data for the
Schwarzschild-de Sitter to construct a maximal development and to gain
sufficient control on the asymptotic behaviour of the various
conformal fields, one could then rescale this solution to obtain a
representation with a conformal boundary of the form $\mathbb{S}^3\setminus\{
\mathcal{Q},\mathcal{Q}' \}$. As discussed in the main text, in this
representation some fields are singular at $\mathcal{Q}$ and $\mathcal{Q}'$. This observation suggests that
this construction could shed some light regarding the propagation (or
lack thereof) of singularities near the asymptotic points
$\mathcal{Q}$ and $\mathcal{Q}'$.

\bigskip


\section*{Acknowledgements}

EG gratefully acknowledges the support from Consejo
Nacional de Ciencia y Tecnolog\'ia (CONACyT Scholarship
494039/218141). We have profited from conversations with C. L\"ubbe, 
H. Friedrich and M. Mars.

\appendix
\section{Appendix: The asymptotic points $\mathcal{Q}$ and $\mathcal{Q'}$ and
conformal geodesics in the Schwarzschild-de Sitter spacetime}
\label{AsymptoticPointsQQprime}

\subsection{Analysis of the asymptotic points $\mathcal{Q}$ and $\mathcal{Q'}$}

In Section \ref{ConformalStructureSdS} it was shown that there exist a
conformal representation of the Schwarzschild-de Sitter spacetime in
which the metric at the conformal boundary is $\bmhbar$ ---i.e. the
standard metric on $\mathbb{S}^3$. In addition, we observed that the
North and South pole of $\mathbb{S}^3$ correspond to special points in
the conformal structure that we have labelled as $\mathcal{Q}$ and
$\mathcal{Q'}$. These asymptotic regions are represented in the
Penrose diagram for the subextremal, extremal and hyperextremal
Schwarzschild-de Sitter spacetime as the points where the conformal
boundary and the Cosmological horizon, Killing horizon and
singularity, respectively, seem to meet ---see Figures
\ref{fig:SubSdSDiagram} and \ref{fig:eSdSDiagram} and
\ref{fig:HypSdSDiagram}.  As discussed in Section
\ref{ConformalStructureSdS} these points correspond to
$(\bar{U},\bar{V})=(\pm \frac{\pi}{2},\pm \frac{\pi}{2})$ for which
the tortoise coordinate $\mathfrak{r}$ is not well defined. In Section
\ref{IdentifyAsympRegularData} we showed that in the conformal
representation in which the initial metric is $\bmhbar$ the data for
the electric part of the rescaled Weyl tensor $\dbar_{\bmi\bmj}$, as
given in equation \eqref{TTtensorSphere}, is singular precisely at
$\mathcal{Q}$ and $\mathcal{Q'}$. Observe that written in spinorial
terms the initial data for the rescaled Weyl spinor in this conformal
representation is given by
\[
\phibar_{\bmA \bmB \bmC \bmD}=
\frac{6m}{\sqrt{1-\omega^{2}(\xi)}}\epsilon^{2}{}_{\bmA \bmB \bmC
  \bmD}
\]
which is singular at both $\mathcal{Q}$ and $\mathcal{Q'}$.  This
situation resembles that of the geometry near spacelike infinity
$i^{0}$ of the Minkowski spacetime and the construction of the
\emph{cylinder at infinity} given in \cite{Fri98a} which allows to
regularise the data for the rescaled Weyl spinor. However, some
experimentation reveals that this type of regularisation procedure (in
contrast with the analysis of Schwarzschild spacetime given in
\cite{Fri98a}) cannot be implemented in the analysis of the
Schwarzschild-de Sitter spacetime without spoiling the regular
behaviour of the conformal factor.  Since the hyperbolic reduction
procedure for the extended conformal Einstein field equations is based
on the existence of a congruence of conformal geodesics in spacetime,
the singular behaviour of the initial data for the rescaled Weyl
spinor suggest that the congruence of conformal geodesics does not
cover the region of the spacetime corresponding to $\mathcal{Q}$ and
$\mathcal{Q'}$. To clarify this point, in the remaining of this
section we analyse the behaviour of conformal geodesics as they
approach the asymptotic points $\mathcal{Q}$ and $\mathcal{Q'}$.

\subsection{ Geodesics in Schwarzschild-de Sitter spacetime }
\label{eSdS:AsymptoticPoints}

The method for the hyperbolic reduction for the extended conformal Einstein field
equations available in the literature requires adapting the gauge
to a congruence of conformal geodesics. The behaviour of metric
geodesics in the Schwarzschild-de Sitter spacetime has been already
studied \cite{JakHel89, HacLam08} and an analysis of conformal geodesics
in Schwarzschild-de Sitter and anti-de-Sitter spacetimes is 
carried out in \cite{GarGasVal15}.  In static coordinates $(t,r,\theta,\varphi)$
the equation for radial timelike geodesics, 
$(\theta=\theta_{\star},\varphi=\varphi_{\star})$ with $\theta_{\star}$ and $\varphi_{\star}$
constant, are
\begin{equation}\label{SSmetricGeodesicEquations}
\frac{\mbox{d}r}{\mbox{d}\tilde{\tau}}=\sqrt{\gamma^2-F(r)}, \qquad
\frac{\mbox{d}t}{\mbox{d}\tilde{\tau}}=\frac{\gamma}{F(r)}.
\end{equation}
The first equation can be formally integrated as
\begin{equation}\label{rGeodesicIntegral}
\tilde{\tau}-\tilde{\tau}_{\star}=\int_{r_{\star}}^{r}\frac{1}{\sqrt{\gamma^2-F(s)}}\mbox{d}s
\end{equation}
where $\tilde{\tau}$ is the $\tilde{\bmg}_{SdS}$-proper time and $\gamma$
is a constant of motion which can be identified with the specific
 energy of a particle moving along the geodesic.
 The equation for $t$ can be solved once equation
\eqref{rGeodesicIntegral} has been integrated. As pointed out in
\cite{BicPod95,Pod99}, by choosing $\gamma=1$
 one can explicitly solve this integral.
However in general, for arbitrary $\gamma$,
 the integral is complicated and cannot be
written in terms of elementary functions.  A side observation is that if $r
\neq r_{b}$ and $r \neq r_{c}$ then the curves of constant $t$ correspond
to geodesics with $\gamma=0$. Finally, its worth noticing that geodesics with
constant $r$ are characterised by the condition
\begin{equation}\label{CriticalCurveCondition}
\gamma^2-F(r)=0.
\end{equation}
This last type of curves, which will be called \emph{critical
  curves}, are analysed in Section \ref{CriticalCurvesDiscussion}.
 In general, the properties of conformal geodesics differ from
their metric counterparts. However, in the case of an Einstein
spacetime with spacelike conformal boundary any conformal geodesic
leaving $\mathscr{I}$ orthogonally is, up to reparametrisation, a metric
geodesic --- see \cite{FriSch87} and Lemma
\ref{ReparametrisationMetricToConformalLeavingScri}.

\subsection{A special class of conformal geodesics in the Schwarzschild-de Sitter spacetime}
\label{UnitEnergySpecialGeodesic}
As briefly mentioned in Section \ref{eSdS:AsymptoticPoints} and
pointed out in \cite{BicPod95,Pod99}, in general, the integral
\eqref{rGeodesicIntegral} cannot be written in terms of elementary
functions except for the the special case when $\gamma=1$ where it
yields
\begin{equation}\label{ExplicitMetricGeodesic}
r(\tilde{\tau})=\mathcal{C} e^{\tilde{\tau}}\left(1-\left(\frac{3m}
{2|\lambda|}\right)\mathcal{C}^{-3}e^{-3\tilde{\tau}}\right)^{2/3},
\end{equation}
where $\mathcal{C}$ is an integration constant.  The last expression
is valid irrespective of the relation between $m$ and $\lambda$.  One can
also use this expression to integrate the second equation in
\eqref{SSmetricGeodesicEquations} to obtain the geodesic parametrised
as $(r(\tilde{\tau}),t(\tilde{\tau}))$.  The integration of $t$ will
not be required for the purposes of the analysis of this section.  A
complete analysis of conformal geodesics in the Schwarzschild-de
Sitter and anti-de Sitter spacetimes will be given in
\cite{GarGasVal15}.  By virtue of Lemma
\ref{ReparametrisationMetricToConformalLeavingScri} one can recast the
geodesic with $\gamma=1$ as a conformal geodesic by reparametrising it
in terms of the unphysical proper time as determined by equations
\eqref{ReparametrisationConformalGeodesics} and
\eqref{UniversalConformalFactor}.  A straightforward computation
yields
\begin{equation}\label{ReparametrisationUnphysicalToPhysical}
\tilde{\tau}(\tau)= \sqrt{\frac{3}{|\lambda|}}\ln
\left\lvert\frac{\tau}{2 + \kappa \tau}\right\rvert.
\end{equation}
Equivalently, assuming either $\kappa>0$ and $ \tau \geq 0$ or
$\kappa<0$ and $ 0 \leq \tau \leq-2/\kappa$ one obtains in both cases
\begin{equation}\label{ReparametrisationPhysicalToUnPhysical}
\tau(\tilde{\tau})=\frac{2\exp \Big(
  \sqrt{\displaystyle\frac{|\lambda|}{3}}\tilde{\tau}\Big)}{1-\kappa\exp \Big(
  \sqrt{\displaystyle\frac{|\lambda|}{3}}\tilde{\tau}\Big)}.
\end{equation}
From the last expression one can verify that
 \[
\lim_{\tilde{\tau}\rightarrow -\infty } \tau(\tilde{\tau})=0, \qquad
\lim_{{\tau}\rightarrow \infty} \tau(\tilde{\tau})=-2/\kappa,
\] 
as expected. In what follows, we will rewrite equation
\eqref{ExplicitMetricGeodesic} in terms of the unphysical proper time
as \begin{equation}
  r(\tau)=\frac{1}{(m|\lambda|)^{2/3}}\frac{(m|\lambda|\mathcal{C}^3\tau^3
    -6(2+\kappa \tau )^3)^{2/3}}{\mathcal{C} \tau(\tau + 2\kappa
    \tau)}.
\end{equation}
From the last expressions one can verify that one has $r \rightarrow
\infty$ as $\tau \rightarrow 0$ and $\tau \rightarrow -2/\kappa$.  The
location of the singularity $r=0$ is determined by 
\[
\tau_{\lightning}=\frac{2}{(m|\lambda|)^{1/3}\mathcal{C}-\kappa}.
\]
Recalling that $\mathcal{C}$ is an integration constant which depends
on the initial data for the congruence, since the only freedom left in
the conformal factor is encoded in $\kappa$, one realises that
$\mathcal{C}=\mathcal{C}(\kappa)$.  So one cannot draw any precise
conclusion about the location of the singularity unless one further
identifies explicitly $\mathcal{C}(\kappa)$.  In particular, 
considering constant $\kappa$ as we have done for the analysis of the
core system and setting $\mathcal{C}$ to be proportional to $\kappa$, say
$\mathcal{C}= \frac{(2\varkappa+1)}{(m|\lambda|)^{1/3}} \kappa $ for
some proportionality constant $\varkappa$, one obtains
\begin{equation*}
\tau_{\lightning}=\frac{1}{\varkappa\kappa},
\end{equation*}
which is in agreement the with the qualitative behaviour
of the core system as shown in Figures \ref{Figure:Kappa0},
\ref{fig:figPos}, and
\ref{fig:Negative}.
 Notice, however, that the arguments of the core system
given in Section \ref{AnalysisOfTheCoreSystem} and
Appendix \ref{FormationOfSingularitiesAndReparametrisations}
 do not rely on integrating
\eqref{rGeodesicIntegral} explicitly as we have done in this
section.

\subsection{Critical curves on the Schwarzschild-de Sitter spacetime}
\label{CriticalCurvesDiscussion}

In order to clarify the role of the asymptotic points, in this section
we show that there are not timelike conformal geodesics reaching
$\mathcal{Q}$ and $\mathcal{Q'}$ orthogonally.  More precisely, we
show that a timelike conformal geodesic becomes asymptotically null as
it approaches $\mathcal{Q}$ and $\mathcal{Q'}$. This is in stark
tension with the required conditions for constructing a
conformal Gaussian system of coordinates in the neighbourhood of
$\mathcal{Q}$ and $\mathcal{Q'}$.

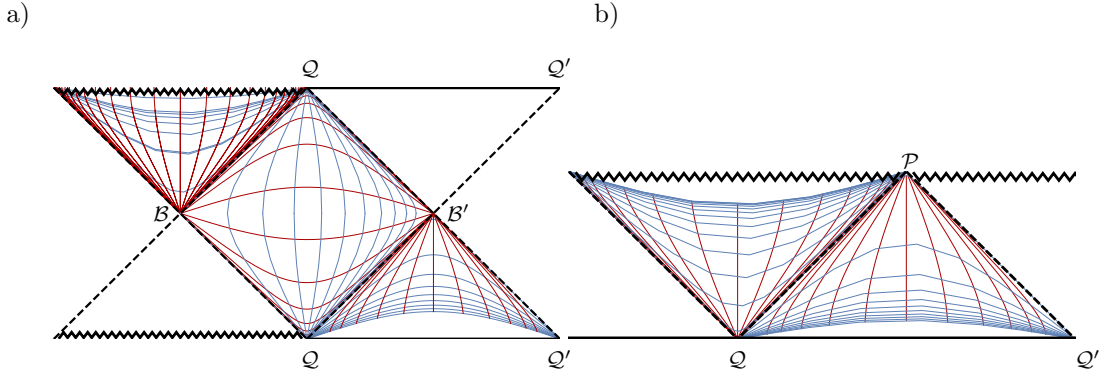
\begin{figure}[t]
\centering
\includegraphics[width=0.45\textwidth]{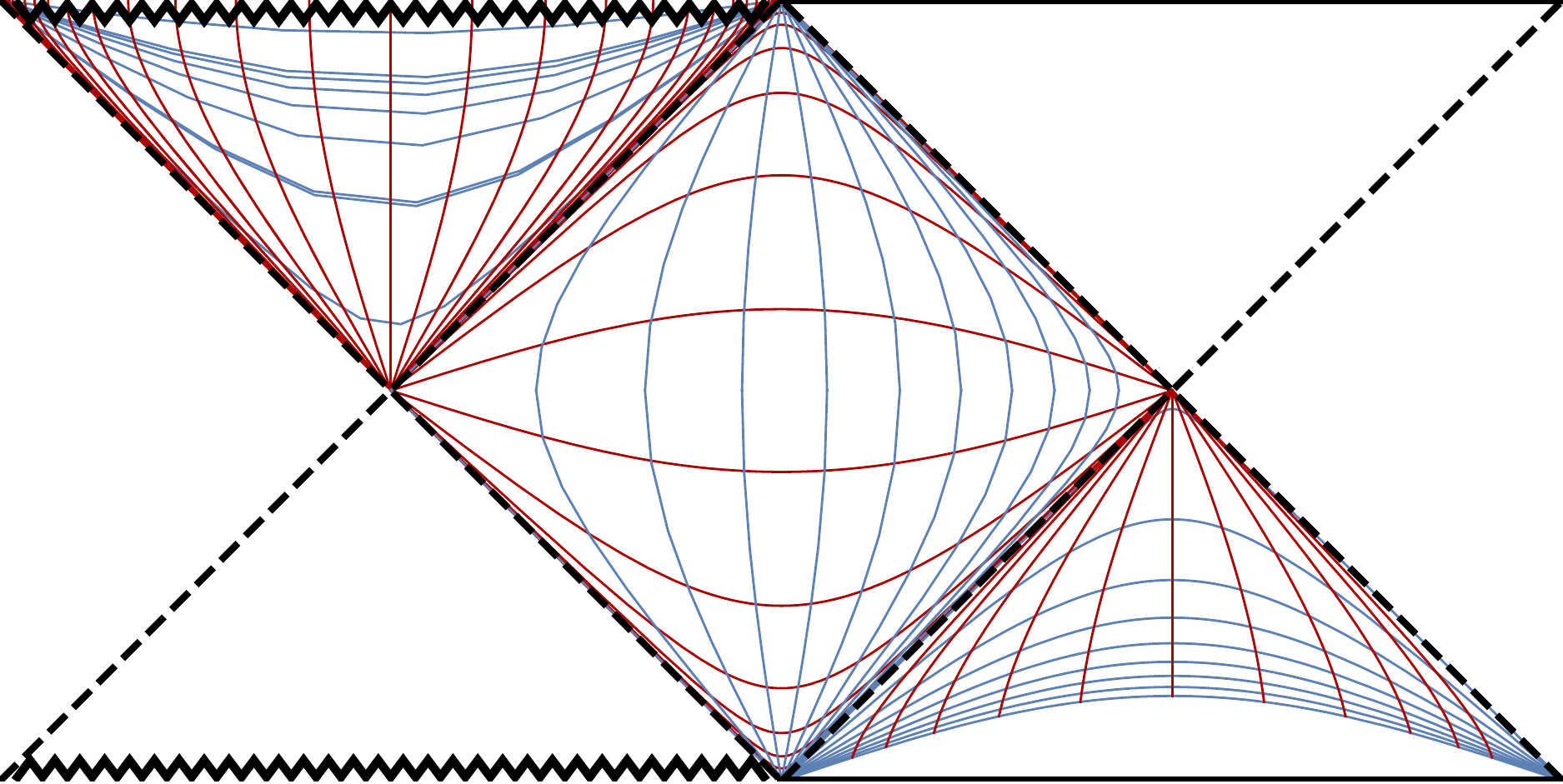}
\includegraphics[width=0.45\textwidth]{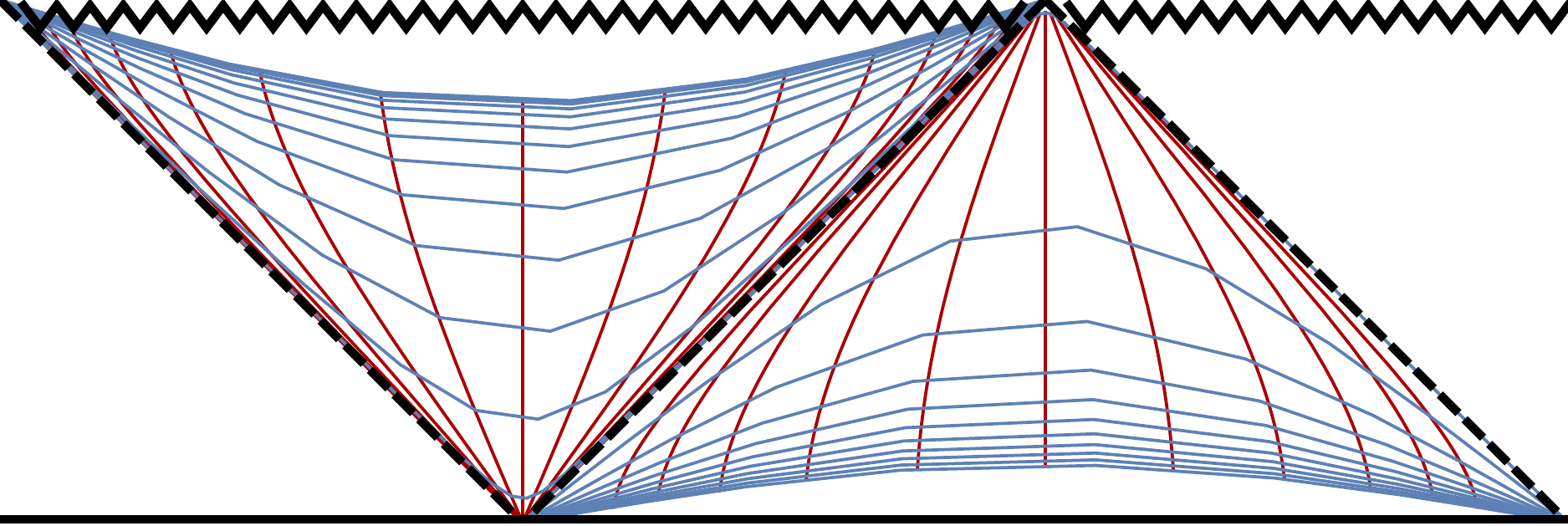}
\put(-400,120){a)}
\put(-180,120){b)}
\put(-290,100){\footnotesize{$\mathcal{Q}$}}
\put(-198,100){\footnotesize{$\mathcal{Q'}$}}
\put(-290,-10){\footnotesize{$\mathcal{Q}$}}
\put(-198,-10){\footnotesize{$\mathcal{Q'}$}}
\put(-65, 65){\footnotesize{$\mathcal{P}$}}
\put(0, -10){\footnotesize{$\mathcal{Q'}$}}
\put(-130, -10){\footnotesize{$\mathcal{Q}$}}
\put(-345,45){\footnotesize{$\mathcal{B}$}}
\put(-235,45){\footnotesize{$\mathcal{B'}$}}
\caption{Curves of constant $r$ and $t$ in the Schwarzschild-de Sitter
  spacetime.  a) Curves with constant $t$ and $r$ (red and blue
  respectively) are plotted on the Penrose diagram of the Subextremal
  Schwarzschild-de Sitter spacetime. Curves of constant $t$ accumulate
  at the bifurcation spheres $\mathcal{B}$, $\mathcal{B'}$ while the
  curves of constant $r$ accumulate at the asymptotic points
  $\mathcal{Q}$ and $\mathcal{Q'}$.  b) Curves with constant $t$ and
  $r$ (red and blue respectively) are plotted on the Penrose diagram
  of the extremal Schwarzschild-de Sitter spacetime. In contrast with
  the subextremal case, curves with constant $t$ in starting from some
  $r_{\star}<3m$ accumulate at the asymptotic points $\mathcal{Q}$ and
  $\mathcal{Q'}$ while those starting from $r_{\star}>3m$ accumulate
  at $\mathcal{P}$. The hyperextremal case is qualitatively similar to
  the extremal one and has been omitted. }
\label{fig:SdSConstantCurves}
\end{figure}

As shown in the Penrose diagram of Figure, \ref{fig:SdSConstantCurves}
in the subextremal case the
curves of constant $t=t_{\star}$ accumulate in the \emph{bifurcation
  spheres} $\mathcal{B}$ and $\mathcal{B'}
$ while the curves of constant $r$ accumulate in the
asymptotic points $\mathcal{Q}$ and $\mathcal{Q'}$. By contrast, in the
extremal case the curves with constant $t=t_{\star}$ approach the
asymptotic points $\mathcal{Q}$ and $\mathcal{Q'}$  ---see
\cite{Gey80} for an extensive discussion on the Penrose diagram for
Schwarzschild-de Sitter spacetime. It follows from the geodesic equation
\eqref{SSmetricGeodesicEquations} that the curves of constant $r$
correspond to geodesics whenever the condition
\eqref{CriticalCurveCondition} is satisfied, this equation explicitly reads
\begin{equation}
\label{CriticalCurveConditionGeneral}
|\lambda|r^3 + 3(\gamma^2-1)r+6m=0.
\end{equation}
Observe that for $\gamma=1$ the last condition reduces to
$|\lambda|r^3+6m=0$ which cannot be solved for positive $r$.

\medskip
 In this
section we perform an analysis of the behaviour of the critical curves
on the Schwarzschild-de Sitter spacetime.  Notice that in the
hyperextremal case the are no timelike geodesics with constant $r$ since for
$|\lambda|>1/9m^2$ one has strictly $F(r) <0$ so that the condition
\eqref{CriticalCurveCondition} can never be satisfied.

\subsubsection{Critical curves in the extremal Schwarzschild-de Sitter spacetime }

We start the analysis in the simpler case in which $|\lambda|=1/9m^2$
so that $F(r)$ is given as in equation \eqref{Ffor-eSdS} and the
condition \eqref{CriticalCurveCondition} reduces to considering 
$r=3m$ and $\gamma=0$. Observe
that the curves with $\gamma=0$ and $r \neq 3m$ correspond to curves with
constant $t=t_{\star}$ which, as discussed in previous paragraphs,
approach asymptotically the points $\mathcal{Q}$ and
$\mathcal{Q'}$. Notice that for $\gamma=0$ the expression
\eqref{rGeodesicIntegral} can easily be integrated to yield
\begin{equation}
\label{eSdSgeodesicNotCrossingHorizon}
  \tilde{\tau}-\tilde{\tau}_{\star}=3m \ln \left(H(r)/H(r_{\star})\right)
\end{equation}
where
\[
H(r)\equiv  \frac{\sqrt{3r}+ \sqrt{r+6m}}{(\sqrt{3r}-
  \sqrt{r+6m})(\sqrt{r}+ \sqrt{r+6m})^{2\sqrt{3}}}.
\]
Observe that equation \eqref{eSdSgeodesicNotCrossingHorizon}, as
pointed out in \cite{Pod99}, implies that the geodesics with $\gamma=0$
never cross the horizon since $\tilde{\tau}\rightarrow \infty$ as
$r\rightarrow 3m$. For simplicity, let
 $M_{\star} \equiv H(r_{\star}) + \exp(\tilde{\tau}_{\star}/3m)$
 with $r_{\star} \neq 3m$ so that
  $\tilde{\tau}= 3m \ln |H(r)/M_{\star}|$.  Reparametrising using
  equation \eqref{ReparametrisationPhysicalToUnPhysical} and that
  $|\lambda|=1/9m^2$ renders
\[
\tau(r)= \frac{2 W(r)}{M_{\star}^{p}-\kappa W(r)}
\]
with $W(r)=H(r)^{1/\sqrt{3}}$. Using L'H\^opital rule one can verify
that $\tau \rightarrow -2/\kappa$ as $r \rightarrow 3m$.  To analyse
the behaviour of these curves as they approach the points
$\mathcal{Q}$ and $\mathcal{Q'}$ let us consider $r$ such that $r=3m +
\epsilon$ . Then, one has that for small $\epsilon>0$ that
\[
W(r) = \left(\frac{m}{r-3m}\right)^{1/\sqrt{3}} \left( \frac{C_{1}}{m}
+ \frac{C_{2}}{m^2} (r-3m) + \mathcal{O}((r-3m)^2) \right)
\]
where $C_{1}$ and $C_{2}$ are numerical
factors whose explicit form is not relevant for the subsequent discussion.
 Hence, to leading order $W(r) = C/\epsilon^{p}$ where $C$
is a constant depending on $m$ only and $p=1/\sqrt{3}$. Consequently,
to leading order
\[
\frac{\mbox{d} \tau}{\mbox{d}\epsilon}= -\frac{pC \kappa \epsilon^{p-2}}{
  \left(M_{\star}^{p}\epsilon^{p}-\kappa C\right)} -\frac{pC
  \epsilon}{M_{\star}^{p}\epsilon^{p}-\kappa C}.
\]
Therefore, since $p<2$ one has that ${\mbox{d} \tau}/{\mbox{d}
  \epsilon}$ diverges as $\epsilon \rightarrow 0$ so that the curves
with $\gamma=0$ become tangent to the horizon as they approach
$\mathcal{Q}$ or $\mathcal{Q'}$ ---that is, they would have to become
null to reach $\mathcal{Q}$ or $\mathcal{Q'}$.  This is analogous to
the behaviour of the critical curve in the Schwarzschild spacetime
pointed out in \cite{Fri03a}, and the subextremal Reissner-Nordstr\"om
spacetime in \cite{LueVal13b} ---in contrast, in the extremal
Reissner-Nordstr\"om spacetime one has $\frac{\mbox{d} \tau}{\mbox{d}
  \epsilon}=0$ as $\epsilon \rightarrow 0$ as discussed in
\cite{LueVal13b} .

\subsubsection{Critical curves in the subextremal Schwarzschild-de Sitter spacetime }

For the subextremal case one could parametrise the roots of the
depressed cubic \eqref{CriticalCurveConditionGeneral} using Vieta's
formulae and choose some
 $\gamma \neq 1$ for which there is at least one
positive root. However, notice that fixing a value for $\gamma$
is equivalent to prescribe initial data for the congruence:
\begin{equation}
t(\tilde{\tau})=t_{\star}, \qquad r(\tilde{\tau})=r_{\star}, \qquad
\frac{\mbox{d}r}{\mbox{d}\tilde{\tau}}\Big\rvert_{r_{\star}}=
\sqrt{\gamma^2-F(r_{\star})}, \qquad
\frac{\mbox{d}t}{\mbox{d}\tilde{\tau}}\Big\rvert_{r_{\star}}=\frac{\gamma}{F(r_{\star})}.
\end{equation}
 Restricting our analysis to the static region $r_{b}< r_{\star}<
 r_{c}$ for which $F(r_{\star})>0$ and setting
 \[
 \frac{\mbox{d}r}{\mbox{d}\tilde{\tau}}\Big\rvert_{r_{\star}}=0,
\]
 one gets
\[
\gamma= \sqrt{F(r_{\star})},
\]
and condition \eqref{CriticalCurveCondition} is equivalent to
\[
F(r_{\star})-F(r)=\frac{ |\lambda|(r-r_{\star})}{3r} Q(r),
\]
where $Q(r)$ is the polynomial

\[
Q(r) \equiv r^2 + r_{\star}r-\frac{6m}{|\lambda|r_{\star}}.
\]
Notice that $Q(r)$ can be factorised as

\[
Q(r) = (r-\alpha_{-}(r_{\star}))(r-\alpha_{+}(r_{\star})),
\]
where
\[
\alpha_{\pm}(r_{\star}) \equiv \frac{r_{\star}}{2}\left(-1 \pm \sqrt{1 +
  \frac{24m}{|\lambda|r_{\star}^3}} \right).
\]
In addition, observe that 

\begin{eqnarray*}
Q(r_{\star})& >& 0 \qquad \text{for} \qquad
r_{\circledast}<r_{\star}<r_{c}, \\ Q(r_{\star})& < &0, \qquad
\text{for} \qquad r_{b}<r_{\star}<r_{\circledast}, \\ Q(r_{\star}) & =
&0, \qquad \text{for} \qquad r_{\star}=r_{\circledast}, \\
\end{eqnarray*}
where  $r_{\circledast} \equiv
\left(\frac{3m}{|\lambda|}\right)^{1/3}$.
In the extremal case one has $r_{b}=r_{c}=r_{\circledast}=3m$.  The
curve $r=r_{\circledast}$, as in the extremal case, will be called the
\emph{critical curve}.  With the above notation the integral
\eqref{rGeodesicIntegral} can be then rewritten as
\begin{equation}\label{rGeodesicIntegralCriticalCurve}
\tilde{\tau}-\tilde{\tau}_{\star}=
\int^{r}_{r_{\star}}\sqrt{\frac{s}{(s-r_{\star})(s-\alpha_{-}(r_{\star}))(s-\alpha_{+}(r_{\star}))}}\mbox{d}s.
\end{equation}
To study the behaviour close to the critical curve consider
$r_{\star}=(1+ \epsilon)r_{\circledast}$ For small $\epsilon>0$ and
considering $s>r_{\star}$ one can expand the right hand side of
equation \eqref{rGeodesicIntegralCriticalCurve} in Taylor series as
\begin{equation}
\tilde{\tau}-\tilde{\tau}_{\star}=
\int_{r_{\star}}^{r}\sqrt{\frac{s}{s+2r_{\circledast}}}
\left(\frac{1}{s-r_{\circledast}}-\frac{3r_{\circledast}^2s\epsilon^2
}{2(s-r_{\circledast})^3}\right)\mbox{d}s + \mathcal{O}(\epsilon^3).
\end{equation}
Integrating we obtain
\begin{eqnarray*}
\tilde{\tau}-\tilde{\tau}_{\star}=
-\frac{2}{\sqrt{3}}\text{arctanh}\left( \sqrt{3}\sqrt{\frac{1+
    \epsilon}{3+\epsilon}}\right) + 2\ln\left(
\sqrt{r_{\circledast}(1+\epsilon)}+\sqrt{r_{\circledast}(3+\epsilon)}
\right) -\frac{2}{\sqrt{3}}\text{arctanh} \left(
\frac{3r}{r+2r_{\circledast}} \right) \\ + 2\ln \left(\sqrt{r} +
\sqrt{r + 2r_{\circledast}} \right)
\ -\frac{3}{4}r_{\circledast}\sqrt{1+2r_{\circledast}}(1+ 2\epsilon)
-\frac{3}{4}r_{\circledast}{}^2\sqrt{1+2r_{\circledast}}\frac{(2r-r_{\circledast})\epsilon^2}{(r_{\circledast}-r)^2}
+ \mathcal{O}(\epsilon^3).
\end{eqnarray*}
As $\epsilon \rightarrow 0$ the last expression diverges ---as is to
be expected. The divergent term can be expanded for small $\epsilon>0$
as
\[
\text{arctanh} \left(\sqrt{3}\sqrt{\frac{1+\epsilon}{3+\epsilon}}
\right) = \frac{1}{2}\ln \left(  \left\lvert-\frac{6}{\epsilon} + 4 +
\frac{\epsilon}{6}+ \mathcal{O}(\epsilon^2) \right\rvert \right)
\]
and the second term can be expanded as
\[
 \ln\left(
 \sqrt{r_{\circledast}(1+\epsilon)}+\sqrt{r_{\circledast}(3+\epsilon)}
 \right) = \ln \left( (1 + \sqrt{3})\sqrt{r_{\circledast}}\right) +
 \frac{\epsilon}{2\sqrt{3}} - \frac{\epsilon^2}{6\sqrt{3}} +
 \mathcal{O}(\epsilon^3).
\]
Hence, to leading order one has
\[
\tilde{\tau}(r) =\frac{1}{\sqrt{3}}\ln \epsilon + f(r) +
\mathcal{O}(\epsilon)
\]
where
\[
f(r)=\tilde{\tau}_{\star} + 2\ln \left( (1 +
\sqrt{3})\sqrt{r_{\circledast}}\right) -\frac{2}{\sqrt{3}}
\text{arctanh} \left( \frac{3r}{r+2r_{\circledast}} \right) \\ + 2\ln
\left(\sqrt{r} + \sqrt{r + 2r_{\circledast}} \right)
\ -\frac{3}{4}r_{\circledast}\sqrt{1+2r_{\circledast}}.
\]
Reparametrising respect to the unphysical proper time using
\eqref{ReparametrisationPhysicalToUnPhysical} one gets
\[
\tau(r)=\frac{2\exp\left(\sqrt{|\lambda|/3}f(r)  +
  \mathcal{O}(\epsilon)\right) \epsilon^{p}} {1-\kappa
  \exp\left(\sqrt{|\lambda|/3}f(r) + \mathcal{O}(\epsilon)\right)\epsilon^{p} }
\]
with $p=1/\sqrt{3}$. Thus one gets

\[
\frac{\mbox{d} \tau}{\mbox{d} \epsilon}
= \frac{2p\exp\left(\sqrt{|\lambda|/3}f(r) +
  \mathcal{O}(\epsilon)\right)\epsilon^{p-1}}{\left(1 -
  \exp\left(\sqrt{|\lambda|/3}f(r) +
  \mathcal{O}(\epsilon)\right)\kappa\epsilon^{p} \right)^2}.
\]
Observe that since $p<1$ then one has that
 $\mbox{d}\tau/\mbox{d}\epsilon$ diverges
as $\epsilon \rightarrow 0$.

\section{Appendix:
The conformal evolution equations in the case $\kappa \neq 0$ 
and reparametrisations }
\label{FormationOfSingularitiesAndReparametrisations}

In Section \ref{AnalysisOfTheCoreSystem} we analysed
  the case $\kappa=0$ ---this corresponds to a conformal boundary with
  vanishing extrinsic curvature. Nevertheless, as discussed in
  Section \ref{AsymptoticInitialValueProblem}, $\kappa$ is a conformal
  gauge quantity arising from the conformal covariance of the conformal
  field equations. Consequently, it is of interest to analyse the
  behaviour of the core system in the case $\kappa \neq 0$.  For
  simplicity, in the remainder of this section, $\kappa$ will be
  assumed to be a constant on the initial hypersurface corresponding
  to $\tau=0$. In first instance, we restrict our attention to
  $|\kappa| >1$ and then
  discuss how to exploit the conformal covariance of
  the equations to extend these results for $\kappa \in
  [-1,0)\cup(0,1]$. 

\subsection{Analysis of the core system with $\kappa>1$}
\label{AnalysisKappaPos}

We begin the discussion of this case observing that, for
$\kappa>1$, one has that $\Theta(\tau) \geq 0$ and $\dot{\Theta}(\tau)
> 0$ for $\tau \geq 0$.  Using this simple observation and the core
equations \eqref{Core1}-\eqref{Core3} we obtain the following:

\begin{lemma} \label{AnalysisKappaPosPropL}
For a solution to the core system \eqref{Core1}-\eqref{Core3} with
initial data given by \eqref{InitialDataCoreSpherical} and $\kappa>1$
one has that $L(\tau) <0$ for $\tau \geq 0$.
\end{lemma}

\begin{proof}
We proceed by contradiction. Assume that there exists
$0<\tau_{L}<\infty$ such that $L(\tau_{L})=0$. Without loss of
generality we can assume that $\tau_{L}$ corresponds to the first zero
of $L(\tau)$. Since for $\kappa>1$ we have $L(0)<0$ then by continuity
it follows that  $\dot{L}(\tau_{L}) \geq
0$ ---$\dot{{L}}(\tau_{L})$ cannot be negative since this would
  imply that ${L}(\tau)$ crossed the $\tau$-axis at some time $\tau <
  \tau_{L}$ but this is not possible since $\tau_{L}$ is the first
  zero of ${L}(\tau)$. It follows then from equation \eqref{Core3} that
\[ 
0 \leq \dot{L}(\tau_{L})=-\chi(\tau_{L})L(\tau_{L})
-\frac{1}{2}\dot{\Theta}(\tau_{L})\phi(\tau_{L}). 
\]
Since $L(\tau_L)=0$ and $\dot{\Theta}(\tau_L)>0$,  the last
inequality implies that $\phi(\tau_L)\leq0$ but this is a
contradiction since we already know from Observation 1 that
$\phi(\tau)>0$ for any $\tau$.
\end{proof}

\medskip
\noindent 
\textbf{Observation 4.} Using that $\dot{\Theta}(\tau) \geq
0$ for $\kappa>1$ and $\tau \geq 0$ and that $\phi(\tau)>0$ we obtain
 from equation \eqref{Core3} the differential inequality
\[ 
\dot{L}(\tau) \leq -\chi(\tau)L(\tau).
\]
 Observing  Lemma
\ref{AnalysisKappaPosPropL} we have that $L(\tau) < 0$. Thus, we can
formally integrate the last differential inequality and obtain
\begin{equation}
\label{EstimateL}
 L(\tau) \leq L(0)\exp\left(-\int^{\tau}_0 \chi(\mbox{s})\mbox{d}\mbox{s} \right).
\end{equation}

\medskip
We now show that the function $\chi(\tau)$ which is initially positive
must necessarily have a zero. 

\begin{lemma}
\label{AnalysisKappaPosChiVanishesSomewhere}
For a solution to the core system \eqref{Core1}-\eqref{Core3} with
initial data given by \eqref{InitialDataCoreSpherical} and $\kappa>1$
there exist $0<\tau_{\chi}<\infty$ such that $\chi(\tau_{\chi})=0$.
\end{lemma}

\begin{proof}
We proceed again by contradiction. Assume that $\chi(\tau)$ never
vanishes. Since $\chi(0)=\kappa>0$ then $\chi(\tau)>0$ for $\tau \geq
0$. From Lemma \ref{AnalysisKappaPosPropL}  we know that $L(\tau)<0$. 
In addition, we know that $\Theta(\tau)\phi(\tau) \geq 0$.
 With these observations equation
\eqref{Core2} gives
\[ 
\dot{\chi}(\tau) < -\chi^2(\tau) \hspace{0.5cm} \text{for} \hspace{0.5cm}
 \tau > 0.
\]
Since we are assuming that $\chi(\tau)$ never vanishes then
\[  
\frac{\dot{\chi}(\tau)}{\chi^2(\tau)} < -1. 
\]
 Integrating from $0$ to $\tau>0$ and using the initial
data \eqref{InitialDataCoreSpherical} we get
\begin{equation}
\label{EstimateContradiction}
 \chi(\tau) < \frac{1}{\tau + {1}/{\kappa}} \hspace{0.5cm}
 \text{for} \hspace{0.5cm} \tau > 0.
\end{equation}
 In a similar way, we can consider equation \eqref{Core2}  and
obtain the differential inequality
\[ 
\dot{\chi}(\tau) < -\frac{1}{2}\Theta(\tau)\phi(\tau) 
\hspace{0.5cm} \text{for} \hspace{0.5cm} \tau \geq 0.
\]
Using now equation \eqref{FormalSolutionE} we get
\[ 
\dot{\chi} < -m\Theta(\tau) \exp\left({-3\int^\tau _0 
\chi(\mbox{s})\mbox{d}\mbox{s}}\right)  \hspace{0.5cm}
 \text{for} \hspace{0.5cm} \tau \geq
0.
\]
Integrating the from $0$ to $\tau>0$ we get
\begin{equation}
 \chi(\tau) < \kappa - m\int^{\tau} _{0}\Theta(\mbox{s}) 
\exp\left({-3\int^{s} _0 \chi(\mbox{s}')\mbox{ds}'}\right)\mbox{ds}  \hspace{0.5cm} 
\text{for} \hspace{0.5cm} \tau \geq 0.
\label{AnalysisKappaPosChiVanishesSomewhereIntermmediate}
\end{equation}

On the other hand, integrating expression
 \eqref{EstimateContradiction} we have
\[
\int^\tau _0 \chi(\mbox{s})\mbox{ds}<  \ln \left(\kappa\tau + 1 \right). 
\]
Consequently,
\[
-m\Theta(\tau) \exp\left({-3\int^\tau _0 \chi(\mbox{s}')\mbox{ds}'}\right)<
 -m\sqrt{\frac{|\lambda|}{3}}\frac{\tau(1+\tfrac{1}{2}\kappa \tau)}
{(1+\kappa \tau)^3}.
 \]
Integrating we get
\[
- m\int^{\tau} _{0}\Theta(\mbox{s}) \exp\left({-3\int^{s} _0
  \chi(\mbox{s}')\mbox{ds}'}\right)\mbox{ds} < -
\frac{m}{2\kappa^2}\sqrt{\frac{|\lambda|}{3}}\left(\frac{1}{(\kappa \tau
  + 1)^2} +\ln(\kappa \tau + 1) -1\right).
\]
Substituting the above result into the inequality
\eqref{AnalysisKappaPosChiVanishesSomewhereIntermmediate} we obtain 
\[ 
\chi(\tau) < \kappa -
\frac{m}{2\kappa^2}\sqrt{\frac{|\lambda|}{3}}\left(\frac{1}{(\kappa \tau
  + 1)^2} +\ln(\kappa \tau + 1) -1\right).
\]
 The right hand side of the last expression becomes negative for some
 sufficiently large $\tau$. This is a contradiction as we have assumed
 that $\chi(\tau)$ never vanishes and $\chi(0)>0$.
\end{proof}

\medskip
\noindent 
\textbf{Observation 5.} Combining Lemma
\ref{AnalysisKappaPosPropL} and Observation 1, we conclude that $L(\tau)< 0$ and
$\Theta(\tau)\phi(\tau)>0$ for $\tau>0$. Using these properties in equation
\eqref{Core2} we get
 \[
\dot{\chi}(\tau)< 0 \hspace{0.5cm} 
\text{for} \hspace{0.5cm}\tau \geq 0.
\] 
Thus, $\chi(\tau)$ is always
decreasing. From Lemma \ref{AnalysisKappaPosChiVanishesSomewhere} 
 we know that there exists a finite
$\tau_{\chi}>0$ such that $\chi(\tau_{\chi})=0$. Then, by continuity, for
any $\tau >\tau_{\chi}$ we have that
$\chi(\tau)<0$ .

\medskip
With this last observation we are in the position of proving the main
result of this section:

\begin{proposition}
\label{Proposition:FormationSingularitiesKappaPositive}
There exists $0<\tau_{\lightning}<\infty$ such that the solution of
\eqref{Core1}-\eqref{Core3} with initial data given by 
  \eqref{InitialDataCoreSpherical} and $\kappa>1$ satisfies 
\[
\chi\rightarrow -\infty, \quad  L \rightarrow -\infty, \quad  \phi
  \rightarrow \infty \quad  \mbox{as} \quad \tau \rightarrow \tau_{\lightning}.
\]
\end{proposition}

\begin{proof}
From Lemma \ref{AnalysisKappaPosChiVanishesSomewhere} we know
 there exists a finite $\tau_{\chi}$ for which $\chi(\tau)$
vanishes. By Observation 5, we have that $\chi(\tau_{\lozenge})<0$ for
any $\tau_{\lozenge}> \tau_{\chi}$. Let $\chi_{\lozenge}\equiv \chi(\tau_{\lozenge})
<0$. We can assume that $\chi_{\lozenge}$ is finite,
otherwise there is nothing to prove.  Now, using Lemma \ref{AnalysisKappaPosPropL} and
that $\Theta(\tau)\phi(\tau)>0$ we get
\[ 
\dot{\chi}(\tau) < -\chi^2(\tau) \hspace{0.5cm}
 \text{for} \hspace{0.5cm} \tau \geq 0.
\]
\noindent Since we know that $\chi(\tau) < 0$ for any $\tau>
\tau_{\lozenge}$ then
\[  
\frac{\dot{\chi}(\tau)}{\chi^2(\tau)} < -1. 
\]
\noindent Integrating form $\tau=\tau_{\lozenge}$ to $\tau>
\tau_{\lozenge}$ we get
\begin{equation}
 \chi(\tau) < \frac{1}{\tau - \tau_{\lozenge} +
   {1}/{\chi_{\lozenge}}} \hspace{0.5cm}
 \text{for} \hspace{0.5cm} \tau >
 \tau_{\lozenge}. \label{BlowUpSpherical}
\end{equation}
 From inequality \eqref{BlowUpSpherical} we can conclude that
$\chi(\tau)\rightarrow -\infty$ for some finite time
$\tau_{\lightning}< \tau_{\lozenge} - {1}/\chi_{\lozenge} $.
  Additionally, observe that $\tau_{\lozenge} -
  {1}/{\chi_{\lozenge}} > \tau_{\lozenge}>0$ since
  $\chi_{\lozenge}<0$ . Now, given that $\chi\rightarrow -\infty$ as $\tau
  \rightarrow \tau_{\lightning}$ it follows from equation
  \eqref{FormalSolutionE} that $\phi \rightarrow \infty$ as $\tau
  \rightarrow \tau_{\lightning}$.  Similarly, from inequality
  \eqref{EstimateL} and that $L(0)<0$ it follows that $L \rightarrow
  -\infty$ as $\tau \rightarrow \tau_{\lightning}$.
\end{proof}

\begin{figure}[t]
\centering

\includegraphics[width=0.65\textwidth]{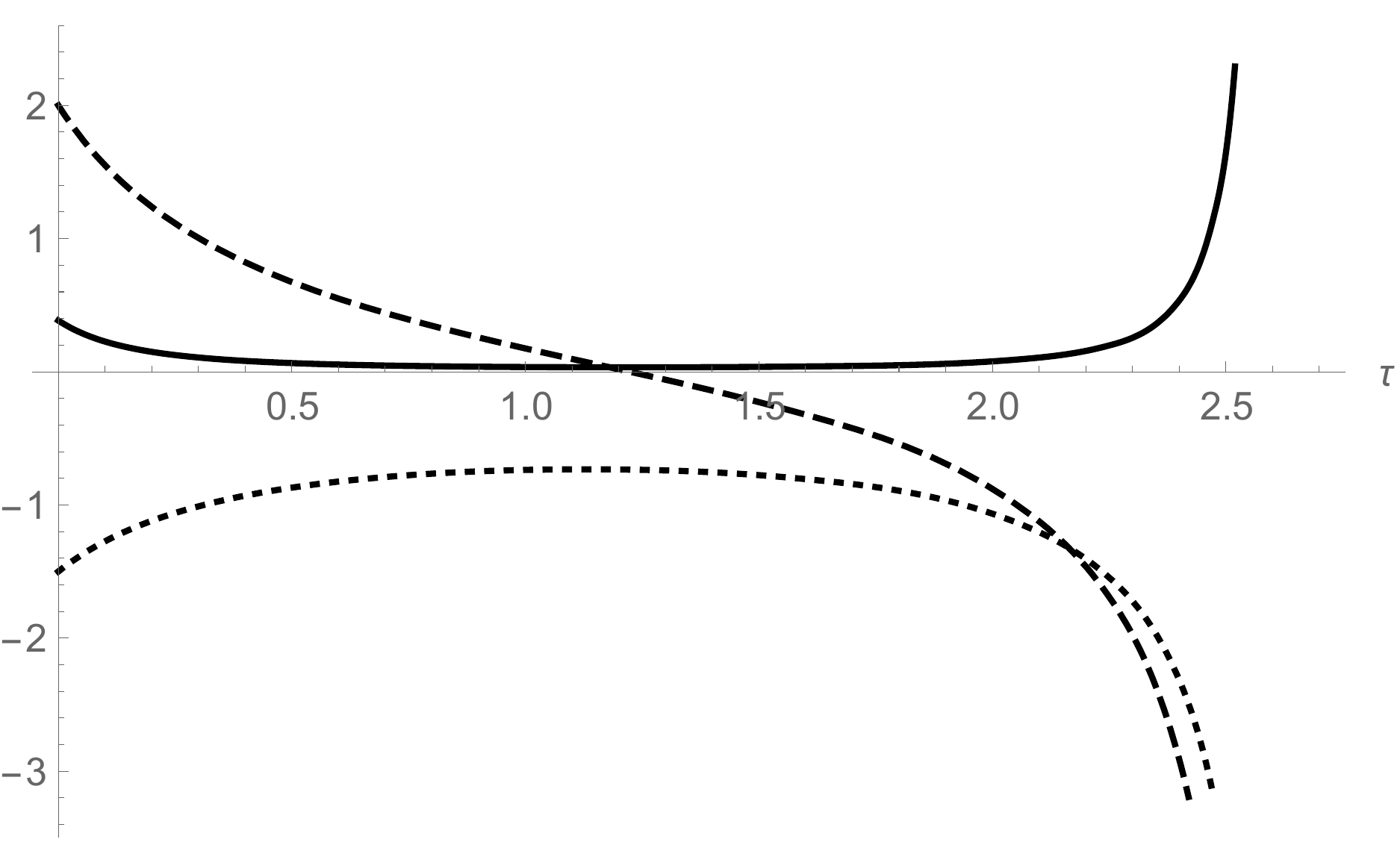}
   \caption[]{Numerical solutions of the core system
  \eqref{Core1}-\eqref{Core3} with initial data given by
  \eqref{InitialDataCoreSpherical} in the case $\kappa=2$ and
  $|\lambda|=3$, $m=1/3\sqrt{3}$.  The solid line describes the evolution of
  $\phi$, the dashed line that of $\chi$ and the dotted line that of
  $L$.  One can observe the formation of a singularity at $\tau
  \approx 2.6392$.}
\label{fig:figPos}
\end{figure}

\begin{remark}
{\em A plot of the numerical evaluation of the solutions
to the core system \eqref{Core1}-\eqref{Core3} with initial data
  \eqref{InitialDataCoreSpherical} in the case $\kappa>1$ can be seen
  in Figure \ref{fig:figPos}. } 
\end{remark}

\subsection{Analysis of the core system with $\kappa <-1$}
\label{AnalysisKappaNeg}

In this section we use a similar approach to that followed in Section
\ref{AnalysisKappaPos} to show that the fields in the core system
diverge for some finite time if $\kappa<-1$.  An interesting feature
of this case is that, assuming one knows that there exists a
singularity in the development, there exists an \emph{a priori} upper
bound for the time of its appearance ---namely, the location of second
component of the conformal boundary at $\tau = 2/|\kappa|$. As a
byproduct of the analysis of this section an improvement of this basic
bound is obtained.

\medskip
An important remark concerning the case $\kappa<-1$ is that if $\tau
\in [0,1/|\kappa|]$ then both
$\Theta(\tau)$ and $\dot{\Theta}(\tau)$ are non-negative. Based on
this observation our first result is:

\begin{lemma}\label{AnalysisKappaNegPropL}
If $\kappa<-1$ then the solution to the core system
\eqref{Core1}-\eqref{Core3} with initial data
\eqref{InitialDataCoreSpherical} satisfies $L(\tau) < 0$ for $\tau \in [0,1/|\kappa|]$.
\end{lemma}

\begin{proof}
We proceed by contradiction. Assume that there exists $0<\tau_{L} \leq
1/|\kappa|$ such that $L(\tau_{L})=0$. Without lost of generality we
can assume that $\tau_{L}$ is the first zero of $L(\tau)$. Since
$L(0)<0$ for $\kappa<-1$ then by continuity $\dot{L}(\tau_{L})\geq0.$
Therefore, proceeding as in Lemma \ref{AnalysisKappaPosPropL} one gets
from \eqref{Core3}
\[ 
0 \leq \dot{L}(\tau_{L})=-\chi(\tau_{L})L(\tau_{L})
-\frac{1}{2}\dot{\Theta}(\tau_{L})\phi(\tau_{L})
\hspace{0.5cm}\text{for}\hspace{0.5cm} \tau \in [0,1/|\kappa|]. 
\]
Since $L(\tau_L)=0$ and $\dot{\Theta}(\tau_L)>0$ the last
inequality implies that $\phi(\tau_L)\leq0$. This is a
contradiction since $\phi(\tau)>0$ ---cfr. Observation 1.
\end{proof}

\begin{lemma}\label{AnalysisKappaNegPropChi1}
If $\kappa<-1$ then the solution to the core system
\eqref{Core1}-\eqref{Core3} with initial data
\eqref{InitialDataCoreSpherical} satisfies $\chi(\tau)<0$ for $\tau
 \in[0,1/|\kappa|]$.
\end{lemma}

\begin{proof}
Again, we proceed by contradiction. Assume that there exists
$0<\tau_{\chi}\leq 1/|\kappa|$ such that
$\chi(\tau_{\chi})=0$. Without lost of generality we can assume that
$\tau_{\chi}$ is  the first zero of
$\chi(\tau)$. Then, by continuity, we have that $\dot{\chi}(\tau_{\chi})
\geq 0$. Using equation \eqref{Core2} one has
\[
 0 \leq \dot{\chi}(\tau_{\chi})=-\chi(\tau_{\chi})^2 + L(\tau_{\chi})-
\frac{1}{2}\Theta(\tau_{\chi})\phi(\tau_{\chi}) \hspace{0.5cm}\text{for}
\hspace{0.5cm}
\tau \in [0,1/|\kappa|].
 \]
Therefore, since $\chi(\tau_{\chi})=0$ one has
\[ 
L(\tau_{\chi}) \geq \frac{1}{2}
\Theta(\tau_{\chi})\phi(\tau_{\chi}) > 0.
\]
This is a contradiction since by Lemma
\ref{AnalysisKappaNegPropL} we know that $L(\tau) < 0$ for $\tau
\in[0, 1/|\kappa|]$.
\end{proof}

\medskip
\noindent 
\textbf{Observation 6}. Proceeding as in
Observation 4 one readily has that for $\kappa<-1$
\[
 L(\tau) \leq L(0) \exp\left( -\int_{0}^{\tau} \chi(\mbox{s})\mbox{ds} \right) 
\text{\hspace{0.5cm} for \hspace{0.5cm}} \tau \in(0, 1/|\kappa|].
\]

This last observation is used, in turn, to prove the main result of
this section:

\begin{proposition}
\label{Proposition:FormationSingularitiesKappaNegative}
If  $\kappa<-1$, then for the solution of \eqref{Core1}-\eqref{Core3} with initial data
  \eqref{InitialDataCoreSpherical} there
  exists $0<\tau_{\lightning}<1/|\kappa|$ such that 
\[
\chi(\tau)\rightarrow -\infty, \quad L(\tau) \rightarrow -\infty, \quad
\mbox{and} \quad 
  \phi(\tau) \rightarrow \infty \quad  \mbox{as} \quad \tau \rightarrow \tau_{\lightning}.
\]
\end{proposition}

\begin{proof}
Consider equation \eqref{Core2} on the interval $\tau \in
[0,1/|\kappa|]$. Using Lemma \ref{AnalysisKappaNegPropL} we
know that $L(\tau)<0$. This observation and the fact that
$\phi(\tau)>0$ leads to the differential inequality
\[ 
\dot{\chi}(\tau)<-\chi^2(\tau) \hspace{0.5cm} \text{for}\hspace{0.5cm}
 \tau \in [0,1/|\kappa|].
\]
 Since by Lemma \ref{AnalysisKappaNegPropChi1}, we know
that $\chi(\tau) \neq 0$ for $\tau \in [0,1/|\kappa|]$ we can rewrite
the last expression as
\[ \frac{\dot{\chi}(\tau)}{\chi^2(\tau)} <-1 \hspace{0.5cm} \text{for} \hspace{0.5cm} 
\tau \in[0,1/|\kappa|].\]
 Integrating from $\tau=0$ to $1/|\kappa|$ and
using the initial data \eqref{InitialDataCoreSpherical} we get
\begin{equation} 
\label{AnalysisKappaNegPropChiInequality}
\chi(\tau) < \frac{1}{\tau} -\frac{1}{|\kappa|}.
\end{equation}
 From inequality \eqref{AnalysisKappaNegPropChiInequality}
one concludes that $\chi(\tau) \rightarrow -\infty$ for some
$0<\tau_{\lightning} \leq 1/|\kappa|$. Finally, using Observation 6
and Observation 1 one concludes that $L(\tau) \rightarrow -\infty$ and
$\phi(\tau) \rightarrow \infty$ as $\tau \rightarrow
\tau_{\lightning}$ for some $0<\tau_{\lightning} \leq 1/|\kappa|$.
\end{proof}
Notice that this upper bound for the location of the singularity is not trivial 
and improves the basic bound $\tau \leq 2/|\kappa|$ given by the location 
of the second component of the conformal boundary.

\begin{figure}[t]

\centering
\includegraphics[width=0.65\textwidth]{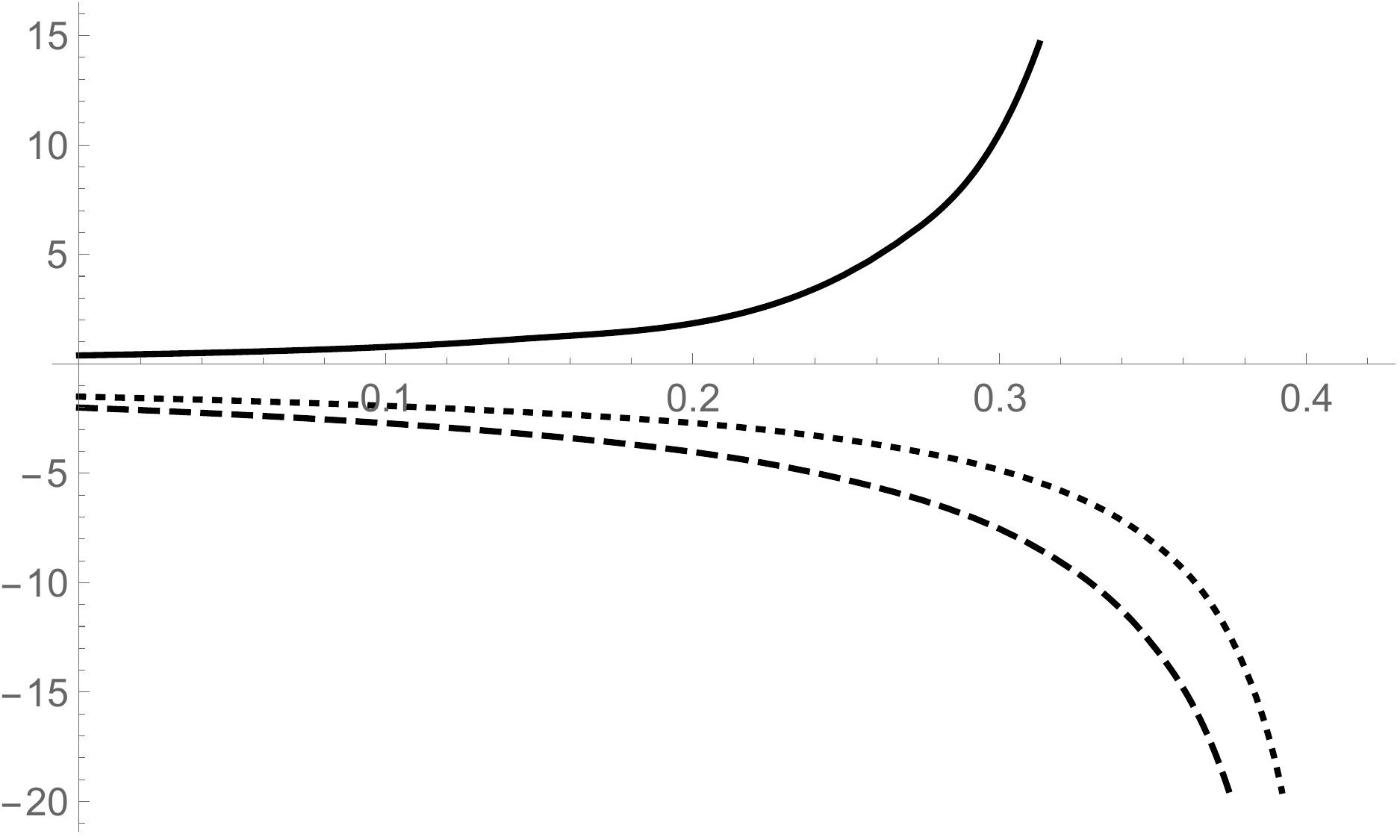}

\caption[]
{Numerical solution of the core system
  \eqref{Core1}-\eqref{Core3} with $|\lambda|=3$, $m=1/3\sqrt{3}$
 in the case $\kappa=-2$. 
  The solid line corresponds to $\phi$, the dashed line to $\chi$ and
  the dotted line to $L$.  One can observe a singularity 
at $\tau\approx0.4203$. }
\label{fig:Negative}
\end{figure}

\begin{remark}
A plot of the numerical evaluation of the solutions
to the core system \eqref{Core1}-\eqref{Core3} with initial data
  \eqref{InitialDataCoreSpherical} in the case $\kappa<-1$ can be seen
  in Figure \ref{fig:Negative}.  
\end{remark}

\subsection{Exploiting the conformal gauge}
\label{ExploitingTheConformalGaugeFreedom}

In Lemma \ref{ForeSdSfVanishes} we have shown that if
$\partial_{\psi}\kappa=0$ then the evolution equations imply, in
particular, $f_{x}=0$. Due to the spherical symmetry Ansatz, the
component $f_{x}$ is the only potentially non-zero component of
$\bmf$.  Thus, one concludes that $\bmf=0$. In Section
\ref{VanishingFields} we will exploit
this feature of the Weyl connection
 to extract further information about $\kappa$ and $s$. These
results are used in Section \ref{ChangingKappa} to discuss
the conformal gauge freedom of the extended conformal field equations
and the role played by reparametrisations of conformal
geodesics. 


\subsubsection{The relation between the Weyl and Levi-Civita connections}
\label{VanishingFields}

 As discussed in Section 
\ref{sec:FrameFormalismAndWeylConnections}, the Weyl connection
$\hat{\bmnabla}$ expressing the extended conformal field equations is
related to the Levi-Civita connection $\bmnabla$ of the unphysical
metric $\bmg$ via the 1-form $\bmf$. If $\bmf$ vanishes then
$\hat{\bmnabla}=\bmnabla$. Exploiting this simple observation we
obtain the following results:

\begin{lemma}
\label{propSconstantAlongGeodesics}
If $\bmf=0$ then the
conformal gauge conditions \eqref{protoGauge} and
\eqref{weylPropagationFrame} imply that $s=\ddot{\Theta}$.
 Moreover,  $s$ is constant along the conformal geodesics. 
\end{lemma}

\begin{proof}
As discussed in Section \ref{sec:FrameFormalismAndWeylConnections} if
$\bmf=0$ then $\hat{L}_{ab}=L_{ab}$ and
$\hat{\Gamma}_{\bma}{}^{\bmc}{}_{\bmb}=\Gamma_{\bma}{}^{\bmc}{}_{\bmb}$. Using
the conformal gauge condition \eqref{protoGauge} it follows that
$L_{\bm0 \bma}=0$ and $\Gamma_{\bm0}{}^{\bma}{}_{\bmb}=0$. Now, the
standard conformal field equations \eqref{CEFEsecondderivativeCF} and
\eqref{CEFEs} in Appendix \ref{Appendix:CFE} give
\begin{subequations}
\begin{eqnarray}
&& \nabla_{\bm0}\nabla_{\bm0}\Theta + \Theta L_{\bm0 \bm0}-s\eta_{\bm0
    \bm0}=0,
 \label{UsingCEFE1}\\ 
&& \nabla_{\bm0}s=-L_{\bm0
   \bmb}\nabla^{\bmb}\Theta. \label{UsingCEFE2}
\end{eqnarray}
\end{subequations}
 Using $L_{\bm0 \bma}=0$ and $\Gamma_{\bm0}{}^{\bma}{}_{\bmb}=0$ in
 equation \eqref{UsingCEFE1} one concludes
 $\ddot{\Theta}=s$. Similarly, from equation \eqref{UsingCEFE2} one
 gets $\dot{s}=0$. Therefore $s$ is constant along the conformal
 geodesics.
\end{proof}

\begin{remark}
In the asymptotic initial value problem the initial
value of $s$ is given by $s_{\star}=\sqrt{|\lambda|/3}\kappa$ ---see
equation \eqref{SolutionContraintsGeneral1}. Thus, if $\bmf=0$ then
$s=\sqrt{|\lambda|/3}\kappa$ along the conformal geodesics.
\end{remark}

\medskip
Finally, one has the following:

\begin{lemma}\label{propKappaConstant}
In the asymptotic initial value problem, if $\bmf=0$, then the
conformal gauge conditions \eqref{protoGauge} and
\eqref{weylPropagationFrame} together with the conformal Einstein field
equations imply that  $\bme_{\bmi}(\kappa)=0$ ---that is, $\kappa$ is a constant.
\end{lemma}

\begin{proof}
 Using $\bmf=0$ and the gauge conditions \eqref{protoGauge} we
 get from the conformal field equation \eqref{ConstantLambda} that
\begin{equation}
\label{constanLambdaEquation}
6\Theta s -3\dot{\Theta}^2 + 3 \delta^{\bmi \bmj}\bme_{\bmi}\Theta\bme_{\bmj}
\Theta=\lambda.
\end{equation}
 Using Lemma \ref{propSconstantAlongGeodesics} we have
 $s=\ddot{\Theta}$. Therefore,  substituting $\Theta(\tau) =
 \dot{\Theta}_\star \tau(1 + \kappa\tau/2)$ into equation 
\eqref{constanLambdaEquation} and
 recalling $\dot{\Theta}_{\star}=\sqrt{|\lambda|/3}$ we obtain
\begin{equation*} 
\tau^4 \delta^{\bmi\bmj}\bme_{\bmi}(\kappa) \bme_{\bmj}(\kappa)=0.
\end{equation*}
 Observe that the last equation is trivially satisfied on
 $\mathscr{I}$ as $\tau=0$. Off the initial hypersurface, where 
$\tau \neq 0$, the last equation implies
\[ \delta^{\bmi \bmj}\bme_{\bmi}(\kappa) \bme_{\bmj}(\kappa)=0.\]
 Therefore, we conclude that $\bme_{\bmi}(\kappa)=0$.
\end{proof}

\subsubsection{Changing the conformal gauge}
\label{ChangingKappa}

 The analysis of the core system given in Sections
 \ref{AnalysisKappaPos}, \ref{AnalysisKappaNeg} and
 Section \ref{AnalysisOfTheCoreSystem} covers the cases for which
 $|\kappa|>1$ and $\kappa=0$. As a consequence of the conformal
 covariance of the extended conformal Einstein field equations one has
 the freedom of performing conformal rescalings and of reparametrising
 the conformal geodesics ---thus, effectively changing the
 representative of the conformal class $[\tilde{\bmg}]$ one is working
 with.  This conformal freedom can be exploited to extend the analysis
 given in Sections \ref{AnalysisKappaPos} and \ref{AnalysisKappaNeg}
 to the case where
 $\kappa \in [-1,0) \cup (0,1]$.

Following the discussion in the previous paragraph, any two spacetimes
$(\mathcal{M},\bmg)$ and $(\bar{\mathcal{M}},\bar{\bmg})$ with $\bmg=
\Theta^2 \tilde{\bmg}$ and $\bar{\bmg}=\bar{\Theta}^2\tilde{\bmg}$
representing two solutions to the extended conformal Einstein field
equations for different choices of parameter $\kappa$ are conformally
related.  From Lemmas \ref{lemmaCGreparametrisations} and
\ref{LemmaCF} we have that
\begin{equation}\label{TransfTheta}
\Theta(\tau)=\sqrt{\frac{|\lambda|}{3}}\tau\Big(1+\frac{1}{2}\kappa
\tau\Big), \hspace{1cm}
\bar{\Theta}(\bar{\tau})=\sqrt{\frac{|\lambda|}{3}}\bar{\tau}
\Big(1+\frac{1}{2}\bar{\kappa}
\bar{\tau}\Big),
\end{equation}
 with
\begin{equation}\label{FracTrans}
 \bar{\tau}= \frac{a \tau}{c \tau + d}.
\end{equation}

The free parameter $b$ in the fractional transformation of Lemma
\ref{lemmaCGreparametrisations} has been set to $b=0$ in order to
ensure that $\Theta$ and $\bar{\Theta}$ vanish at $\tau=0$ and
$\bar{\tau}=0$, respectively. Thus, the conformal boundary
$\mathscr{I}$ is equivalently represented by the hypersurfaces with
$\tau=0$ or $\bar{\tau}=0$. As $\bmg$ and $\bar{\bmg}$ are conformally
related one can write
\[ 
\bar{\bmg}= \omega^2 \bmg \hspace{0.5cm} \text{with} \hspace{0.5cm} 
\omega \equiv \bar{\Theta} \Theta^{-1}.
\]
Using the relations in \eqref{TransfTheta} and \eqref{FracTrans} we
obtain, after a calculation, that 
\begin{equation}
\label{OmegaAnsatz}
 \omega(\tau)=\frac{a\left(1+ \displaystyle \frac{a\bar{\kappa}\tau}{2(c\tau+d)}\right)}{
   \left((c\tau + d)\left(1+\displaystyle\frac{1}{2}\kappa \tau\right)\right)}.
\end{equation}

The conformal transformation law
for the field $s$ can be seen to be given by
 \[ \bar{s}= \omega^{-1}s +
 \omega^{-2}\nabla_{\bmc}\omega\nabla^{\bmc}\Theta + \frac{1}{2}\omega^{-3}
\Theta\nabla_{\bmc}\omega\nabla^{\bmc}\omega.
\]
As discussed in Section \ref{VanishingFields}, in the analysis of the
extremal Schwarzschild-de Sitter spacetime one can assume that
$\partial_{\psi}\kappa=0$ and $\bmf=0$. Now, Propositions
\ref{propSconstantAlongGeodesics} and \ref{propKappaConstant} imply
that $s= \sqrt{|\lambda|/3}\kappa$ and $\bar{s}=
\sqrt{|\lambda|/3}\bar{\kappa}$ are constant.  Exploiting this
observation, the transformation law for $s$ can be read as an equation
for $\omega$ ---namely
\begin{equation}
\label{DifferentialEquationForOmega}
\Theta\dot{\omega}^ 2 +2\omega \dot{\Theta}\dot{\omega} + \omega^2s
-\omega^3\bar{s}=0 .
 \end{equation}
Substituting expression \eqref{OmegaAnsatz} into equation
\eqref{DifferentialEquationForOmega} one gets the condition 
\begin{equation}
\label{relationFracTrans}
2c + a \bar{\kappa} - d\kappa =0.
\end{equation}
 One can read equation \eqref{relationFracTrans} as the transformation
law for $\bar{\kappa}$ so that
\[ 
\bar{\kappa} = \frac{d\kappa-2c}{a}. 
\]
In order to have a meaningful transformation law between $\bar{\tau}$
and $\tau$, neither $a$ nor $d$ can vanish. Substituting equation
\eqref{relationFracTrans} into the reparametrisation formula
\eqref{FracTrans} and expression
\eqref{OmegaAnsatz} one can observe that $a/d$ actually corresponds to
$\omega(0)\equiv\omega_{\star}$. Therefore, one has that
\begin{equation}
\label{TransformationLawTauAndOmega}
\bar{\tau}(\tau)=\frac{2\omega_{\star}\tau}
{(\omega_{\star}\bar{\kappa}-\kappa)\tau-2}, \qquad \qquad
\omega(\tau)=\frac{4
  \omega_{\star}}{\big((\omega_{\star}\bar{\kappa}-\kappa)\tau-2\big)^2}.
\end{equation}

From the last expression one can identify
$\dot{\omega}_{\star}\equiv
\dot{\omega}(0)=\omega_{\star}(\omega_{\star}\bar{\kappa}-\kappa)$. 
In addition, notice
that $\bar{\tau} \rightarrow \infty$ and $\omega \rightarrow \infty$
as $\tau \rightarrow 2/(\omega_{\star}\bar{\kappa}-\kappa)$. Therefore,
the hypersurface defined by
$\tau=2/(\omega_{\star}\bar{\kappa}-\kappa)$ is at an infinite
distance from the conformal boundary as measured with respect to the
$\bar{\bmg}$-proper time.

\begin{remark}
{\em An alternative approach to deduce equations
\eqref{relationFracTrans} and \eqref{TransformationLawTauAndOmega} is
to write $\bar{\Theta}(\bar{\tau}(\tau))=\omega(\tau)\Theta(\tau)$ and use
equations \eqref{TransfTheta} and \eqref{FracTrans} to identify 
$\kappa$ and $\omega$.}
\end{remark}

\section{Appendix: Cartan's structure equations and space spinor formalism}
\label{Appendix:CartanStrucureEquations}

In this appendix we give a brief discussion of Cartan's structure
equations and the space spinor formalism.

\subsection{Cartan's structure equations in frame formalism}
\label{CartanFrame}

Consider a $\bmh$-orthonormal frame $\{ \bme_{\bmi}\}$ with
corresponding coframe $\{\bmomega^\bmi\}$. By construction, one has
$\langle \bmomega^{\bmi},\bme_{\bmj}\rangle = \delta_{\bmi}{}^{\bmj}$.
The connection coefficients of the Levi-Civita connection $\bm D$ of
$\bmh$ respect to this frame are defined as
\[
\langle \bmomega^\bmj , D_\bmi\bme_\bmk \rangle \equiv
\gamma_{\bmi}{}^{\bmj}{}_{\bmk}.
\]
As a consequence of the metricity of $\bm D$ it follows that
$\gamma_{\bmi\bmj\bmk}=-\gamma_{\bmi\bmk\bmj}$. The connection form is
accordingly defined as
 \[
\bmgamma^{\bmj}{}_{\bmk} \equiv\gamma_{\bmi}{}^{\bmj}{}_{\bmk} \wedge
\bmomega^\bmi.
\]
 With these definitions, the first and second Cartan's structure
 equations are, respectively, given by
\begin{subequations}
\begin{eqnarray}
&& \mathbf{d} \bmomega^\bmi = -\bmgamma^\bmi{}_{\bmj}\wedge
  \bmomega^\bmj, \label{FirstCartanStructureEquationFrame}\\ &&
  \mathbf{d}\gamma^{\bmi}{}_{\bmj}= -\bmgamma^{\bmi}{}_{\bmk}\wedge
  \bmgamma^{\bmk}{}_{\bmj} +
  \bm\Omega^{\bmi}{}_{\bmj},\label{SecondCartanStructureEquationFrame}
\end{eqnarray}
\end{subequations}
where $\bm\Omega^{\bmi}{}_{\bmj}$ is the curvature 2-form defined as
\[ 
\bm\Omega^{\bmi}{}_{\bmj} \equiv R^{\bmi}{}_{\bmj \bmk \bml}
\bmomega^{\bmk} \wedge\bmomega^{\bml}.
\]

\subsection{Basic spinors}
\label{Appendix:SpaceSpinorFormalism}

In the space spinor formalism, given a spin basis
$\{\epsilon_{\bmA}{}^{A}\}$ where $_{\bmA = \bm0,\bm1}$, any of the
spinorial fields appearing in the extended conformal Einstein field
equations can be decomposed in terms of basic irreducible spinors. The
basic valence-2 symmetric spinors are:
\begin{equation}
x_{\bmA\bmB} \equiv \sqrt{2}\epsilon_{(\bmA}{}^{\bm0}\epsilon_{\bmB)}{}^{\bm1}, 
\qquad  y_{\bmA\bmB}\equiv
-\frac{1}{\sqrt{2}}\epsilon_{(\bmA}{}^{\bm1}\epsilon_{\bmB)}{}^{\bm1},
 \qquad  z_{\bmA\bmB}\equiv
 \frac{1}{\sqrt{2}}\epsilon_{(\bmA}{}^{\bm0}\epsilon_{\bmB)}{}^{\bm0}.  \label{compsBasicxyz}
\end{equation}
The basic valence 4 spinors are given by
\begin{subequations}
\begin{eqnarray}
&\epsilon_{\bmA\bmC}x_{\bmB\bmD}+\epsilon_{\bmB\bmD}x_{\bmA\bmC},\qquad
  \epsilon_{\bmA\bmC}y_{\bmB\bmD}+\epsilon_{\bmB\bmD}y_{\bmA\bmC},
  \qquad
  \epsilon_{\bmA\bmC}z_{\bmB\bmD}+\epsilon_{\bmB\bmD}z_{\bmA\bmC},& \label{Valence4Basic1}
  \\
 &h_{\bmA\bmB\bmC\bmD}\equiv-\epsilon_{\bmA(\bmC}\epsilon_{\bmD)\bmB}
  ,
\qquad  \epsilon^{\bmi}{}_{\bmA\bmB\bmC\bmD}=
\epsilon_{(\bmA}{}^{(\bm E}\epsilon_{\bmB}{}^{\bm
  F}\epsilon_{\bmC}{}^{\bm G}\epsilon_{\bmD)}{}^{\bm H)_{\bmi}}. &\label{Valence4Basic2}
\end{eqnarray}
\end{subequations}
In the last expression ${}^{(\bmA\bmB\bmC\bmD)_\bmi}$ indicates
that an $\bmi$ number of indices are set equal to $\bm1$ after
symmetrisation.  Any valence 4 spinor $\zeta_{ABCD}$ with the
symmetries $\zeta_{(AB)(CD)}$ can be expanded in terms of these basic
spinors. One has the identities
\begin{subequations}
\begin{eqnarray}
& x_{(\bmA\bmB}x_{\bmC\bmD)}=2\epsilon^{2}{}_{\bmA\bmB\bmC\bmD}, \hspace{3mm}
  y_{(\bmA\bmB}y_{\bmC\bmD)}=\displaystyle\frac{1}{2}\epsilon^{4}{}_{\bmA\bmB\bmC\bmD}, \hspace{3mm}
  z_{(\bmA\bmB}z_{\bmC\bmD)}=\displaystyle\frac{1}{2}\epsilon^{0}{}_{\bmA\bmB\bmC\bmD},&  \label{xyzProducts1}
  \\ 
& y_{\bmA\bmB}x_{\bmC\bmD}=-\epsilon^{3}{}_{\bmA\bmB\bmC\bmD} -
  \displaystyle\frac{1}{2\sqrt{2}}(\epsilon_{\bmA\bmC}y_{\bmB\bmD}+\epsilon_{\bmB\bmD}y_{\bmA\bmC}),&
  \label{xyzProducts2}\\ 
& z_{\bmA\bmB}x_{\bmC\bmD}=\epsilon^{1}{}_{\bmA\bmB\bmC\bmD} +\displaystyle\frac{1}{2\sqrt{2}}
  (\epsilon_{\bmA\bmC}z_{\bmB\bmD}+\epsilon_{\bmB\bmD}z_{\bmA\bmC}),&\label{xyzProducts3}\\ 
& y_{\bmA\bmB}z_{\bmC\bmD}=-\displaystyle\frac{1}{2}\epsilon^{2}{}_{\bmA\bmB\bmC\bmD}+
  \displaystyle\frac{1}{4\sqrt{2}} (\epsilon_{\bmA\bmC}x_{\bmB\bmD}+
  \epsilon_{\bmB\bmD}x_{\bmA\bmC})  -\displaystyle\frac{1}{6}h_{\bmA\bmB\bmC\bmD}.&\label{xyzProducts4}
\end{eqnarray}
\end{subequations}
Another set of identities used in the main text is given by
\begin{subequations} 
 \begin{eqnarray} 
& x_{\bmA\bmB}x^{\bmA\bmB}=1, \hspace{0.33cm}
   x_{\bmA\bmB}y^{\bmA\bmB}=0, \hspace{0.33cm}
   x_{\bmA\bmB}z^{\bmA\bmB}=0, \hspace{0.33cm}
   z_{\bmA\bmB}z^{\bmA\bmB}=0, \hspace{0.33cm}
   y_{\bmA\bmB}z^{\bmA\bmB}=- \displaystyle\frac{1}{2},&
   \label{UsefulIdentities1}\\ 
& x_{\bmA}{}^{\bmQ}x_{\bmB\bmQ}=
    \displaystyle\frac{1}{2}\epsilon_{\bmA\bmB}, \hspace{0.33cm}
   y_{\bmA}{}^{\bmQ}x_{\bmB\bmQ} =
    \displaystyle\frac{1}{\sqrt{2}}y_{\bmA\bmB}, \hspace{0.33cm}
   z_{\bmA}{}^{\bmQ}x_{\bmB\bmQ} =
   - \displaystyle\frac{1}{\sqrt{2}}z_{\bmA\bmB}, \hspace{0.33cm}y_{\bmA}{}^{\bmQ}y_{\bmB\bmQ}=0,&
 \label{UsefulIdentities2}
\\ &    y_{\bmA}{}^{\bmQ}z_{\bmB\bmQ}= - \displaystyle\frac{1}{2\sqrt{2}}x_{\bmA\bmB} +
    \displaystyle\frac{1}{4}\epsilon_{\bmA\bmB}, \hspace{0.33cm}
   z_{\bmA}{}^{\bmQ}z_{\bmB\bmQ}=0, & \label{UsefulIdentities3} 
\\ & \epsilon^{2}_{\bmA \bmB \bmC \bmD}x^{\bmC \bmD}=-\displaystyle\frac{1}{3}x_{\bmA \bmB}, \hspace{0.55cm}
 \epsilon^{2}_{\bmA \bmB \bmC \bmD}y^{\bmC \bmD}=\displaystyle\frac{1}{6}y_{\bmA \bmB}, \hspace{0.55cm}
\epsilon^{2}_{\bmA \bmB \bmC \bmD}z^{\bmC \bmD}=\displaystyle\frac{1}{6}z_{\bmA \bmB}. \label{UsefulIdentities4}
\end{eqnarray}
\end{subequations}

 These identities and a more exhaustive list has been given
in \cite{FriKan00}.

\subsection{Cartan's structure equations in spinor form}
\label{CartanSpacespinor}

Let $\tau^{\bmA \bmA'}$ denote a Hermitian spinor $\tau^{\bmA \bmA'}$
with normalisation $\tau^{\bmA
  \bmA'}\tau_{\bmA \bmA'}=2$. Consider an adapted spin dyad
$\{\epsilon_{\bmA}{}^{A}\}$ such that the matrix representation of
$\tau^{\bmA \bmA'}$ is given by the identity $2\times 2$ matrix. The spatial Infeld-van
de Waerden symbols are related to the usual Infeld-van der
Waerden via
\begin{equation} \label{defSpatialInfeld}
 \sigma_{\bmA \bmB}{}^{\bmi} \equiv
 \tau_{(\bmB}{}^{\bmB'}\sigma_{\bmA) \bmB'}{}^{\bmi}.
 \end{equation}
Equivalently, one has 
\[
\sigma^{\bmA \bmB}{}_{\bmi}= -\tau^{(\bmB}{}_{\bmB'}\sigma^{\bmA)
  \bmB'}{}_{\bmi}. 
\]
The matrix representation of the spatial Infeld-van der Waerden
symbols is given by
\begin{eqnarray*}
&& \sigma_{\bmA \bmB}{}^{\bm1} \equiv \frac{1}{\sqrt{2}} \begin{pmatrix}
  -1 & 0 \\ 0 & 1 \\ \end{pmatrix}, \hspace{1.5cm} \sigma_{\bmA
  \bmB}{}^{\bm2} \equiv \frac{1}{\sqrt{2}} \begin{pmatrix} \mbox{i} &
  0 \\ 0 & \mbox{i} \\ \end{pmatrix} , \hspace{1.5cm} \sigma_{\bmA
  \bmB}{}^{\bm3} \equiv \frac{1}{\sqrt{2}} \begin{pmatrix} 0 & 1 \\ 1
  & 0 \\ \end{pmatrix} , \\
&& \sigma^{\bmA \bmB}{}_{\bm1} \equiv \frac{1}{\sqrt{2}} \begin{pmatrix}
  -1 & 0 \\ 0 & 1 \\ \end{pmatrix}, \hspace{1.5cm} \sigma^{\bmA
  \bmB}{}_{\bm2} \equiv \frac{1}{\sqrt{2}} \begin{pmatrix} \mbox{-i} &
  0 \\ 0 & \mbox{i} \\ \end{pmatrix} , \hspace{1.5cm} \sigma^{\bmA
  \bmB}{}_{\bm3} \equiv \frac{1}{\sqrt{2}} \begin{pmatrix} 0 & 1 \\ 1
  & 0 \\ \end{pmatrix} .
\end{eqnarray*}
 Thus, the space spinor counterpart of coframe and connection
coefficients can be obtained succinctly by contraction with the
spatial
Infeld-van der Waerden symbols as $\bmomega^{\bmA \bmB} \equiv
\bmomega^{\bmi}\sigma_{\bmi}{}^{\bmA \bmB}$ and $\gamma_{\bmA
  \bmB}{}^{\bmC \bmD}{}_{\bmE \bmF} = \gamma_{\bmi}{}^{\bmj}{}_{\bmk}
\sigma^{\bmi}{}_{\bmA \bmB}\sigma_{\bmj}{}^{\bmC \bmD}
\sigma^{\bmk}{}_{\bmE \bmF}$.  With these definitions the spinorial
version of the  Cartan
structure equations is given by
\begin{subequations}
\begin{eqnarray} 
\label{FirstCartanStructureEquationSpinor}
&&\mathbf{d}\bmomega^{\bmA \bmB}= -\bmgamma^{\bmA}{}_{\bmB} \wedge
\bmomega^{\bmB \bmE} - \gamma^{\bmB}{}_{\bmE}\wedge \bmomega^{\bmA
  \bmE},\\
&&\label{SecondCartanStructureEquationSpinor}
 \mathbf{d}\bmgamma^{\bmA}{}_{\bmB}=-\bmgamma^{\bmA}{}_{\bmE}\wedge
 \bmgamma^{\bmE}{}_{\bmB} + \Omega^{\bmA}{}_{\bmB},
\end{eqnarray}
\end{subequations}
 where 
\[
\bmgamma^{\bmA}{}_{\bmB} \equiv
 \frac{1}{2}\gamma_{\bmC \bmD}{}^{\bmA \bmQ}{}_{\bmB \bmQ}
 \bmomega^{\bmC \bmD},
\]
 and $\Omega^{\bmA}{}_{\bmB}$ is the spinor
version of the curvature 2-form, with
 \[
\Omega^{\bmA}{}_{\bmB}\equiv \frac{1}{2}r^{\bmA}{}_{\bmB \bmC \bmD
   \bmE \bmF}\omega^{\bmC \bmD}\wedge \bmomega^{\bmE \bmF}.
\]
 In the last expression the spinor $r_{\bmA \bmB \bmC \bmD
  \bmE \bmF}$ can be decomposed as
\[ 
r_{\bmA \bmB \bmC \bmD \bmE \bmF}= \left(\frac{1}{2}  s_{\bmA \bmB \bmC \bmD}
 - \frac{1}{12}rh_{\bmA\bmB \bmC \bmE}\right) \epsilon_{\bmD \bmF} +
 \left(\frac{1}{2}s_{\bmA \bmB \bmD \bmF}-\frac{1}{12}rh_{\bmA \bmB \bmD
   \bmF}\right) \epsilon_{\bmC \bmE}
\]
 where $s_{\bmA \bmB \bmC \bmD}$ and
 $r$ correspond to the space spinor version of the trace-free part of
 the Ricci tensor and Ricci scalar of $\bmh$, respectively.

\medskip
To relate the previous discussion with the basic spinors
 $x_{\bmA \bmB}$, $y_{\bmA \bmB}$ and $z_{\bmA
  \bmB}$,  observe that using \eqref{compsBasicxyz} and
\eqref{defSpatialInfeld} one obtains that
\begin{subequations}
\begin{eqnarray} 
&\sigma_{\bmA \bmB}{}^{\bm1} = -z_{\bmA \bmB}-y_{\bmA
    \bmB}, \qquad  \sigma_{\bmA \bmB}{}^{\bm2} =
  \mbox{i}(z_{\bmA \bmB}-y_{\bmA \bmB}) , \qquad \sigma_{\bmA
    \bmB}{}^{\bm3} = x_{\bmA \bmB}, & \label{Infeldxyz1}
  \\ 
&\sigma^{\bmA \bmB}{}_{\bm1} = z^{\bmA \bmB} + y^{\bmA
    \bmB}, \qquad \sigma^{\bmA \bmB}{}_{\bm2} =
  \mbox{i}(-z_{\bmA \bmB} + y_{\bmA \bmB}) , \qquad
  \sigma^{\bmA \bmB}{}_{\bm3} = -x^{\bmA \bmB}.\label{Infeldxyz2} &
\end{eqnarray}
\end{subequations}

\section{Appendix: The frame conformal Einstein field equations}
\label{Appendix:CFE}

The tensorial (frame) version of the standard vacuum conformal
Einstein field equations are given by the following system ---see
e.g. \cite{Fri81a,Fri81b,Fri82,Fri83}:
\begin{subequations}
\begin{eqnarray}
&& \Sigma_{\bma}{}^{\bmc}{}_{\bmb}e_{\bmc}=0 , \label{CEFEnoTorsion}
  \\ && \nabla_{\bme}d{}^{\bme}{}_{\bma\bmb\bmf}{}=0
  , \label{CEFErescaledWeyl} \\ &&
  \nabla_{\bmc}L_{\bmd\bmb}-\nabla_{\bmd}L_{\bmb\bmc}
  -\nabla_{\bma}\Xi d{}^{\bma}{}_{\bmb\bmc\bmd} =0
  , \label{CEFESchouten}\\ && \nabla_{\bma}\nabla_{\bmb}\Xi +\Xi
  L_{\bma\bmb} - s g_{\bma\bmb}=0 ,
 \label{CEFEsecondderivativeCF}\\
&& \nabla_{\bma}s +L_{\bma\bmc} \nabla ^{\bmc}\Xi=0
 , \label{CEFEs}\\ && R^{\bmc}{}_{\bma\bmb\bmd} -
 \rho^{\bmc}{}_{\bma\bmb\bmd}=0, \label{CEFERiemann} \\
&& 6\Xi s -3\nabla_{\bma}\Xi \nabla^{\bma}\Xi = \lambda, \label{ConstantLambda}
\end{eqnarray}
\end{subequations}
where $\Sigma_{\bma}{}^{\bmc}{}_{\bmb}$ is the torsion tensor, given
in terms of the connection coefficients, as
\[
\Sigma_{\bma}{}^{\bmc}{}_{\bmb}\bme_{\bmc} \equiv
      [\bme_{\bma},\bme_{\bmb}]-
      (\Gamma_{\bma}{}^{\bmc}{}_{\bmb}-\Gamma_{\bmb}{}^{\bmc}{}_{\bma})\bme_{\bmc};
\]
$L_{\bma \bmb}$ is the Schouten tensor; $\Xi$ is the conformal factor and
$s$ is a concomitant of the conformal factor defined by
\[
s\equiv \frac{1}{4}\nabla_{\bma}\nabla^{\bma}\Xi + \frac{1}{24}R\Xi.
\]
In addition, $\rho^{\bma}{}_{\bmb \bmc \bmd}$ is the algebraic
curvature and $R^{\bmc}{}_{\bmd\bma\bmb}$ is the geometric curvature.
\begin{eqnarray*}
&& \rho^{\bma}{}_{\bmb\bmc\bmd} \equiv\Xi d^{\bma}{}_{\bmb\bmc\bmd} +
  2( g^{\bma}{}_{[\bmc}L_{\bmd]\bmb} - g_{\bmb[
      \bmc}L_{\bmd]}{}^{\bma} ), \\ 
&& R^{\bmc}{}_{\bmd \bma
    \bmb}\equiv \bme_\bma(\Gamma_{\bmb}{}^{\bmc}{}_{\bmd})-\bme_{\bmb}
(\Gamma_{\bma}{}^{\bmc}{}_{\bmd})+\Gamma_{\bmf}{}^{\bmc}{}_{\bmd}
(\Gamma_{\bmb}{}^{\bmf}{}_{\bma}-\Gamma_{\bma}{}^{\bmf}{}_{\bmb}) \\
 && 
\qquad \qquad \qquad 
  \qquad
  +\Gamma_{\bmb}{}^{\bmf}{}_{\bmd}\Gamma_{\bma}{}^{\bmc}{}_{\bmf}-
\Gamma_{\bma}{}^{\bmf}{}_{\bmd}\Gamma_{\bmb}{}^{\bmc}{}_{\bmf}-
\Sigma_{\bma}{}^{\bmf}{}_{\bmb}\Gamma_{\bmf}{}^{\bmc}{}_{\bmd}.
\end{eqnarray*}



\end{document}